\documentclass[pre,aps,10pt,superscriptaddress,notitlepage,nofootinbib,floatfix]{revtex4}  
\usepackage{mathrsfs} \usepackage{graphicx,color} \usepackage{bbm}
\usepackage{hyperref}

\usepackage{epsfig}
\usepackage{epstopdf}
\usepackage{bm}
\usepackage{anysize}
\usepackage[geometry]{ifsym}
\usepackage{enumitem}
\usepackage{mathtools}

\newcommand{\eqdef}{=\vcentcolon}
\usepackage{amsthm}

\usepackage{amssymb}
\usepackage{amsfonts}

\usepackage{amsmath}
\graphicspath{{./figures/}}

\newtheorem{theorem}{Theorem}

\let\l=\lambda

\let\D=\Delta   
   
 \let\r=\rho

\def\de{\mathrm d}

\def\to{\rightarrow}

\newcommand{\beq}{\begin{equation}} \newcommand{\eeq}{\end{equation}}
 
\newcommand{\Tr}{\text{Tr}}

\marginsize{2cm}{2cm}{2.5cm}{2.5cm}     % {left}{right}{top}{bottom}

\begin{document}

\title{
Approximate Survey Propagation for Statistical Inference 
%Approximate message passing and its replica symmetry breaking generalizations
}

\author{Fabrizio Antenucci}
\affiliation{
Institut de physique th\'eorique, Universit\'e Paris Saclay, CNRS, CEA, F-91191 Gif-sur-Yvette, France
}
\affiliation{Soft and Living Matter Lab., Rome Unit of CNR-NANOTEC, Institute of
Nanotechnology, Piazzale Aldo Moro 5, I-00185, Rome, Italy}

\author{Florent Krzakala}
\affiliation{Laboratoire de Physique Statistique, CNRS \& Sorbonnes Universit\'es \& \'Ecole Normale Supérieure, PSL University, Paris, France.}

\author{Pierfrancesco Urbani}
\affiliation{
Institut de physique th\'eorique, Universit\'e Paris Saclay, CNRS, CEA, F-91191 Gif-sur-Yvette, France
}

\author{Lenka Zdeborov\'a}
\affiliation{
Institut de physique th\'eorique, Universit\'e Paris Saclay, CNRS, CEA, F-91191 Gif-sur-Yvette, France
}

\begin{abstract}
  Approximate message passing algorithm enjoyed considerable attention
  in the last decade. In this paper we introduce a variant of the AMP
  algorithm that takes into account glassy nature of the system under
  consideration. We coin this algorithm as the approximate survey
  propagation (ASP) and derive it for a class of low-rank matrix
  estimation problems.
  We derive the state evolution for the ASP algorithm and prove that it
  reproduces the one-step replica symmetry breaking (1RSB) fixed-point equations,
  well-known in physics of disordered systems.
  Our derivation thus gives a concrete algorithmic
  meaning to the 1RSB equations that is of independent interest.
  We characterize the performance of ASP in terms of convergence and
  mean-squared error as a function of the free Parisi parameter $s$.
  We conclude that
%
%In a replica symmetric situation, and in particular in the Bayes optimal setting, the ASP algorithm converges to the same fixed-point 
%as Approximate Message Passing (AMP) algorithm for any value of $s$.
%
  when there is a model mismatch between the true generative model and
  the inference model, the performance of AMP rapidly degrades both in
  terms of MSE and of convergence, while ASP converges in a larger
  regime and can reach lower errors.
%
%In particular, AMP stops convergencing as soon as the replica symmetry
%gets broken.
%
%The definition of ASP encompasses a replica symmetry breaking structure for the inferred variables, providing an algorithm that can converge also in a RSB region, 
%when the mismatch with the generative model is substantial.
%
%We observed that for the low-rank matrix estimation problem, in any
%point of the phase diagram the ASP algorithm converges point-wise for
%at least some values of $s$.
%
%This property overcomes a common problem of AMP, as in many practical tasks it is not possible 
%to reach exact knowledge of the generative process.
% 
  Among other results, our analysis leads us to a striking hypothesis
  that whenever $s$ (or other parameters) can be set in such a way
  that the Nishimori condition $M=Q>0$ is restored, then the
  corresponding algorithm is able to reach mean-squared error as low
  as the Bayes-optimal error obtained when the model and its
  parameters are known and exactly matched in the inference procedure.
%
%In such situations we observe that ASP converges point-wise to a fixed
%point with minimal MSE even if the generative model and the inference
%model are highly mismatched.
%
%An algorithmic procedure able to select values of the Parisi parameter
%$s$ that lead to low errors remains an open problem.
% 
%In this paper we remark that the standard methods based on expectation maximization or choice of the equilibrium value of $s^*$ are
%sub-optimal. 
%
%In particular, the point-wise convergence of $s^*$ is not guaranteed.
\end{abstract}

\maketitle

\tableofcontents

\newpage

\section{Introduction}

\subsection{General Motivation}

Many problems in data analysis and other data-related science can be formulated as optimization or
inference in high-dimensional and non-convex landscapes. General
classes of algorithms that are able to deal with some of these
problems include Monte Carlo and Gibbs-based sampling \cite{neal2000markov}, variational
mean field methods \cite{wainwright2008graphical}, stochastic gradient
descent \cite{bottou2010large} and their variants. Approximate message passing (AMP) algorithms \cite{ThoulessAnderson77,DonohoMaleki09} is another such class
that has one very remarkable advantage over all the before mentioned:
for a large class of models where instances are
generated from a probabilistic model, the performance of AMP on large such
instances can be tracked via a so-called {\it state evolution} \cite{bolthausen2014iterative,bayati2011dynamics,javanmard2013state}. 
This development has very close connections to statistical physics of
disordered systems because the state evolution that describes AMP
coincides with fixed-point equations that arise from the replica and
cavity methods as known in the theory of glasses and spin glasses
\cite{MPV87}. 

AMP and its state evolution have been so far mainly used in two
contexts. On the one hand, for optimization in cases where the
associated landscape is {\it convex}. This the the case e.g.~in the original
work of Donoho, Maleki, Montanari \cite{DonohoMaleki09} where
$\ell_1$-norm minimization for sparse
linear regression is analyzed, or in the study of the so-called
M-estimators \cite{donoho2016high,advani2016statistical}. On the other
hand, in the setting of Bayes-optimal inference where the model that
generated the data is assumed to be known perfectly, see
e.g.~\cite{zdeborova2015statistical,LesKrzZde17}, 
where the so-called Nishimori conditions ensure that the associated
posterior probability measure is of so-called replica symmetric kind in the context
of spin glasses \cite{MPV87,NishimoriBook}. Many (if not most) of
inference and optimization problems that are solved in the currently most
challenging applications are highly non-convex and the underlying model
that generated the data is not known. It is hence an important general
research direction to
understand the behavior of algorithms and find their more robust generalizations
encompassing such settings. 

In the present paper we make a step in
this general direction for the class of AMP algorithms. We analyze in
detail the phase diagram and phases of convergence of the AMP
algorithm on a prototypical example of a problem that is non-convex
and not in the Bayes-optimal setting. The example we choose is
rank-one matrix estimation problem that has the same phase diagram as
the Sherrington-Kirkpatrick (SK) model with ferromagnetic bias~\cite{SherringtonKirkpatrick75}. AMP reproduces the corresponding 
replica symmetric phase diagram with region of non-convergence being
given by the instability towards \emph{replica symmetry breaking} (RSB).  We note
that while
this phase diagram is well known in the literature on spin glasses, its algorithmic consequences are
obscured in the early literature by the fact that unless the AMP algorithm
is iterated with the correct time-indices
\cite{bolthausen2014iterative} convergence is not
reached, see discussion e.g.~in \cite{zdeborova2015statistical}.

Our main contribution is the development 
of a new type of approximate message passing algorithm that takes into
account breaking of the replica symmetry and reduces to AMP for a
special choice of parameters. We call this the {\it
  approximate survey propagation} (ASP) algorithm, following up on the
influential work on survey propagation in sparse random constraint
satisfaction problems \cite{MPZ02,BMZ05}. We show that there are
regions (away from the Nishimori line) in which ASP converges while AMP does not, and where at
the same time ASP provides lower estimation error than AMP. We show
that the state evolution of ASP leads to the one-step-replica symmetry
breaking (1RSB)  fixed-point equations well known from the study of spin
glasses. This is the first algorithm that provably converges
towards fixed-points of the 1RSB equations.  Again we stress that, while the 1RSB phase
diagram and corresponding physics of the ferromagnetically biased SK
model is well-known, its algorithmic confirmation is new to our knowledge: even if the 1RSB
versions of the Thouless-Anderson-Palmer (TAP)
\cite{ThoulessAnderson77} equations was previously discussed, e.g.~in \cite{MPV87,yasuda2012replica}, the corresponding 
time-indices in ASP are crucial in order to reproduce this phase
diagram algorithmically. Our work gives a concrete
algorithmic meaning to the 1RSB fixed-point equations, and can thus be
potentially used to
understand this concept independently of the heuristic replica or cavity methods. 
%We also discuss the
%version of the ASP algorithm taking into account full-step-replica symmetry breaking. 

\subsection{The low-rank estimation model and Bayesian setting}

As this is a kind of exploratory paper we focus on the problem of
rank-one matrix estimation. This problem
is (among) the simplest where the ASP algorithm can be tested. In particular, for binary (Ising) variables it is equivalent to the Sherrington-Kirkpatrick model with
ferromagnetically biased couplings as studied e.g.~in
\cite{SherringtonKirkpatrick75,de1978stability,NishimoriBook}.  
%
%
% Push the low-rank estimation on a second level, it is only the
% application for a more general framework here. 
Low-rank matrix estimation is a problem ubiquitous in modern data
processing. Commonly it is stated in the following way: Given an observed matrix $Y$ one aims to write it as
$Y_{ij} =   \lambda u^{\top}_i v_j + \xi_{ij} $, where $u_i, v_j\in
{\mathbb R}^r$ with $r$ being the rank, and $\xi_{ij}$ 
is a noise or a part of the data matrix $Y$ that is not well explained
via this decomposition. Commonly used methods such as principal
component analysis or clustering can be formulated as low-rank matrix
estimation. Applications of these methods range from medicine over to 
signal processing or marketing. From an algorithmic point of view and
under various constraints on the factors $u_i, v_j$ and the noise
$\xi_{ij}$ the problem is highly non-trivial (non-convex and NP-hard in
worst case). Development of algorithms for low-rank matrix estimation
and their theoretical properties is an object of numerous studies in
statistics, machine learning, signal processing, applied mathematics,
etc. \cite{wright2009robust,candes2011robust,candes2009exact}. 
In this
paper we focus on the symmetric low-rank matrix estimation where
the matrix $Y_{ij}$ is symmetric, i.e. $Y_{ij}=Y_{ji}$ and the desired
decomposition is $Y_{ij} =   \lambda x^{\top}_i x_j + \xi_{ij} $, where $x_i\in
{\mathbb R}^r$.

A model for low-rank matrix estimation that can be solved exactly in the
large size limit using the techniques studied in this paper is the
so-called general-output low-rank matrix estimation model \cite{lesieur2015mmse,LesKrzZde17} where the matrix $Y_{ij}\in {\mathbb R}^{N
  \times N}$ is generated from a given probability distribution $P_{\rm
out}^{(0)}(Y_{ij}|w_{ij}^{(0)})$ with 
\beq
w_{ij}^{(0)}  \equiv \frac{\left(x_i^{(0)}\right)^T x^{(0)}_j}{\sqrt{N}}\, ,
\label{eq:Y}
\eeq 
and where each component $x_i^{(0)}\in {\mathbb
  R}^{r_0}$ is generated independently from a probability distribution $P_{0}(x_i^{(0)})$. 
The $N\times r$ matrix $X^{(0)}$ can then be seen as an unknown signal
that is to be reconstructed from the observation of the matrix
$Y_{ij}$.  

Following Bayesian approach, a way to infer $X^{(0)}$ given $Y$ is to analyze the posterior probability distribution
\begin{align}
 P(X|Y) = \frac{1}{Z_X(Y)} \prod_{1 \leq i \leq N} P_X(x_i)  \;
  \prod_{1 \leq i < j \leq N}  \exp\left[ g \left( Y_{ij} |
  \frac{x^\top_i x_j}{\sqrt{N}} \right) \right]\, ,
\label{eq:posterior}
\end{align}
where $P_X(x_i)$ is the assumed prior on signal components $x_i\in
{\mathbb R}^r$ and $P_{\rm out} \left( Y_{ij} | w_{ij} \right) =
\exp\left[ g \left( Y_{ij} | w_{ij} \right) \right]$ is the assumed model
for the noise. % the output channel.
The normalization $Z_X(Y)$ is the partition function in the
statistical physics notation. And the posterior distribution
(\ref{eq:posterior}) is nothing else but the Boltzmann measure on the
$r$-dimensional spin variables $x_i$. 

% Put the way we compute the estimator and errors here ...... 

In a Bayes-optimal estimation, we know exactly the model that generated the data, i.e. 
\beq
P_x=P_0,  \quad 
P_{\rm out}  = P_{\rm out}^{(0)}, \quad  r=r_0\, . \label{Bayes_optimal}
\eeq
 The Bayes-optimal
estimator is defined as the one that among all the estimators minimizes the mean-squared error to the ground
truth signal $x_i^{(0)}$. This is computed as the mean of the posterior
marginals or, in physics language, the local magnetizations. 

In this paper we will focus on
inference with the mismatching model where at least one of the above
equalities (\ref{Bayes_optimal}) does not hold. Yet, we will still aim to evaluate the
estimator that computes the mean of the posterior marginals 
\beq
          \hat x_i = \int x_i P(X|Y) \prod_{j=1}^N {\rm d} x_j \, .
\eeq
We will call this the marginalization estimator. 

Another common estimator in inference is the maximum posterior
probability (MAP)
estimator where 
\beq
       \hat X^{\text{MAP}} = \underset{X}{\text{argmax}} \; P(X|Y) \, .
\label{eq:MAP_estimator}       
\eeq
Generically, optimizing a complicated cost function is simpler than
marginalizing over it, because optimization is in general simpler than
counting. Thus MAP is often thought of as a simpler estimator to
evaluate. Moreover, in statistics the problems are usually set in a regime
where the MAP and the marginalization estimator coincide. However,
this is not the case in the setting considered in the present paper
and we will comment on the difference in the subsequent sections.

In our analysis we are interested in the \textit{high-dimensional statistics} limit, $N \to \infty$ whereas $r = \mathcal{O}(1)$. 
The factor $1/\sqrt{N}$ in Eq.~(\ref{eq:Y}) in this limit ensures that
the estimation error we are able to obtain is in a regime where it
goes from zero to randomly bad as the parameters vary.
This case is different from traditional statistics, where one is
typically concerned with estimation error going to zero as $N$ grows. 
%with a small $N$ and with a noise comparable to the signal.
%Here we have huge data but also a very weak signal.

In the $N \to \infty$ limit, the above defined model benefits of an
universality property in the noise channel
\cite{lesieur2015mmse,LesKrzZde17} (see \cite{krzakala2016mutual} for
a rigorous proof) as the estimation error depends on the function $g$
only through the matrices
\begin{align}
& S_{ij}=\left.\frac{\partial g(Y_{ij}|w)}{\partial w}\right|_{w=0}
\, , &
& \hat R_{ij}=-\left.\frac{\partial^2 g(Y_{ij}|w)}{\partial w^2}\right|_{w=0} \, .
\end{align}
%\begin{align}
%& S_{ij} \equiv \frac{\partial g (Y_{ij}|w)}{\partial w} \Bigr|_{w=0} \, , &
%& R_{ij} \equiv \left( \frac{\partial g (Y_{ij}|w)}{\partial w} \Bigr|_{w=0}  \right)^2 +\frac{\partial^2 g (Y_{ij}|w)}{\partial w^2} \Bigr|_{w=0}
%\end{align}
Because of this universality, in the following we will restrict the
assumed output channel to be Gaussian with 
\begin{align}
& g \left( Y|w \right) = -\frac{(Y-w)^2}{2 \D} - \frac{1}{2} \log 2 \pi \D 
 \, , &
 & S_{ij}= \frac{Y_{ij}}{\D}
 \, , &
 & \hat R_{ij}= \frac{1}{\D} \, .
\label{eq:definition_S_R}
\end{align}
The ground truth channel $P_{\rm out}^{(0)}(Y_{ij}|w_{ij}^{(0)})$ that was used to generate the observation $Y_{ij}$
is also Gaussian centered around $w_{ij}^{(0)}=x_i^{(0)} x_j^{(0)}
/\sqrt{N}$ with variance $\Delta_0$.

\subsection{Ferromagnetically biased SK model and related work}

The numerical and example section of this paper will focus on one of
the simplest cases of of rank-one $r_0=r=1$ estimation with binary
signal, i.e.  \beq
P_0(x_i)=P_X(x_i)=\frac{1}{2}[\delta(x_i-1)+\delta(x_i+1)] \,
. \label{eq:SK_prior} \eeq In physics language such a prior
corresponds to the Ising spins and the Boltzmann measure
(\ref{eq:posterior}) is then the one of a Sherrington-Kirkpatrick (SK)
model \cite{SherringtonKirkpatrick75} with interaction matrix
$Y_{ij}$. After a gauge transformation $x_i \to s_i x_i^{(0)}$,
$\xi_{ij} \to \tilde \xi_{ij} x_i^{(0)} x_j^{(0)}$ this is equivalent
to the SK model at temperature $\Delta$ with random iid interactions
of mean $1/N$ and variance $\Delta_0/N$ (see for instance the
discussion in Sec. II.B of \cite{zdeborova2015statistical} or Sec. II.F
of \cite{zdeborova2010generalization}).

This variant of the SK model with
ferromagnetically biased interactions is very well known
in the statistical physics literature. The original paper
\cite{SherringtonKirkpatrick75} presents the replica symmetric phase
diagram of this model, \cite{de1978stability} computes the AT line
below which the replica symmetry is broken and Parisi famously
presents the exact solution of this model below the AT line in
\cite{parisi1979infinite}. A rather complete account for the physical
properties of this model is reviewed in \cite{MPV87}.

Note that while the replica solution and phase diagram of this model
is very well known in the physics literature, the algorithmic interpretation of the phase diagram in
terms of the AMP algorithm is recent. It is due to Bolthausen
\cite{bolthausen2014iterative} who noticed that in order for the TAP
equations \cite{ThoulessAnderson77} to converge and to reproduce the features of the well known phase
diagram, one needs to adjust the iteration indices in the TAP
equations. We will call the TAP equations with adjusted time indices
the AMP algorithm. Bolthausen proved that in the
Sherrington-Kirkpatrick model the AMP algorithm converges if and only
if above the de Almeida-Thouless (AT) line \cite{de1978stability}. 

The work of Bolthausen together with the development of AMP for sparse
linear estimation and compressed sensing \cite{DonohoMaleki09} revived
the interest in the algorithmic aspects of dense spin glasses. For a
review of the recent progress see e.g.~\cite{zdeborova2015statistical}. 
AMP for low-rank matrix estimation was studied e.g.~in 
\cite{rangan2012iterative,NIPS2013_5074,DeshpandeM14,lesieur2015phase,deshpande2015asymptotic,lesieur2015mmse},
its rigorous asymptotic analysis called the state evolution in
\cite{bayati2011dynamics,rangan2012iterative,
  deshpande2015asymptotic,javanmard2013state}. 

Correctness of the replica theory in the Bayes-optimal setting was
proven rigorously in a sequence of work
\cite{DeshpandeM14,deshpande2015asymptotic,krzakala2016mutual,barbier2016mutual,LelargeMiolane16,barbier2017stochastic,alaoui2018estimation}. While
the first complete proof is due to \cite{krzakala2016mutual}, the
Ising case discussed here is equivalent the Gauge-Symmetric
Sherrington-Kirkpatrick proven earlier in \cite{korada2009exact}.
Various applications and phase diagrams for the problem are discussed
in detail in e.g. \cite{LesKrzZde17}.

While it is well known in the physics literature \cite{de1978stability,parisi1979infinite,MPV87} that below the AT
line replica symmetry breaking is needed to describe correctly the
Boltzmann measure (\ref{eq:posterior}) the algorithmic consequences of
that stay unexplored up to date. There are several versions of the
TAP equations embodying a replica symmetry breaking structure in the literature
\cite{MPV87,yasuda2012replica} but they do not include the proper
time-indices and hence will not be able to reproduce quantitatively
the RSB phase diagram (just as it was the case for the TAP equations
before the work of \cite{DonohoMaleki09,bolthausen2014iterative}). 

In this paper we close this gap and derive approximate survey
propagation, an AMP-type of algorithm that takes
into account replica symmetry breaking. Using the state evolution
theory \cite{rangan2012iterative,DeshpandeM14,deshpande2015asymptotic}
we prove that the ASP algorithm reproduces the 1RSB phase diagram in
the limit of large system sizes. We study properties of the ASP
algorithm, resulting estimation error as a function
of the Parisi parameter,  its convergence and finite size
behavior.

%%%%%%%%%%%%%%%%%%%%%%%%%%%%%%%%%%%%%%%%%%%%%%%%%%%%%%%%%%%%%%%%%%%%%%%%%%%%%%%%%%%%

\section{Properties of approximate message passing}

\subsection{Reminder of AMP and its state evolution}

In this section we recall the standard Approximate Message Passing
(AMP) algorithm.
Within the context of low-rank matrix estimation, the AMP equations are referred as Low-RAMP and are discussed extensively in \cite{LesKrzZde17}.
In the physics literature, the Low-RAMP would be equivalent to the TAP
equations  \cite{ThoulessAnderson77} (with corrected iteration indices)
for a model of vectorial spins with local magnetic fields and general kind of two-body interactions. 
In this sense, the Low-RAMP equations encompass as a special case the original
TAP equations for the Sherrington-Kirkpatrick model \cite{SherringtonKirkpatrick75}, for the Hopfield model \cite{Mezard17,MPV87} and for the
restricted Boltzmann machine \cite{GabTraKrz15,Traelal16,Mezard17,TubMon17}.

%\subsubsection{The symmetric low-rank matrix estimation model. }

\subsubsection{Low-RAMP equations. }
\label{sec:lowRAMP}
Let us state the Low-RAMP equations to emphasize the differences and similarities with the replica symmetry breaking approach 
of Sec.~\ref{sec:the_cavity_approach_1RSB}.
The Low-RAMP algorithm evaluates the marginals of Eq.~(\ref{eq:posterior}) starting from %on the lines of 
the belief propagation (BP) equations \cite{yedidia2003understanding} (cf. the factor graph of Fig. \ref{fig:factor_graph_rep} in the case where $s=1$).
The main assumptions of BP is that the incoming messages are probabilistically
independent when conditioned on the value of the root. 
In the present case the factor $1/\sqrt{N}$ in Eq.~(\ref{eq:Y}) makes the interactions sufficiently weak 
so that the assumption of independence of incoming messages is plausible at the leading order in the large size limit. 
Moreover, this assumption is particularly beneficial in the case of continuous variables: 
since we have an accumulation of many independent messages,
the central limit theorem assures that it is sufficient to consider only means and variances to represent the result (\textit{relaxed-belief propagation})
instead of dealing with whole probability distributions.
To finally obtain the Low-RAMP equations, the further step is the so-called \textit{TAPification}: 
from the relaxed-belief propagation equations one notices that the
algorithm can be further simplified if instead of dealing with  $\mathcal{O}(N^2)$ messages associated to each directed edge,
one works with only node-dependent quantities. % (order $\mathcal{O}(N)$). 
This generates the so-called Onsager terms. Keeping track of the correct time
indices under iteration in order to preserve convergence of the iterative scheme \cite{zdeborova2015statistical}, one ends up with the Low-RAMP equations \cite{LesKrzZde17}
\begin{align}
 B^t_{i} &=  \frac{1}{\sqrt{N}} \sum_{k=1}^N  \, S_{ki} {\hat x}_k^t
 - \left( \frac{1}{N} \sum_{k=1}^N S_{ki}^2 \sigma_{k}^t  \right) \;
           {\hat x}_i^{t-1}\, ,
 \label{eq:B_AMP}
 \\
  A^t_{i} &=  \frac{1}{N} \sum_{k=1}^N  \left[  {\hat R}_{ki} \left( {\hat x}_k^t \; {\hat x}_k^{t,\top} + \sigma_{k}^t  \right)
 - S_{ki}^2  \, \sigma_{k}^t \right]\, ,
 \\
 {\hat x}_i^{t+1} &=
 \frac{\partial f_{\text{in}}}{\partial B} \left[   A_{i}^t , \,
                    B_{i}^t \right]\, ,
 \label{eq:x_AMP}
  \\
 \sigma_{i}^{t+1} &=
 \frac{\partial^2 f_{\text{in}}}{\partial B^2} \left[   A_{i}^t , \,
                    B_{i}^t \right]\, ,
  \label{eq:sigma_AMP}
\end{align}
where
\begin{align}
 f_{\text{in}} \left[   A , \, B \right]
 &\equiv  \log  \left[ \int dx \, P_X(x) \, \exp \left( B^\top x -
   \frac{1}{2} x^\top A x \right)  \right] \, .
% \\
% &\equiv  \frac{\int dx \, P_X(x) \, \exp \left( B^\top x - \frac{1}{2} x^\top A x \right) x}{\int dx \, P_X(x) \, \exp \left( B^\top x - \frac{1}{2} x^\top A x \right)}
\end{align}
Note also that these equations can be further simplified replacing $S_{ij}^2$ by its mean, without changing the leading order in $N$ of the expressions.
This simplification is also exploited in the rigours derivation of the state evolution, cf. Sec. \ref{sec:rigorous_RS}.

Practically, one initializes  ${\hat x}_i^0 = \sigma_i^1 = 0$ and ${\hat x}_i^1$ to some small numbers, then evaluates  $B^1_{i}$ and $A^1_{i}$,
then ${\hat x}_i^2$ and $\sigma_i^2$ and keep going till convergence.
%Typically in the update of $B_{i}$ and $A_{i}$ it is important to use a damping factor to achieve convergence.
The values of ${\hat x}_i \in \mathbb{R}^r$ and $\sigma_i \in \mathbb{R}^{r \times r}$ at convergence are the estimators 
of the mean and the covariance matrix of the variable $x_i$.
The mean squared error (MSE) with respect to the ground truth
$X^{(0)}$ that is reached by the algorithm is then
\begin{align}
 \text{MSE} (\hat x) = %\frac{1}{N} \int P_0\left( x^{(0)} \right) P_0\left( Y|w^{(0)} \right) \sum_{k=1}^N \left\lVert \hat x_k - x^{(0)}_k \right\rVert^2  d x d Y\, .
 { \frac{1}{N} \sum_{i=1}^N \left\lVert \hat x_i - x^{(0)}_i \right\rVert^2  } \, .
\end{align}

\subsubsection{Bethe Free Energy.}
The fixed-points of the Low-RAMP equations are stationary points of the Bethe free energy of the model.
In general, the free energy of a probability measure is defined as the logarithm of its normalization\footnotemark[1].
%\begin{align}
% \Phi(Y) = \log \langle  Z_X(Y)  \rangle - \sum_{1 \leq i < k \leq N} g\left( Y_{ij},0 \right)
%\end{align}
%where 
\footnotetext[1]{Note as in physics the free energy is usually defined as the negative logarithm.}
Within the same assumptions of Low-RAMP, the free energy can be approximated using the Plefka expansion \cite{plefka1982convergence,georges1991expand},
obtaining the Bethe free energy \cite{LesKrzZde17}
\begin{align}
 \Phi_{\text{Bethe}} =&
 \max_{A_{i},B_{i}}
 \sum_{1 \leq i \leq N}  f_{\text{in}}(A_{i},B_{i})
 - B^\top_{i} \hat x_i + \frac{1}{2} \Tr \left[ A_{i} \left( \hat x_i \hat x_i^\top + \sigma_{i} \right) \right]
 \\
 & + \frac{1}{2}   \sum_{1 \leq i , j \leq N} \biggl\{  
 \frac{S_{ij}}{\sqrt{N}} \hat x_i^\top \hat x_j  
 -  \frac{\hat R_{ij}}{2N} \Tr \left[ \left( \hat x_i \hat x_i^\top + \sigma_{i} \right) \left( \hat x_j \hat x_j^\top + \sigma_{j} \right) \right] 
 \\
& \qquad \qquad \qquad +   \frac{S_{ij}^2}{2N} \Tr \left[ \hat x_i \hat x_i^\top \sigma_j + \sigma_i  \hat x_j \hat x_j^\top - \sigma_i \sigma_j  \right] \biggl\}  
\label{eq:freeEnergyBethe}
\end{align}
where $\hat x_i$ and $\sigma_{i}$ are considered explicit functions of $A$ and $B$ as in Eqs. (\ref{eq:x_AMP}), (\ref{eq:sigma_AMP}).
The Bethe free energy is useful to analyze situations in which the AMP equations have more than one fixed-point:
the best achievable mean squared error is associated to the largest free energy.
The Bethe free energy is also useful in order to use adaptive damping to improve the convergence of the Low-RAMP equations \cite{rangan2016fixed,vila2015adaptive}.

\subsubsection{State Evolution.}
\label{sec:state_evolution_RS}
One of the advantages of AMP-type algorithms is that one can analyze their performance in the large size limit 
via the so-called \textit{State Evolution} (SE), equivalent to the cavity method in the physics literature.
Assume that $Y$ is generated from the following process:
the signal $\{x_i^{(0)} \in \mathbb{R}^{r_0} \}$ is extracted from a probability distribution $ P_0 ( \left\{ x^{(0)} \right\} ) =\prod_{i=1}^N P_0 (x_i^{(0)})$
and then it is measured through a Gaussian channel of zero mean and variance $\Delta_0$, so that
%\begin{align}
%&Y_{ij} = \frac{x^{(0)}_ix^{(0)}_j}{\sqrt{N}} + \xi^{(0)}_{ij} 
%\, , &
%& \forall i\leq j\:. 
%\end{align}
%Therefore 
the probability distribution of the matrix elements $Y_{ij}$ is given by 
\begin{align}
& P_{\text{out}}(Y_{ij}) = 
\exp\left[g^{(0)} \left( Y_{ij} \bigg| \frac{ x_i^{(0)}  x_j^{(0)} }{ \sqrt{N}} \right)\right]
\, , &
& g^{(0)}(Y|w) = -\frac{1}{2}\ln(2\pi \Delta_0) - \frac{1}{2\D_0}(Y-w)^2 \, . 
\label{eq:def_output_channel}
\end{align}
Note as here we are considering the general situation in which the prior $ P_0 (x) $ and the noise channel $P_{\text{out}}(Y_{ij})$ (and possibly also the rank $r_0$)
are not known and are in principle different from the ones used in the posterior Eq.~(\ref{eq:posterior}).
If both the prior and the channel are exactly known, we say to be in the \textit{Bayes-optimal case}.

Central limit theorem assures that the averages of $B_{i}$ and $A_{i}$ over $Y$ 
%can be performed in a simple way: central limit theorem assures that they become 
are Gaussian with 
\begin{align}
\overline{B_{i}^t} &= \frac{M^t}{\D} x_i^{(0)}
\, , &
\overline{(B_{i}^t)^2} - \overline{B_{i}^t}^2
&= \frac{\D_0}{\D^2} Q^t 
\, , &
 \overline{A_{i}^t} &= 
 \frac{\D_0}{\D^2} Q^t - \frac{\D_0-\D}{\D^2}  \Sigma_X^t \, ,
\label{cavity:eq}
\end{align}
while the variance of $A_{i}$ is of smaller order in $N$ and
where we have defined the order parameters
\begin{align}
& M^t =\frac 1N \sum_{k=1}^N\hat x_k^t \, x_k^{(0),\top} \in
  \mathbb{R}^{r \times r_0}
\, , &
& Q^t =\frac 1N \sum_{k=1}^N \hat x_k^t \, \hat x_k^{t,\top} \in \mathbb{R}^{r \times r}
\, , &
& \Sigma_X^t = \frac 1N \sum_{k=1}^N \sigma_{k}^t \in \mathbb{R}^{r
  \times r}\, .
\label{eq:def_MtQ_AMP}
\end{align}
Using then Eqs.~(\ref{eq:x_AMP}), (\ref{eq:sigma_AMP}) to fix self-consistently the values of the order parameters, one obtains the state evolution equations  \cite{LesKrzZde17}
\begin{align}
& M^{t+1} = \mathbf{E}_{x^{(0)},W} \left[  
 \frac{\partial f_{\text{in}}}{\partial B} \left(  A^t  ,
 B^t
 \right) x^{(0),\top} \right] \, ,
 \label{eq:SE_RS_M}
 \\
& Q^{t+1} = \mathbf{E}_{x^{(0)},W} \left[  
 \frac{\partial f_{\text{in}}}{\partial B} \left(   A^t  , B^t
 \right) \,
  \frac{\partial f_{\text{in}}}{\partial B} \left( A^t  , B^t
  \right)^\top \right] \, ,
 \\
& \Sigma_X^{t+1} = \mathbf{E}_{x^{(0)},W} \left[  
 \frac{\partial^2 f_{\text{in}}}{\partial B^2} \left(  A^t  , B^t
  \right)  \right] \, ,
 \label{eq:SE_RS_Sigma}
\end{align}
where $x^{(0)} $ is distributed according to $P_0(x)$, $W$ is a Gaussian noise of zero mean and unit covariance matrix
and we have defined
\begin{align}
& A^t \equiv \frac{\D_0}{\D^2} Q^t - \frac{\D_0-\D}{\D^2}  \Sigma_X^t 
\, , &
& B^t \equiv \frac{M^t}{\D} x^{(0)} + \sqrt{\frac{\D_0}{\D^2} Q^t } \; W \, .
\end{align}
Similarly, the large size limit of the Bethe free energy Eq.~(\ref{eq:freeEnergyBethe}) is given by
\begin{align}
 \Phi_{\text{RS}} = \max \left\{ \phi_{\text{RS}} \left( M,Q,\Sigma \right) , 
 \frac{\partial \phi_{\text{RS}}}{\partial M} =  \frac{\partial \phi_{\text{RS}}}{\partial Q} =  \frac{\partial \phi_{\text{RS}}}{\partial \Sigma} = 0\right\}
\end{align}
with
% [I think there must be a wrong factor 2 in eq (144) of Thibault]
\begin{align}
\label{eq:RSfreeEnergy}
 \phi_{\text{RS}} \left( M,Q,\Sigma \right) &=
 \mathbb{E}_{x_0,W} \left[   f_{\text{in}} \left(  A  , B \right) \right] + \frac{\D_0}{4 \D^2} \Tr \left[Q Q^\top \right] -  \frac{1}{2 \D} \Tr \left[M M^\top \right]
 - \frac{\D_0 - \D}{4  \D^2} \Tr \left[(Q+\Sigma)(Q+\Sigma)^\top \right] \, .
\end{align}
This free energy coincides with the one obtained by replica theory under replica symmetric (RS) ansatz.

%\subsubsection{Error and Estimators.}
%Ideally, an inference procedure aims to minimize the mean squared error (MSE) %with the ground truth 
%\begin{align}
% \text{MSE} (\hat x) = %\frac{1}{N} \int P_0\left( x^{(0)} \right) P_0\left( Y|w^{(0)} \right) \sum_{k=1}^N \left\lVert \hat x_k - x^{(0)}_k \right\rVert^2  d x d Y\, .
% \overline{ \frac{1}{N} \sum_{k=1}^N \left\lVert \hat x_k - x^{(0)}_k \right\rVert^2  } \, .
%\end{align}

For Low-RAMP, the mean squared error (MSE) can be evaluated from the state evolution as
\begin{align}
 \text{MSE}_{\text{Low-RAMP}}  = \text{Tr} \left[ \mathbf{E}_{x_0} \left[ x_0 x_0^\top \right] - 2 M + Q \right] 
 \label{eq:MSE_lowRAMP}
\end{align}
allowing us to assess the typical performance of the algorithm.

\subsubsection{The rigorous approach}
\label{sec:rigorous_RS} 
The fact that state evolution accurately tracks the behavior of the AMP algorithm has been proven \cite{bayati2011dynamics,bolthausen2014iterative,rangan2012iterative,DeshpandeM14,deshpande2015asymptotic}. In this section, we shall discuss the main lines of this progress.

First, let us rewrite the AMP equations (\ref{eq:B_AMP}--\ref{eq:sigma_AMP}) as follows, using a vectorial form:
\begin{eqnarray}
{\bf B}^{t} &=& \frac S{\sqrt N} \hat {\bf x}^t - b^t \hat {\bf x}^{t-1} \, , \label{flo1-1}\\
{\bf x}^{t+1} &=&\eta_t({\bf B}^{t})\, , \label{flo1-2}
\end{eqnarray}
where the scalar quantity $b$ is computed as
\begin{eqnarray}
b^{t+1} = \frac{ {\mathbb E} [S^2] }N\sum_{i=1}^N  \eta'_t({\bf B}^{t}_i)\, , \label{flo1-3}
\end{eqnarray}
with $\mathbb E [S^2]=\Delta_0/\Delta^2$ is the average value of the square of each element of the matrix $S$. The link with the AMP equations written previously is direct when one choose the denoising function $\eta_t(B)$ to be
\begin{eqnarray}
\eta_t(B) &\vcentcolon=& \partial_B f_{\rm in}(B,A^t) \, . \label{flo1-4}
\end{eqnarray}
The strong advantage of the rigorous theorem is that it can be stated for {\it any} function $\eta_t()$ (under some Lipschitz conditions), and we shall make advantage of this in the next chapter.

With these notations, the state evolution is rigorous thanks to a set of works due to Rangan and Fletcher \cite{rangan2012iterative}, Deshpande and Montanari \cite{DeshpandeM14}, and Deshpande, Abbe, and Montanari \cite{deshpande2015asymptotic} \footnote{See in particular lemma 4.4 in \cite{deshpande2015asymptotic}. One can go from their notation to ours by a simple change or variable. First what they denote as $Y^{\rm DAM}$ corresponds to $Y\sqrt{\lambda}$, with $\lambda=\Delta_0^{-1}$, so that $Y = Y^{\rm DAM}/\sqrt{\lambda}$. The message passing is then easily mapped by the change of variable: $Y=S {\Delta}$, $B=x/(\Delta \sqrt{\lambda})$ and the denoising function $f_{\rm DAM}(x)$ is replaced by $\eta(B)$. }
\begin{theorem}[State Evolution for Low-RAMP \cite{rangan2012iterative,DeshpandeM14,deshpande2015asymptotic}] \label{theorem1} Consider the problem $Y=X^{(0)}[X^{(0)}]^T/\sqrt{N}+\sqrt{\Delta_0} \xi$ as in (\ref{eq:Y}), and define $S=Y/\Delta$.  Let $\eta_t(B^t)$ be a sequence of functions such that both $\eta$ and $\eta'$ and Lipschitz continuous, then the empirical averages \begin{align}
 M^{t}& = \frac 1N \sum_i \hat x^t_i x^{(0)}_i \,, &Q^{t} &= \frac 1N  \sum_i (\hat x^t_i)^2\,,&\Psi^t &= \frac 1N  \sum_i \psi(B_i^t, x^{(0)}_i)\,,
\end{align}
for a large class of function $\psi$ (see \cite{deshpande2015asymptotic}) converge, when $N\to  \infty$, to their state evolution predictions where 
\begin{align}
M_{\rm SE}^{t} &= \mathbb{E} \left[x^{(0)} \eta_t(Z)\right]\, , &Q^{t}_{\rm SE} &=  \mathbb{E}\left[ \eta_t(Z)^2 \right]\, ,& \Psi^t_{\rm SE} &=  \mathbb{E} \left[\psi(Z, x^{(0)}) \right]\,,
\end{align}
 where $Z$ is a random Gaussian variable with mean $\frac{M^t x^{(0)}}{\Delta}$ and variance $\frac{\Delta Q^t}{\Delta^2}$, and $x^{(0)}$ is distributed according to the prior $P_0$.
\end{theorem}
This means that the variable $B_i^t$ converges to a Gaussian with mean $M^t X_0/\Delta$ and variance $\Delta_0 Q^t/\Delta^2$  as predicted within the cavity method in (\ref{cavity:eq}). This provides a rigorous basis for the analysis of AMP. In particular, one can choose $\eta_t(B) =\partial_B f_{\rm in}(B,A_{\rm SE}^t)$ with $A_{\rm SE}^t = \mathbb{E} \left[\eta'_t(Z)\right]$ and this covers the AMP algorithm discussed in this section.

\subsection{Phase diagram and convergence of AMP out of the Nishimori line}
\label{sec:phase_diagram_RS}

In the Bayes-optimal setting, where the knowledge of the model is
complete as in Eq.~(\ref{Bayes_optimal}), the statistical physics analysis 
of the problem presents important simplifications known as the Nishimori conditions \cite{NishimoriBook}.
Specifically, as direct consequence of the Bayesian formula and basic properties of probability distributions, it is easy to see that \cite{zdeborova2015statistical}
\begin{align}
&\text{Bayes-optimal: } & & \mathbb{E} \left[ f(x_1,x_2) \right] =
                            \mathbb{E} \left[ f(x^{(0)},x) \right] \, ,
\label{eq:Nishimori_condition}
\end{align}
where $f$ is a generic function and $x_1$, $x_2$ and $x$ are distributed according the posterior while $x^{(0)}$ is distributed according $P_0(x)$.
The most striking property of systems that verify
the Nishimori conditions is that there cannot be any replica symmetry breaking in the equilibrium solution of these
systems\footnote{Note, however, that the dynamics can still be complicated, and 
in fact a dynamic spin glass (d1RSB) phase can appear also in the Bayes-optimal case \cite{LesKrzZde17}.} \cite{nishimori2001absence}. 
This simplifies considerably the analysis of the Bayes-optimal inference. 
From the algorithmic point of view, the replica symmetry assures that the marginals are asymptotically exactly described 
by the belief-propagation algorithms \cite{zdeborova2015statistical}. 
In this sense, AMP provides an optimal analysis of Bayes-optimal inference. 
In real-world situations it is, however, difficult to ensure the satisfaction of the Nishimori conditions,
as the knowledge of the prior and of the noise is limited and/or the parameter learning is not perfect.
The understanding of what happens beyond the Bayes-optimal condition is then crucial.

Using state evolution Eqs.~(\ref{eq:SE_RS_M})-(\ref{eq:SE_RS_Sigma}) - or equivalently replica theory for a replica symmetric ansatz - 
one can obtain the phase diagram of Low-RAMP in a mismatching models setting.
As soon as we deviate from the Bayes-optimal setting and mismatch the models,
the Nishimori conditions are not valid anymore and so the replica symmetry is not guaranteed
and Low-RAMP is typically not optimal.
If the mismatch is substantial, the replica symmetry gets broken and we enter in a glassy phase.

\begin{figure}[h!]
\centering
\includegraphics[width=0.8\columnwidth]{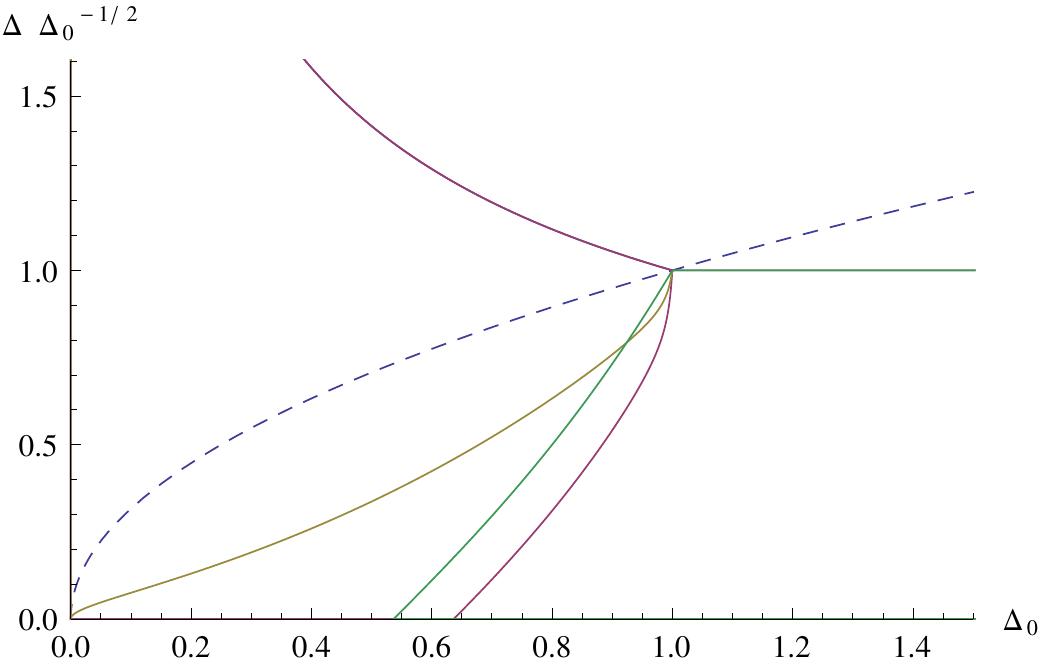} \\
\includegraphics[width=0.4\columnwidth]{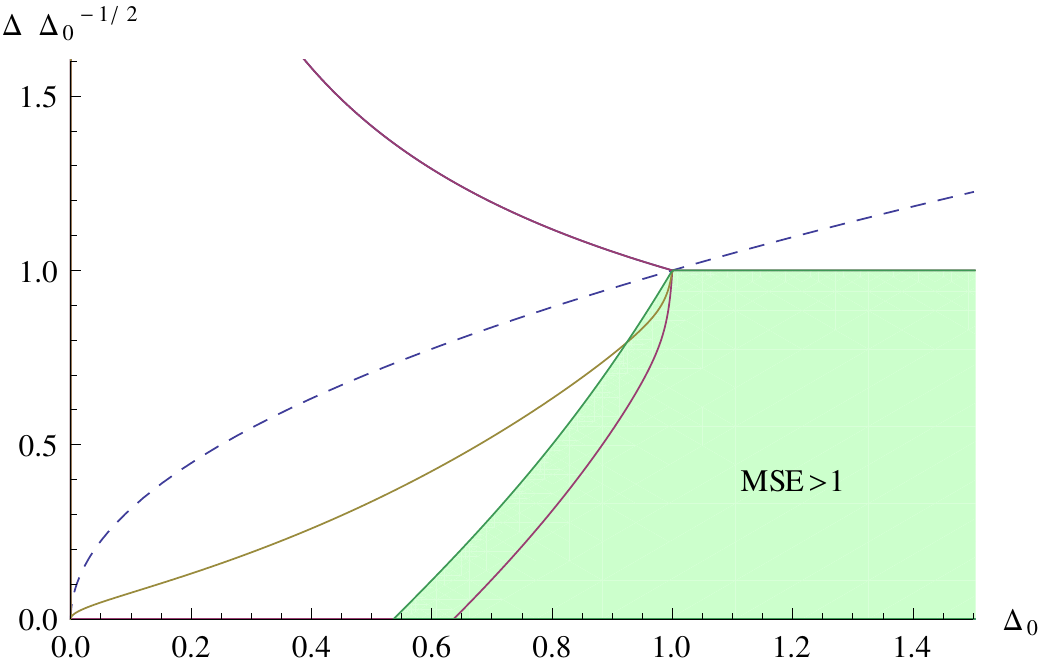}
\includegraphics[width=0.4\columnwidth]{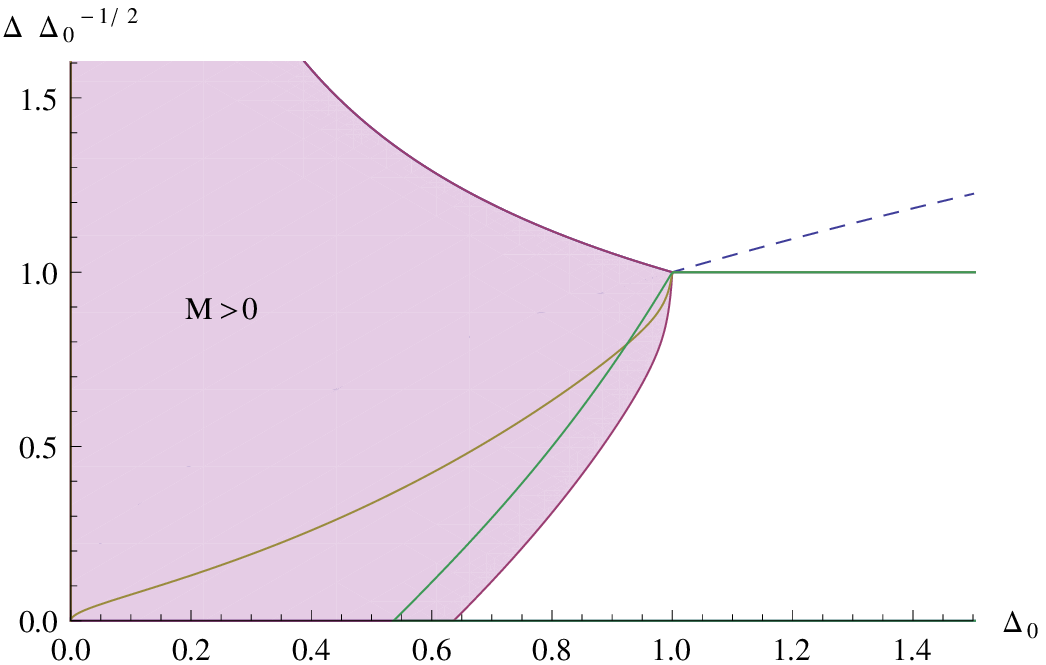}
\includegraphics[width=0.4\columnwidth]{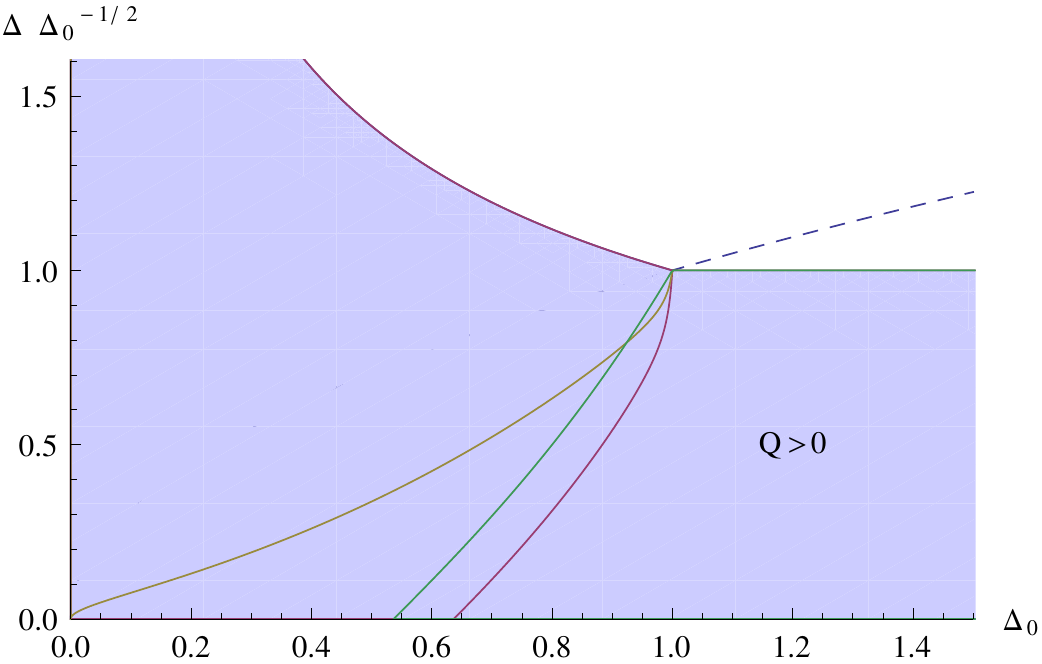}
\includegraphics[width=0.4\columnwidth]{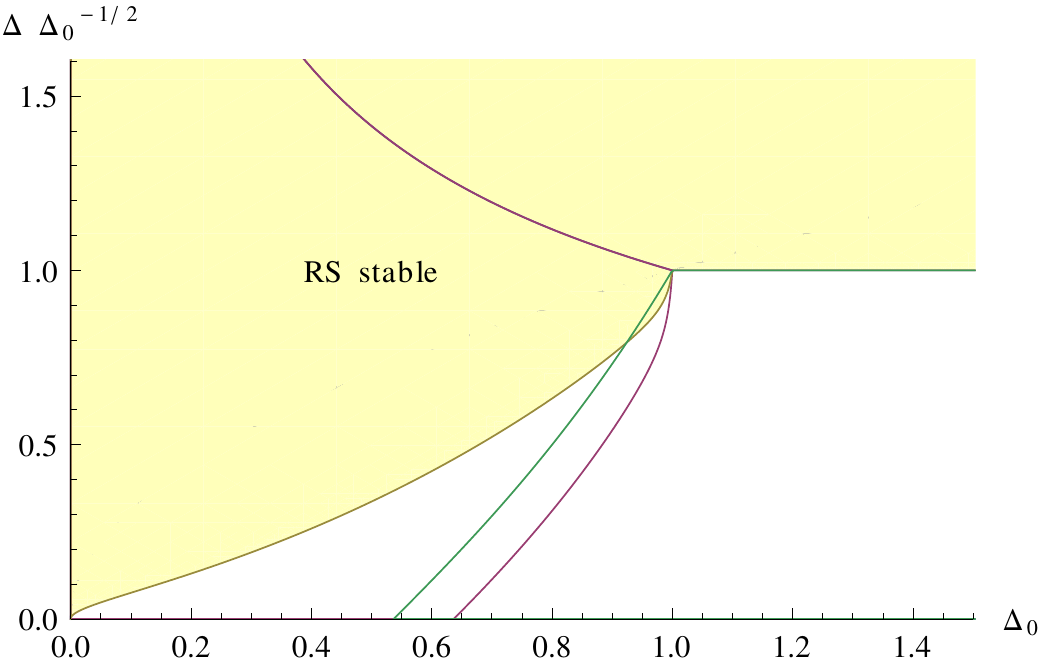}
\caption{Phase diagram for the planted SK model (rank-one matrix
  estimation with prior $\pm 1$ with
  probability $1/2$) with $\Delta_0$ being the truth noise variance
  and $\Delta$ the assume noise variance, obtained solving the SE Eqs.~(\ref{eq:SE_RS_M})-(\ref{eq:SE_RS_Sigma}). 
The dashed line is the Nishimori line $\Delta = \Delta_0$.
Four different regions are featured: $\text{MSE}>1$, $M>0$, $Q>0$ and
the region of RS stability. % obtained by solving the state evolution equations.
Note that the Nishimori line lies in the RS stability region.
For $\Delta_0 < 1$, the line where $M$ and $Q$ starts to be different from zero (coming from large $\D$) corresponds to $\Delta=1$.
%The line for $\Delta_0 > 1$ is $\Delta = \sqrt{\Delta_0}$.
The line $M=0$ crosses $\Delta = 0$ at $\Delta_0 = 2 / \pi \sim 0.6366$.
The lines $\text{MSE}=1$ intersects the RS stability line at $\Delta_0 \simeq 0.923$, $\Delta \simeq 0.757$,
so that for $0.923\lesssim \Delta_0 < 1 $ we have stable solutions with $\text{MSE}>1$.
}
\label{fig:phaseDiagramRS}
\end{figure}

\begin{figure}[h!]
\centering
\includegraphics[width=0.84\columnwidth]{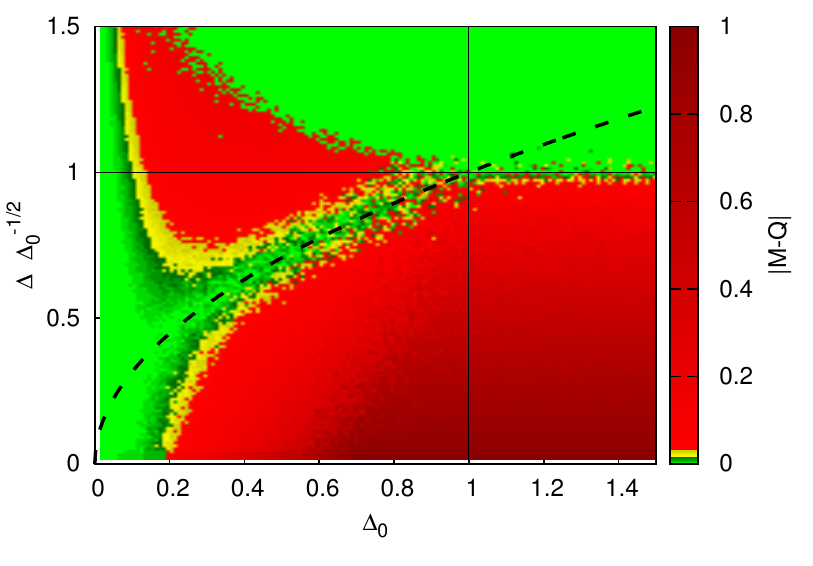} \\
\includegraphics[width=0.42\columnwidth]{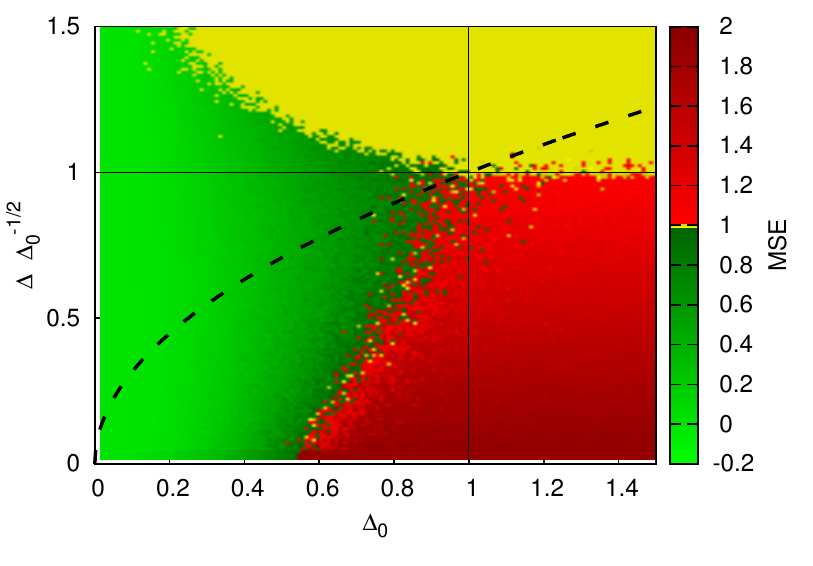}
\includegraphics[width=0.42\columnwidth]{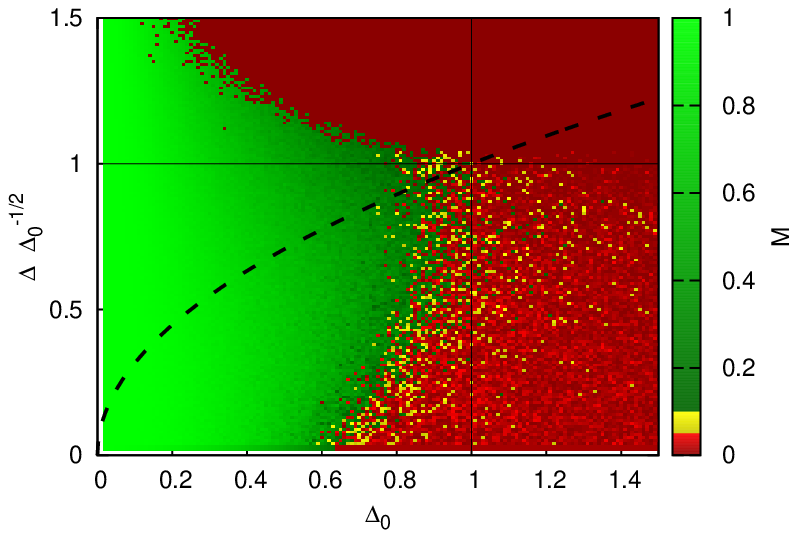}
\includegraphics[width=0.42\columnwidth]{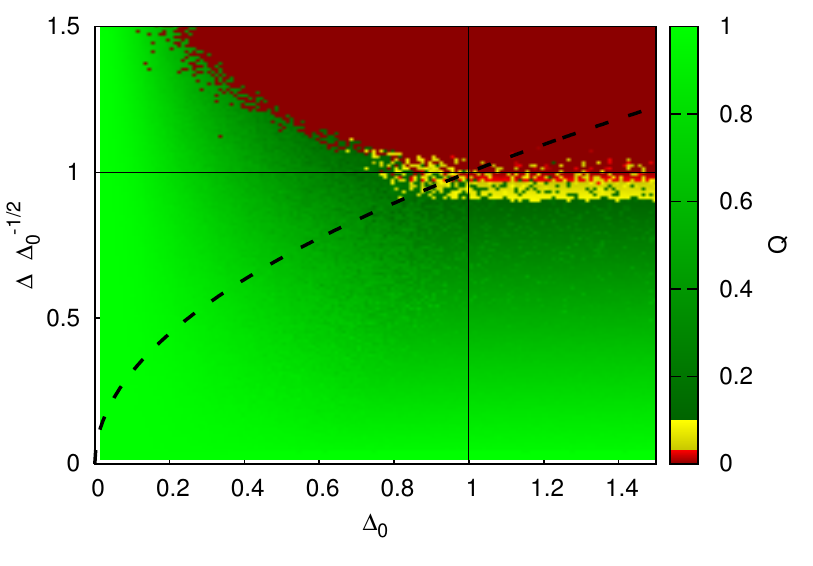}
\includegraphics[width=0.42\columnwidth]{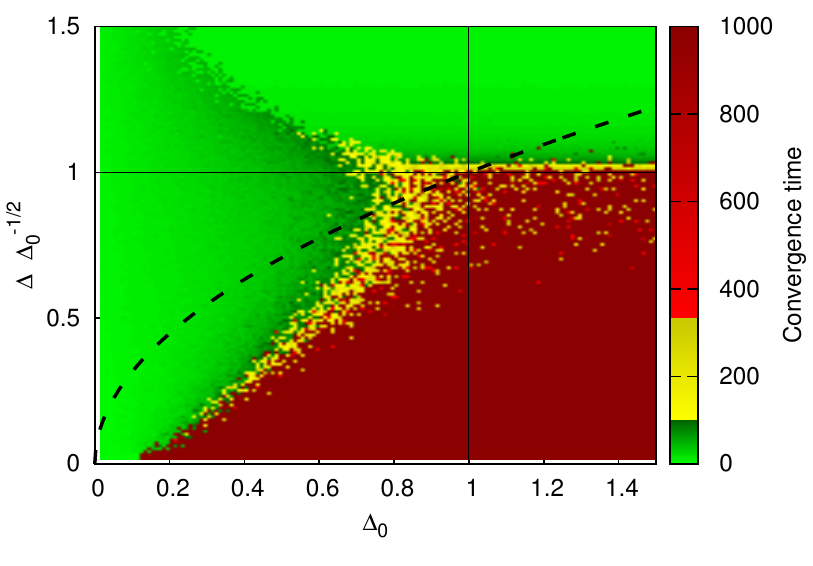}
\caption{Heat-map of the values obtained running AMP for size $N=5000$
  for the planted SK model. 
The dashed line is the Nishimori line $\Delta = \Delta_0$.
We show the value of $|M-Q|$ and $\text{MSE}$, $M$, $Q$ and the
convergence time. 
The iterations are stopped when the average change of a variable in a single iteration is less than $10^{-8}$ or at 1000 iterations if convergence is not reached.
All is in perfect agreement with the state evolution
prediction from Fig. \ref{fig:phaseDiagramRS}.
%The iterations are stopped when the average change of a variable in a single iteration of AMP is less than $10^{-8}$.
}
\label{fig:phaseDiagramAMP5k}
\end{figure}

To illustrate the behaviour of AMP out of the Bayes-optimal case, let us consider the planted Sherrington-Kirkpatrick (SK) model \cite{SherringtonKirkpatrick75},
a well-known model both in computer science and physics communities.
This is a particular case of the low-rank matrix estimation model Eq.~(\ref{eq:Y})
for a bimodal prior, cf. Eq.~(\ref{eq:SK_prior}). 
%\begin{align}
% P_0(x)=P_X(x)= \frac{1}{2}\delta(x-1)+ \frac{1}{2} \delta(x+1)
%\label{eq:SK_prior}
%\end{align}
The output noise is Gaussian with the assumed variance $\Delta$
mismatching the ground-truth variances $\Delta_0$, cf. Eqs.~(\ref{eq:definition_S_R}),(\ref{eq:def_output_channel}).
In Fig.~\ref{fig:phaseDiagramRS} we show the phase diagram of the model obtained by the SE Eqs.~(\ref{eq:SE_RS_M})-(\ref{eq:SE_RS_Sigma}) as
$\Delta$ and $\Delta_0$ are independently changed. The dashed line is the Nishimori line where the Bayes-optimal condition $\Delta=\Delta_0$ holds.
%The Nishimori conditions guarantee that $M=Q$, so that $\text{MSE}=1-M\leq 1$, cf. Eq.~(\ref{eq:MSE_lowRAMP}). Out of Nishimori, $M$ and $Q$ are usually different and the MSE can be
%larger than 1, meaning that the error can be larger than the one achieved by a random guess from the prior.
For $\Delta_0 > 1$ the signal/noise ratio is too low and the estimation is impossible, $M$ is always zero and $\text{MSE} \geq 1$.
Note that a MSE larger than 1 is worse than the error achieved by a random guess from the prior.
At $\Delta_0 = 1$, the inference starts to be successful only on the Nishimori line.
As one is further and further away from the Nishimori line, the algorithm needs a lower and lower value of noise in order to achieve a positive overlap $M$ and a MSE lower than random guess. 
In particular, if $\D$ is too large, i.e. $\D > 1$, one always reaches the trivial fixed point $M=Q=0$.
The case $\D \to 0$ is also especially meaningful for an evaluation of the MAP estimator Eq.~(\ref{eq:MAP_estimator}).
In this case, the value of the overlap with the ground truth $M$ is zero till $\Delta_0 = 2 / \pi \sim 0.6366$ \cite{NishimoriBook}
while the MSE is larger than 1 till $\Delta_0 \simeq 0.55$, almost half than the equivalent threshold on the Nishimori line.
%{\bf Explain here the Gauge mapping to the feromagnetically biased SK
%model.}

The SE equations describe the average (or large size) behaviour of Low-RAMP and as such they converge in all the phase diagram.
%cf. Fig.~\ref{fig:convergence_AMP_vs_SE}. 
Single finite size instances of Low-RAMP do not converge point-wise
too far from Nishimori line, cf. Fig.~\ref{fig:phaseDiagramAMP5k}.

The analysis of the point-wise convergence on single instances of
Low-RAMP can be obtained looking at the Hessian eigenvalues of
Low-RAMP equations. A simple way to derive a necessary convergence
criterium is to ask whether or not the fixed-point of AMP is stable
with respect to weak, infinitesimal, random perturbations.
Let us see how it can be derived with the notations of section
\ref{sec:rigorous_RS}. If we perform, in Eq.~(\ref{flo1-1}), the
perturbation ${\bf B}^t \to {\bf B}^t+ {\bm \epsilon}^t$, where
${\bm \epsilon}$ is a i.i.d.  infinitesimal vector sampled from
${\cal N} ({0,\epsilon})$, then we may ask how this perturbation is
carried out at the next step of the iteration. From the recursion
(\ref{flo1-1})-(\ref{flo1-4}), we see that (a) $b^{t+1}$ is not
modified to leading order, as
$\eta'(B+\delta B)=\eta'(B) + \delta B \eta''(B)$ and the average in
(\ref{flo1-3}) makes the perturbation of $O(1/\sqrt N)$ and (b)
${\bf B}^{t+1} \to {\bf B}^{t+1}+ {\bm \epsilon}^{t+1}$ with
${\bm \epsilon}^{t+1}=S {\bm \epsilon}^{t} \eta'({\bf
  B^t})/\sqrt{N}$. The $\ell_2$ norm of the perturbation has thus been
multiplied, up to a constant ${\mathbb E}(S^2)=\Delta_0/\Delta^2$, by
the empirical average of $\eta'({\bf B})$. The fixed-point will be
stable if the perturbation does not grow. This yields, coming back to
the main notations of the article, to the following criterion:
\begin{align}
 \lambda = 1 - \frac{\D_0}{\D^2} \mathbb{E}_{x^{(0)},W} \left[ \left(
  \frac{\partial^2 f_{\rm in}(A,B)}{\partial B^2}  
 \bigg|_{A=\frac{\D_0}{\D^2} Q - \frac{\D_0-\D}{\D^2} \Sigma , \, B=
  \frac{M}{\D} x^{(0)} + \sqrt{\frac{\D_0 Q}{\D^2}} W } \right)^2
  \right]\, .
\label{eq:repliconRS}
\end{align}
%aka
%
For positive $\lambda$ the perturbation decreases and AMP algorithm converges, for negative
$\lambda$ it grows and the algorithm does not converge.  
Interestingly, condition (\ref{eq:repliconRS}) is equivalent to the stability of replica symmetric
(RS) solutions in the replica theory given by  RS replicon
\cite{de1978stability} (and indeed \cite{bolthausen2014iterative} has
shown rigorously in the SK model, convergence of the AMP algorithm in
a phase where the replica symmetric solution is stable). 
%the TAP equations converges in the replica
%symmetric phase).
%aka
%\begin{align}
% \lambda = 1 - \frac{\D_0}{\D^2} \int dh \, P(0,h) \, \left[ f''(1,h) \right]^2
%\label{eq:repliconRS}
%\end{align}
%where
%\begin{align}
% f(1,h) &= \ln \int dx \, P(x) \, \exp \left( - \frac{\D-\D_0}{2 \D^2} q_d x^2 - \frac{\D_0}{2\D^2} q x^2 + h x  \right)
%\\
% P(0,h) &= \frac{\D}{\sqrt{2 \pi \D_0 q}} \int dx \, P^{(0)}(x) \, \exp \left[ - \frac{\left( \D h - m x \right)^2}{2 \D_0 q}  \right]
%\end{align}
The line where the RS solution becomes unstable (and Low-RAMP stops converging point-wise) is shown in Fig.~\ref{fig:phaseDiagramRS} with a yellow line.
Note that, as expected, the RS stability line lies always below the Nishimori line and the two touch at the tri-critical point $\D=\D_0=1$.
For small $\D_0$, the stability line is roughly given by $\D \simeq 4 \sqrt{\D_0 / (2 \pi)} \exp(-1/(2\D_0))/3 $
and it touches the $\D = 0$ axis only for $\D_0 \to 0$. 
This means that for any finite $\D_0$ the RS solution becomes always unstable for $\D$ small enough.
%In particular, if one adopt the MAP estimator one ends up in the unstability region.
%
The stability line is always above the $M=0$ line, cf. Fig.~\ref{fig:phaseDiagramRS}.
%This means that typically when Low-RAMP converges point-wise, it always converges either to the trivial fixed-point $M=Q=0$ % (easy to spot)
%or to a nontrivial fixed-point with positive overlap with the ground truth (namely, the case $Q>0$ and $M=0$ is excluded).
Note that it is possible to get RS stable solutions with $\text{MSE}>1$: 
so it is possible that AMP converges point-wise
to something worse than random guess.
%, cf. Fig.~\ref{fig:MSE_and_Eig_AMP}.

In Fig.~\ref{fig:phaseDiagramAMP5k} we show the results obtained by running Low-RAMP for single instances of the same problem for $N=5000$.
We iterate the Low-RAMP equations, without any damping, for $10^3$ steps or we stop when the average change of the variables in a single step of iteration is less than $10^{-8}$.
The four regions highlighted in Fig. \ref{fig:phaseDiagramRS} are well
distinguishable. We also show the value of $|M-Q|$, which is zero on
the Nishimori line.
In particular, it is clear that the point-wise convergence is very fast (less than 100 iterations) well inside the RS stability region,
then the convergence time increases rapidly at the boundary of the RS stability region and then stops converging point-wise outside of it.

%Within the framework of replica theory it is possible to obtain information about the transition at the RS instability point. 
%In this case, the Parisi breaking point $s$ is always in the range [0,1], implying a continuous transition.
%The derivative $\dot q(x)$ in $s$ is positive, suggesting a transition towards a full RSB phase.
%
The scenario illustrated in this section for SK is quite general within low-rank matrix estimation problems.
Some differences arise when $x$ is very sparse: in this case on the Nishimori line there is a first order transition \cite{LesKrzZde17}
and the scenario becomes more complicated, with the coexistence of different phases close to the transition.

\begin{figure}[h!]
\centering
\includegraphics[width=0.49\columnwidth]{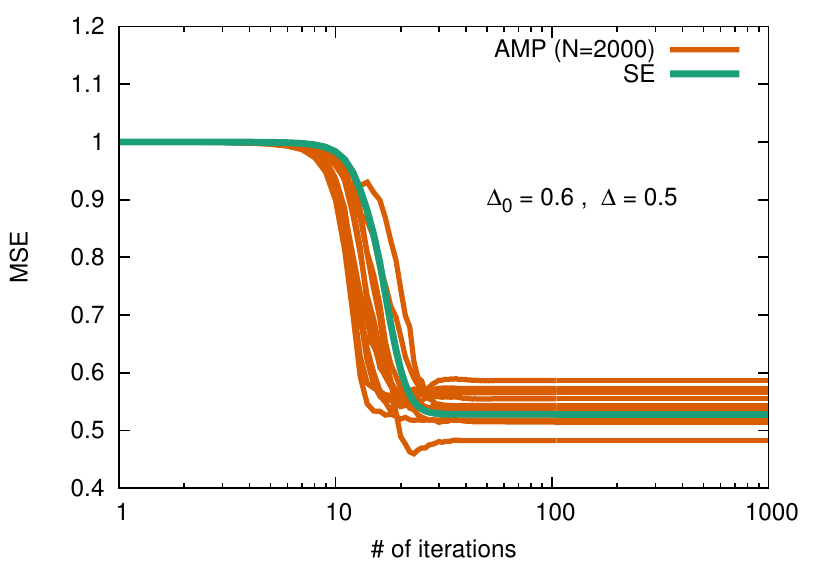} 
\includegraphics[width=0.49\columnwidth]{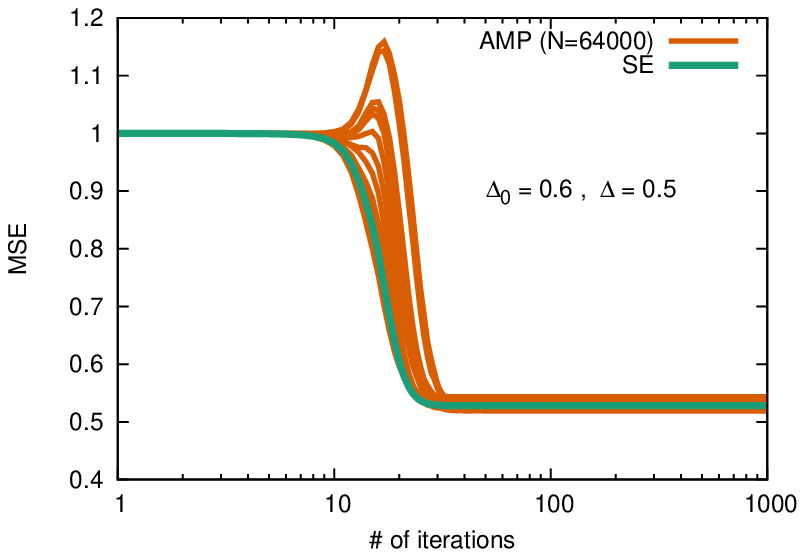} \\
\includegraphics[width=0.49\columnwidth]{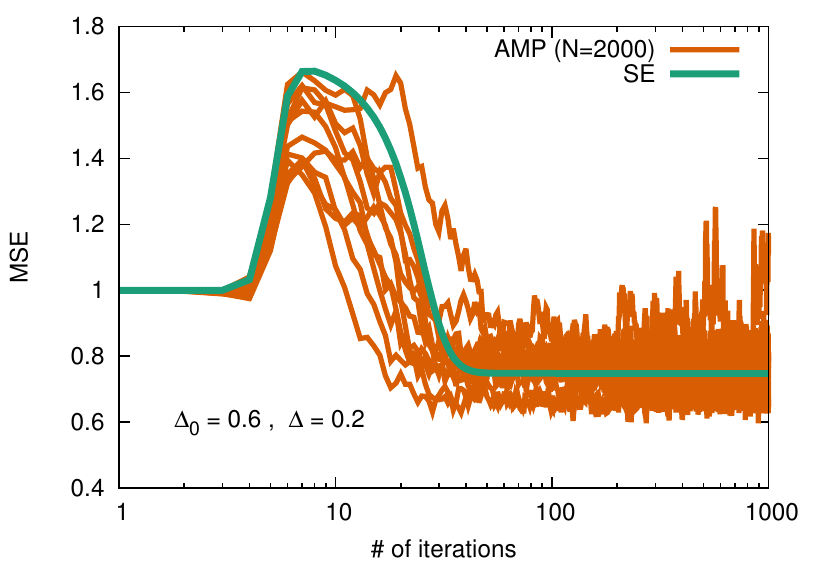} 
\includegraphics[width=0.49\columnwidth]{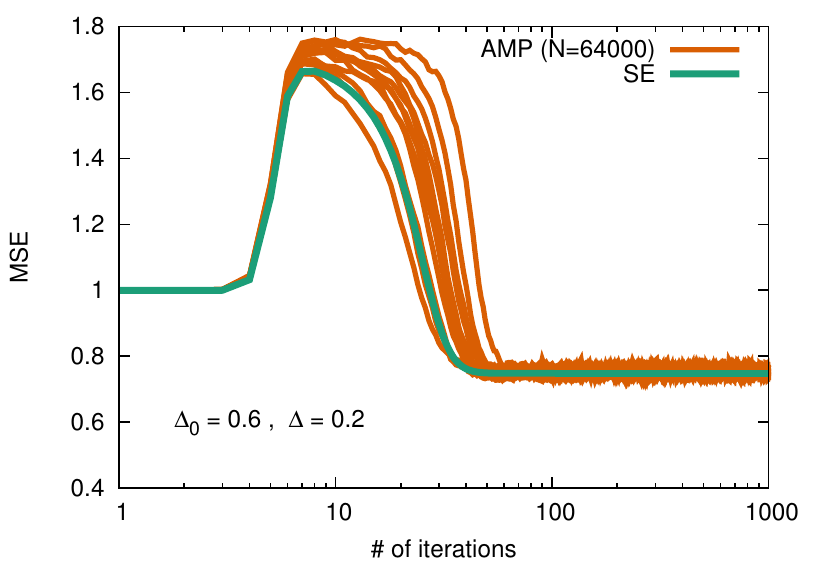} 
\caption{Value of the MSE under iterations of AMP equations for several instances for $N=2000$ (left) and $N=64000$ (right) in the planted SK model 
with $\Delta_0 = 0.6$ and $\Delta=0.5$ (top), $\Delta=0.2$ (bottom), compared with the SE.
}
\label{fig:convergence_AMP_vs_SE}
\end{figure}

\subsection{Optimality and restored Nishimori condition}
\label{sec:optimality_and_mismatching}

The Bayes-optimal setting, as reminded in Sec.~\ref{sec:phase_diagram_RS},  is a very special situation:
the Nishimori condition Eq.~(\ref{eq:Nishimori_condition}) guarantees that 
the replica symmetry is preserved and that AMP is optimal (in absence
of a 1st order phase transition).
One consequence of the Nishimori condition is that the typical overlap with the ground truth $M$ 
and the self-overlap $Q$, as defined in Eqs.~(\ref{eq:def_MtQ_AMP}), are equal.
In this case the MSE is then given by
\begin{align}
 \text{MSE} = \mathbb{E} \left[ x_0^2 \right] - 2M +Q 
 = \mathbb{E} \left[ x_0^2 \right] - M \, ,
\end{align}
and the minimum MSE is obtained at the maximum overlap $M$.
For mismatching models $M$ and $Q$ are typically different from each other and it is 
not immediate to realize under what conditions the MSE is minimized. 
We highlighted this property in the main panel of Fig.~\ref{fig:phaseDiagramAMP5k}:
apart from the trivial solutions $M=Q=0$ (for large $\D_0$ and $\D$) and $M=Q=1$ (for very small $\D_0$),
the overlap with the ground truth $M$ and the self-overlap $Q$ are very close only near the Nishimori line.
The Nishimori line is also the region in which AMP is optimal, so it is spontaneous to ask what is the relation (if any) between these two properties
in a general setting.

Consider a situation in which we have two or more parameters to tune. These can be parameters in the prior, the variance of the noise $\D$ or also the free parameters
in the general belief-propagation equations, see for example the parameter $s$ for the ASP algorithm that we will describe in Eqs.~(\ref{eq:TT})-(\ref{eq:DD1}).
The fundamental question is then for which set of values for the parameters it is possible
to obtain optimal estimation error, and how to find them algorithmically.
One evident answer is that if all the parameters are chosen to exactly match the values of the ground truth distribution,
one ends up on the Nishimori line and then we know that the inference
is optimal, and the MSE minimum. But this is not the unique answer.

\begin{figure}[h!]
\centering
\includegraphics[width=0.47\columnwidth]{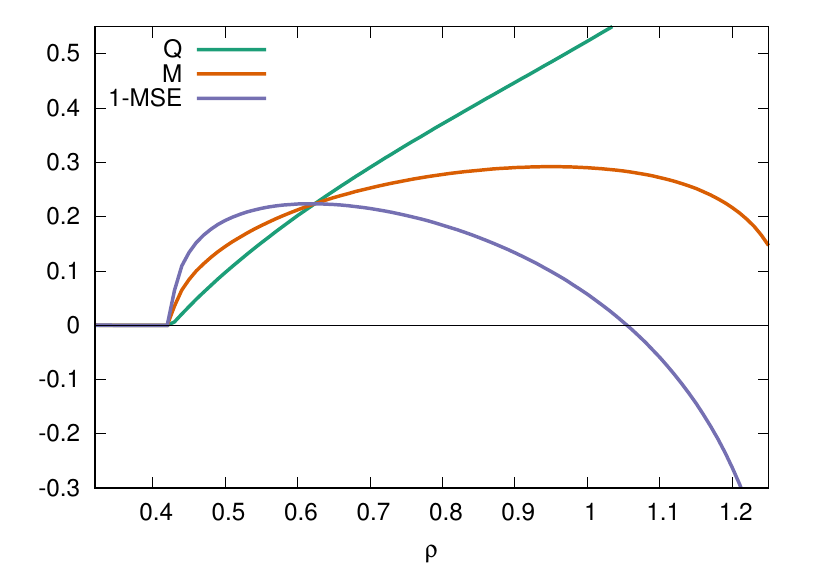} \\
\includegraphics[width=0.47\columnwidth]{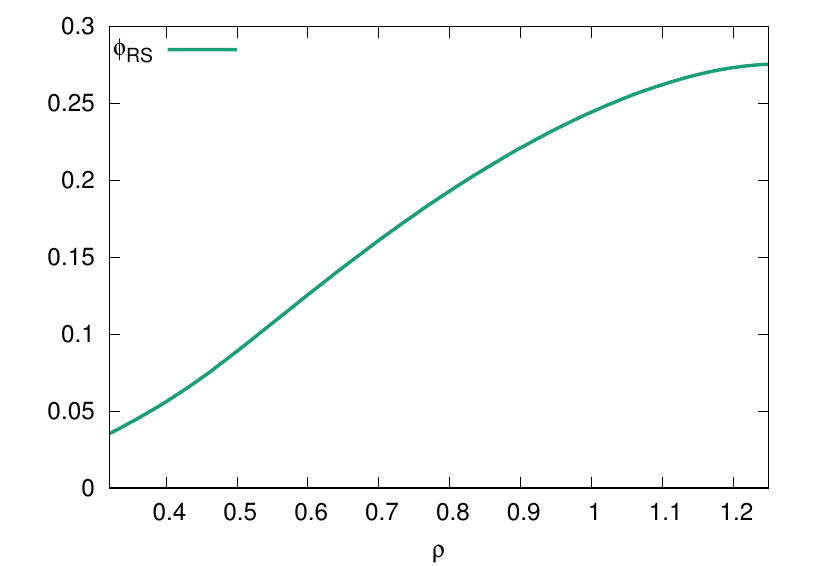}
\includegraphics[width=0.47\columnwidth]{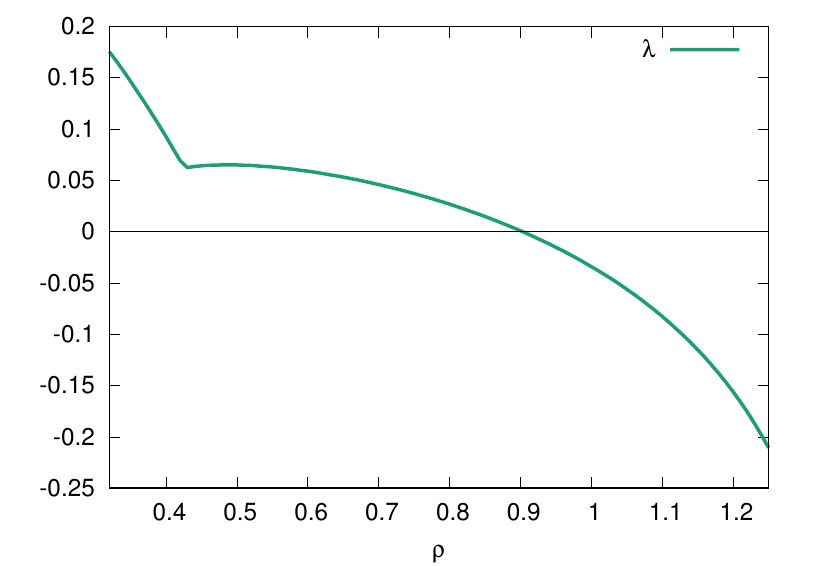} 
\caption{Value of $M$, $Q$ and MSE (top), Bethe free energy
  $\phi_{\text{RS}}$ (bottom left) and RS stability eigenvalue $\lambda$ (bottom right) as
obtained by SE changing the value of $\r$ in a Rademacher-Bernoulli prior for data generated by a planted SK model with $\D_0=0.8$
and assuming a noise channel with $\D=0.5$. The point where $M=Q$ and the
MSE is minimized at $\r = 0.623$.
}
\label{fig:M_Q_MSE_rho_neq_rho0}
\end{figure}

Let us illustrate what we found with a simple example.
The data $Y$  are generated by the planted SK model with $\D_0 = 0.8$.
In this case the Bayes-optimal inference, knowing the correct prior and noise, would get $\text{MSE} = 0.776$.
Consider now instead a situation in which one (wrongly!) assumes that the variance of the noise is $\D=0.5$ but
also, being not sure about the correct prior, assume a more general Rademacher-Bernoulli prior
\begin{align}
 P(x) = \frac{\r}{2} \delta \left( x-1 \right ) + \frac{\r}{2} \delta \left( x+1 \right ) + \left( 1 -\r \right) \delta(x) 
\end{align}
with some unknown parameter $\r$.
In Fig.~\ref{fig:M_Q_MSE_rho_neq_rho0} we show the values of $M$, $Q$ and MSE obtained by the SE Eqs.~(\ref{eq:SE_RS_M})-(\ref{eq:SE_RS_Sigma}) 
as a function of $\r$ while $\D$ is kept fixed at $0.5$.
Taking $\r = 1$, which corresponds to the correct prior, we obtain $Q \simeq 0.521$, $M\simeq 0.29$ and $\text{MSE} \simeq 0.94$, considerably worse than the Bayes-optimal error.
The maximum $M$ is obtained for $\r \simeq 0.95$, where $Q \simeq 0.48$ and $\text{MSE} \simeq 0.90$.
But, if we look to restore the optimality condition $M=Q$, we are led to accept the value $\r = 0.623$:
here $M=Q=0.224$ and $\text{MSE}=0.776$, equal to the value obtained by Bayes-optimal inference.

The above observations leads us to a hypothesis: Any combination of
the values of parameters that restores the condition $M=Q$, achieves the same performances of the Bayes-optimal setting.
We tested this hypothesis in several settings of the low-rank matrix estimation problem, 
including sparse cases with first order transition~\cite{LesKrzZde17}.
While we always found it true, the underlying reason for this eludes
us and is left for future work. 
The optimality from the restoration of the $M=Q$ condition is relevant in cases in which 
we assume a wrong functional form of prior, as in the previous example, 
so that it is indeed impossible to actually reduce to the Bayes-optimal setting.
%Furthermore, this property will be central in the use of the generalized AMP approach that we will introduce in Sec.~\ref{sec:the_cavity_approach_1RSB}.

%[SHALL WE CHECK IN A CASE WITH PRIOR WITH NONZERO MEAN??]

Note that
it is nontrivial to turn the above observation into an algorithmic procedure.
A common parameter learning procedure is to maximize the Bethe free energy of the model, Eq.~(\ref{eq:freeEnergyBethe}).
This procedure gives asymptotically the Bayes-optimal parameters, if
learning of the exact prior and noise is possible.
Nevertheless, if the functional form of the prior and noise is incorrect, this procedure does not return the optimal values of the parameters -
that, in our hypothesis, would be the ones associated with the restoration of the $M=Q$ condition. 
In the previous example, the Bethe free energy is monotonically increasing in $[0,1]$ and has a local maximum only at $\r \simeq 1.26$, 
where the prior is not a well-defined probability distribution and the
MSE is larger than 1, cf. Fig.~\ref{fig:M_Q_MSE_rho_neq_rho0}. 
It is interesting to look at the RS eigenvalue Eq.~(\ref{eq:repliconRS}) for this solution, that is associated with the point-wise convergence of AMP.
The eigenvalue is positive for $\r \lesssim 0.90 $, and in particular is positive for the value $\r \simeq 0.623$, where the optimality condition $M=Q$ is restored.
Moreover, for very low $\r$, in this case for $\r \lesssim 0.42$, the only solution is the trivial $M=Q=0$, $\text{MSE}=1$.
These two extremes give the finite interval $0.42 \lesssim \r \lesssim 0.90$ in which one should indeed look to find the optimal $\r$. 
We will discuss further about this observation in Sec.~\ref{sec:use_of_ASP} about the use of ASP.

%%%%%%%%%%%%%%%%%%%%%%%%%%%%%%%%%%%%%%%%%%%%%%%%%%%%%%%%%%%%%%%%%%%%%%%%%%%%%%%%%%%%

\section{The Approximate Survey Propagation: a 1RSB version of AMP}
\label{sec:the_cavity_approach_1RSB}

AMP is an established approach to analyze systems with a ferromagnet-like transition,
where one expects to have just two possible fixed-points of the iterations. 
There are, however, situations where there exists a huge number of fixed-points for the AMP equations. 
In particular, we have shown in the previous section that in the low-rank matrix estimation problem with mismatching models 
there is a region where the replica symmetry is broken and the AMP
algorithm does not converge. % and a glass transition appears.
%This situation is the same as the one arising when there is a glass transition. 
In this case one needs to use the cavity method in conjunction with
a \textit{replica symmetry breaking} (RSB) approach, as was introduced by M\'ezard and Parisi \cite{MP01,mezard2003cavity,MPZ02}. 
In the following we show how this approach can be carried out for the low-rank matrix estimation problem
and how this provides a systematic method to deal with mismatching models in a natural way.
Note that
%, while we consider mismatching models, 
we need to insist on the notion of independence of the noise elements,
that is an essential for belief-propagation-based approaches to work.

\subsection{Derivation of the ASP algorithm for the low-rank matrix estimation problem}

We derive here the 1-step replica symmetry breaking (1RSB) approximate massage passing, 
that we call \textit{approximate survey propagation} (ASP) algorithm,
for the low-rank matrix estimation problem.
To derive the correct equations in this case, let us consider a
replicated inference problem, basically turning the method of Monasson
\cite{Mo95} into a message passing algorithm. 
Given the matrix $Y_{ij}$ and $a=1,\ldots,s$, being $s$ the number of \textit{real replicas}, we assume that
\begin{align}
Y_{ij} = \frac{x_i^{(a)}x_j^{(a)}}{\sqrt{N}}+\xi_{ij}^{(a)}\, ,
\end{align}
where $\xi_{ij}^{(a)}$ are independent Gaussian noises with zero mean and variance $\Delta$.
The partition function becomes
\begin{align}
Z_{\text{rep}} \left( Y \right) = \int\left[ \prod_{a=1}^s
  \prod_{i=1}^N \de x_i^{(a)}
  P_X(x_i^{(a)})\right]\exp\left[\sum_{i\leq j}\sum_{a=1}^s
  g\left(Y_{ij}| w^{(aa)}_{ij} \right)\right]\, ,
\label{eq:rep_Z}
\end{align}
where $w^{(ab)}_{ij} \equiv x_i^{(a)} x_j^{(b)} / \sqrt{N} $
and $P_X(x)$ and $g$ have the same definition as the previous sections, cf. Eq.~(\ref{eq:posterior}).
If $s=1$ we get back the usual inference problem of Eq.~(\ref{eq:Y}).

In order to evaluate the marginals over the variables $\{x_i^{(a)}\}$, we need to write the BP equations for the replicated system.
In the following let us omit normalization factors when they are irrelevant (they can be determined a posteriori).
The BP equations are
\begin{align}
m_{i\to ij}(\underline x_i) &\sim \left[\prod_{a=1}^s
                              P_X(x_i^{(a)})\right] \prod_{k \neq j}
                              m_{ki\to i}(\underline x_i) \, ,\\
m_{ij\to i}(\underline x_i)&\sim \int\de \underline  x_j m_{j\to
                             ij}(\underline x_j)\exp\left[\sum_{a=1}^s
                             g\left(Y_{ij}|
                             w^{(aa)}_{ij}\right)\right]\, ,
\label{eq:BPeqs}
\end{align}
where we have introduced the notation $\underline x_i=\{x_i^{(1)}, \ldots, x_i^{(s)}\}$. % and we have neglected the normalization factors.
This follows directly from the factor graph represented in Fig.~\ref{fig:factor_graph_rep}.

\begin{figure}
\centering
\includegraphics[scale=1.]{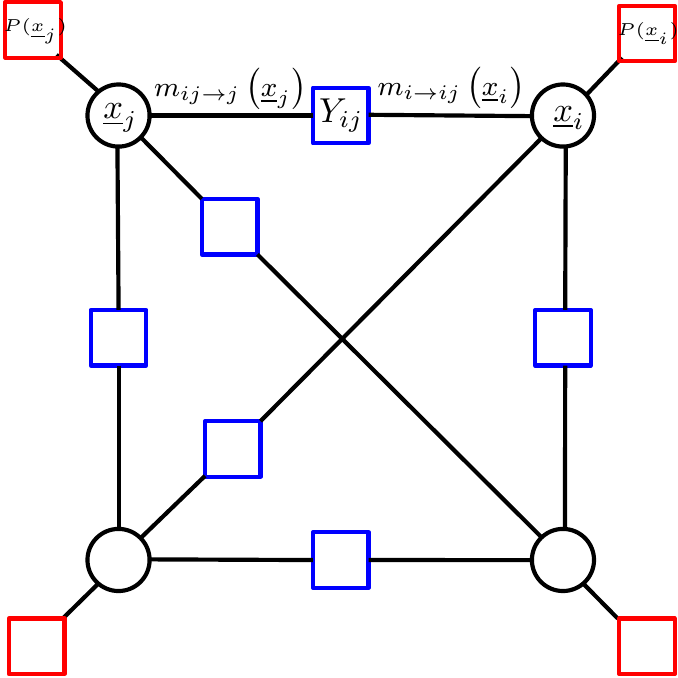}
\caption{The factor graph representing the structure of the replicated posterior measure whose partition function is described in Eq.~(\ref{eq:rep_Z}). 
The variable nodes are replicated variables $\underline x_i \equiv\{x_i^{(1)},\ldots,x_i^{(s)}\}$.}
\label{fig:factor_graph_rep}
\end{figure}

At this point we introduce a 1RSB parametrization for the cavity distributions 
\begin{align}
m_{j\to ij}(\underline x_j) \sim 
\int \de h \, \exp \left[-  \frac{1}{2\Delta^{(0)}_{j\to ij}}\left(h-\hat x_{j\to ij}\right)^2 \right]
\prod_{a=1}^s \exp \left[ -  \frac{1}{2\Delta^{(1)}_{j\to ij}}\left(x_j^{(a)}-h\right)^2 \right] \:.
\label{eq:1RSB_ansatz_cavity}
\end{align}
%\begin{align}
%m_{j\to ij}(\underline x_j) \sim 
%\int \de h \, \exp \left[-  \left(h-\hat x_{j\to ij}\right)^\top \frac{1}{2\Delta^{(0)}_{j\to ij}}\left(h-\hat x_{j\to ij}\right) \right]
%\prod_{a=1}^s \exp \left[ - \left(x_j^{(a)}-h\right)^\top \frac{1}{2\Delta^{(1)}_{j\to ij}}\left(x_j^{(a)}-h\right) \right] \:.
%\label{eq:1RSB_ansatz_cavity}
%\end{align}
The form of this ansatz elucidates why we are considering a replicated system:
if the posterior measure Eq.~(\ref{eq:posterior}) develops a 1RSB structure, 
the phase space of the solutions of the inference problem gets clustered in exponentially many basins of solutions \cite{MPV87}.
Therefore, if we consider a set of $s$ \emph{real} copies of the same inference problem and we infinitesimally couple them, 
they will acquire a finite probability to fall inside the same cluster of solutions. 
This line of reasoning is exactly the same as the one considered to describe the real replica approach for glasses \cite{Mo95,CKPUZ17} and 
the ansatz of Eq.~(\ref{eq:1RSB_ansatz_cavity}) reproduces the infinite dimensional caging ansatz (at the 1RSB level) for hard spheres in infinite dimension \cite{CKPUZ14, CKPUZ13}.

The ansatz of Eq.~(\ref{eq:1RSB_ansatz_cavity}) allows for three relevant correlation functions
\begin{align}
\nonumber
\left\langle x_j^{(a)} \right\rangle &= \hat x_{j\to ij}  \, ,
\\
\left.\left\langle x_j^{(a)}x_j^{(b),\top} \right\rangle\right|_{a\neq
  b} &= \Delta_{j\to ij}^{(0)} + \hat x_{j\to ij} \, \hat x_{j\to
       ij}^\top\, ,
\label{eq:correlations}
\\
\left\langle x_j^{(a)}x_j^{(a),\top} \right \rangle &= \Delta_{j\to
                                                      ij}^{(1)}+\Delta_{j\to
                                                      ij}^{(0)} + \hat
                                                      x_{j\to ij} \,
                                                      \hat x_{j\to
                                                      ij}^\top\, ,
\nonumber
\end{align}
where the averages are taken with the measure defined in Eq.~(\ref{eq:1RSB_ansatz_cavity}).
The last two lines give access to the correlation between cavity messages from one copy to another.
Note that if $\Delta^{(0)}_{j\to ij}=0$ different replicas become uncorrelated. 
Therefore in the limit $\Delta^{(0)}_{j\to ij}\to 0$ one recovers the replica symmetric AMP. 
Equivalently, if one fixes $s=1$ the ansatz reduces naturally to the RS parametrization for the cavity messages and one gets back again to the AMP algorithm.

%In the following, to simplify the notation, let us consider the rank one case, sometimes also called spiked low-rank matrix estimation problem. The extension to $r \neq 1$ is immediate.
Given the 1RSB ansatz, we can plug Eq.~(\ref{eq:1RSB_ansatz_cavity}) in Eq. (\ref{eq:BPeqs}) to get
\begin{align}
m_{ij\to i}(\underline x_i)  \sim  &
\int \de h e^{-\frac{1}{2\Delta^{(0)}_{j\to ij}}\left(h-\hat x_{j\to ij}\right)^2}
\prod_{a=1}^s \int\de   x_j^{(a)} \exp\left[-\frac{1}{2\Delta^{(1)}_{j\to ij}}\left(x_j^{(a)}-h\right)^2
+g\left(Y_{ij}| w^{(ab)}\right)\right] \nonumber
\\
 \sim & \exp\left[\frac{1}{\sqrt N} S_{ij}\hat x_{j\to ij}\sum_{a=1}^s x_j^{(a)}
-\frac{1}{2N} \hat R_{ij}\left(\Delta^{(1)}_{j\to ij}\Delta^{(0)}_{j\to ij}+\left(\hat x_{j\to ij}\right)^2\right)\sum_{a=1}^s \left(x_j^{(a)}\right)^2\right.
\nonumber
\\
& \qquad +\left.\frac{1}{2N}S_{ij}^2 \Delta^{(0)}_{j\to
  ij}\left(\sum_{a=1}^sx_j^{(a)}\right)^2+\frac{1}{2N}S_{ij}^2
  \Delta^{(1)}_{j\to ij}\sum_{a=1}^s\left(x_j^{(a)}\right)^2\right]\, ,
\end{align}
where the matrices $S_{ij}$ and $\hat R_{ij}$ are the same as introduced in Eq.~(\ref{eq:definition_S_R}).
%\begin{align}
%S_{ij}=\left.\frac{\partial g(Y_{ij}|w)}{\partial w}\right|_{w=0}\ \ \ \ \ \ \ \ \ \ \hat R_{ij}=-\left.\frac{\partial^2 g(Y_{ij}|w)}{\partial w^2}\right|_{w=0}
%\end{align}
Plugging this result inside the first one of Eqs.~(\ref{eq:BPeqs}) we get
\beq
m_{i\to ij}(\underline x_i)\sim 
\left[\prod_{a=1}^s P_X(x_i^{(a)})\right]\exp\left[ T_{i\to ij}
  \sum_{a=1}^s x_i^{(a)} -\frac{1}{2} V_{i\to ij}^{(1)} \sum_{a=1}^s
  \left(x_i^{(a)}\right)^2+\frac 12 V_{i\to ij}^{(0)}
  \left(\sum_{a=1}^s x_i^{(a)}\right)^2\right]\, ,
\nonumber
\eeq
where we have defined
\begin{align}
T_{i\to ij} &= \frac{1}{\sqrt N} \sum_{k \neq j} S_{ik}\hat x_{k\to
              ik}\, ,
\label{eq:1RSB_BF_T}
\\
V^{(1)}_{i\to ij} &=\frac 1N \sum_{k \neq j}\left[\hat
                    R_{ik}\left(\Delta^{(1)}_{k\to
                    ik}+\Delta^{(0)}_{k\to ik}+\left(\hat x_{k\to
                    ik}\right)^2\right)-   S_{ik}^2 \Delta_{k\to
                    ik}^{(1)}  \right]\, ,\\
V^{(0)}_{i\to ij} &= \frac 1N \sum_{k \neq j} S_{ik}^2\Delta_{k\to ik}^{(0)}\:.
\label{eq:1RSB_BF_V0}
\end{align}
We can now close the equations.
We have
\begin{align}
\left \langle x_i^{(a)}\right \rangle &= \frac1s
                                        \frac{\partial}{\partial
                                        T_{i\to ij}}f_{\text{in}}
                                        \left[T_{i\to
                                        ij},V^{(1)}_{i\to
                                        ij},V^{(0)}_{i\to ij}\right]\,
  ,\\
\left \langle\left( x_i^{(a)}\right)^2\right \rangle &= -\frac{2}{s}
                                                       \frac{\partial}{\partial
                                                       V^{(1)}_{i\to
                                                       ij}}
                                                       f_{\text{in}}\left[T_{i\to
                                                       ij},V^{(1)}_{i\to
                                                       ij},V^{(0)}_{i\to
                                                       ij}\right]\, ,\\
\left \langle x_i^{(a)}x_i^{(b)}\right \rangle &= \frac{2}{s(s-1)}
                                                 \left[\frac{\partial}{\partial
                                                 V_{i\to
                                                 ij}^{(0)}}+\frac{\partial}{\partial
                                                 V_{i\to
                                                 ij}^{(1)}}\right]f_{\text{in}}\left[T_{i\to
                                                 ij},V^{(1)}_{i\to
                                                 ij},V^{(0)}_{i\to
                                                 ij}\right]\, ,
\end{align}
where we have introduced the function
\begin{align}
f_{\text{in}}\left[T,V^{(1)},V^{(0)}\right] &= 
\ln \int \de \underline x \left[\prod_{a=1}^sP_X(x^{(a)})\right] \exp\left[T \sum_{a=1}^s x^{(a)}-\frac 12 V^{(1)}\sum_{a=1}^s \left(x_i^{(a)}\right)^2+\frac 12 V^{(0)}\left(\sum_{a=1}^s x_i^{(a)}\right)^2\right]
\nonumber
\\
&=\ln \int \de h e^{-\frac 12 V^{(0)}h^2} \sqrt{\frac{V^{(0)}}{2\pi}} \left\{ \int \de x P_X(x) \exp\left[-\frac 12 V^{(1)}x^2+ (T+V^{(0)}h)x\right] \right\}^s\:.
\end{align}
Using Eqs.~(\ref{eq:correlations}), we obtain 
\begin{align}
\hat x_{i\to ij} &= \frac1s \frac{\partial}{\partial T_{i\to
                   ij}}f_{\text{in}} \left[T_{i\to ij},V^{(1)}_{i\to
                   ij},V^{(0)}_{i\to ij}\right] \, ,
\label{eq:1RSB_BF_x}
\\
\Delta^{(0)}_{i\to ij}&=\frac{2}{s(s-1)} \left[\frac{\partial}{\partial V^{(1)}_{i\to ij}}+\frac{\partial}{\partial V^{(0)}_{i\to ij}}\right]f_{\text{in}} \left[T_{i\to ij},V^{(1)}_{i\to ij},V^{(0)}_{i\to ij}\right]
-\frac{1}{s^2}\left(\frac{\partial }{\partial T_{i\to ij}}
                        f_{\text{in}} \left[T_{i\to ij},V^{(1)}_{i\to
                        ij},V^{(0)}_{i\to ij}\right] \right)^2\, ,\\
\Delta^{(1)}_{i\to ij}&=2\left[\frac{1}{s(1-s)}\left(\frac{\partial}{\partial V^{(1)}_{i\to ij}}+\frac{\partial}{\partial V^{(0)}_{i\to ij}}\right) -\frac{1}{s}\frac{\partial}{\partial V^{(1)}_{i\to ij}} \right]
 f_{\text{in}} \left[T_{i\to ij},V^{(1)}_{i\to ij},V^{(0)}_{i\to ij}\right] \:.
 \label{eq:1RSB_BF_delta1}
\end{align}
The Eqs.~(\ref{eq:1RSB_BF_x})-(\ref{eq:1RSB_BF_delta1}) together with Eqs.~(\ref{eq:1RSB_BF_T})-(\ref{eq:1RSB_BF_V0})
constitute the 1RSB-BP equations, within the Gaussian ansatz of
Eq.~(\ref{eq:1RSB_ansatz_cavity}), and with Parisi parameter $s$.

%\begin{align}
%T_{i\to ij} &= \frac{1}{\sqrt N} \sum_{k \neq j} S_{ik}\hat x_{k\to ik}\\
%V^{(1)}_{i\to ij} &=\frac 1N \sum_{k \neq j}\left[\hat R_{ik}\left(\Delta^{(1)}_{k\to ik}+\Delta^{(0)}_{k\to ik}+\left(\hat x_{k\to ik}\right)^2\right)-   S_{ik}^2 \Delta_{k\to ik}^{(1)}  \right]\\
%V^{(0)}_{i\to ij} &= \frac 1N \sum_{k \neq j} S_{ik}^2\Delta_{k\to ik}^{(0)}\\
%\hat x_{i\to ij} &= \frac1s \frac{\partial}{\partial T_{i\to ij}}f_{\text{in}}\\
%\Delta^{(0)}_{i\to ij}&=\frac{2}{s(s-1)} \left[\frac{\partial}{\partial V^{(1)}_{i\to ij}}+\frac{\partial}{\partial V^{(0)}_{i\to ij}}\right]f_{\text{in}} -\frac{1}{s^2}\left(\frac{\partial f_{\text{in}}}{\partial T_{i\to ij}}\right)^2\\
%\Delta^{(1)}_{i\to ij}&=2\left[\frac{1}{s(1-s)}\left(\frac{\partial}{\partial V^{(1)}_{i\to ij}}+\frac{\partial}{\partial V^{(0)}_{i\to ij}}\right) -\frac{1}{s}\frac{\partial}{\partial V^{(1)}_{i\to ij}} \right]f_{\text{in}}\:.
%\label{1RSB-BP}
%\end{align}
%
Finally, the moments of the marginal distributions are obtained by introducing
\begin{align}
T_{i} &= \frac{1}{\sqrt N} \sum_{k } S_{ik}\hat x_{k\to ik}\, ,\\
V^{(1)}_{i} &=\frac 1N \sum_{k}\left[\hat
              R_{ik}\left(\Delta^{(1)}_{k\to ik}+\Delta^{(0)}_{k\to
              ik}+\left(\hat x_{k\to ik}\right)^2\right)-   S_{ik}^2
              \Delta_{k\to ik}^{(1)}  \right]\, ,\\
V^{(0)}_{i} &= \frac 1N \sum_{k} S_{ik}^2\Delta_{k\to ik}^{(0)}\, ,
\end{align}
through which we have
\begin{align}
\nonumber
\hat x_{i} \equiv \frac 1s\sum_{a=1}^s \left[ x_i^{(a)} \right] 
=& \frac1s \frac{\partial}{\partial
   T_{i}}f_{\text{in}}\left[T_i,V^{(1)}_i,V^{(0)}_i\right] \, ,
\\
\nonumber
\left.\left[x_i^{(a)}x_i^{(b)}\right]\right|_{a\neq b}- \hat x_i^2
=\Delta^{(0)}_{i}
=& \frac{2}{s(s-1)} \left[\frac{\partial}{\partial V^{(1)}_{i}}+\frac{\partial}{\partial V^{(0)}_{i}}\right]f_{\text{in}}\left[T_i,V^{(1)}_i,V^{(0)}_i\right] 
\\
&-\frac{1}{s^2}\left(\frac{\partial
  f_{\text{in}}\left[T_i,V^{(1)}_i,V^{(0)}_i\right]}{\partial
  T_{i}}\right)^2\, ,
\\
 \left[\left(x_i^{(a)}\right)^2\right]-\left.\left[x_i^{(a)}x_i^{(b)}\right]\right|_{a\neq b} 
= \Delta^{(1)}_{i}
=& 2\left[\frac{1}{s(1-s)}\left(\frac{\partial}{\partial V^{(1)}_{i}}+\frac{\partial}{\partial V^{(0)}_{i}}\right) 
 -\frac{1}{s}\frac{\partial}{\partial V^{(1)}_{i}} \right]f_{\text{in}}\left[T_i,V^{(1)}_i,V^{(0)}_i\right]
 \, ,\nonumber
\end{align}
where the square brackets $\left[ \cdot \right]$ indicate the average over the replicated posterior measure defined in Eq.~(\ref{eq:rep_Z}).
The complexity of the 1RSB-BP algorithm described in
Eqs.~(\ref{eq:1RSB_BF_T})-(\ref{eq:1RSB_BF_V0}), (\ref{eq:1RSB_BF_x})-(\ref{eq:1RSB_BF_delta1}) 
can be reduced by working directly with the moments of the real marginals instead using the moments of the cavity marginals.
This \emph{TAPification} procedure is well known in inference problems (cf. Sec. \ref{sec:lowRAMP} and references therein). 
%The interested reader can find the details in Refs.~\cite{KZ16,LesKrzZde17}.
The result is the ASP algorithm for the low-rank matrix estimation problem:
\begin{align}
T_{i}^t =& \frac{1}{\sqrt N} \sum_{k } S_{ik}\hat x_{k}^t 
- \frac 1N \hat x_i^{t-1} \sum_{k}S_{ik}^2\left(\Delta^{(1),t}_k + s
           \Delta^{(0),t}_k\right) \, , \label{eq:TT}
\\
V^{(1),t}_{i} =&\frac 1N \sum_{k}\left[ \hat R_{ik}\left(\Delta^{(1),t}_{k}+\Delta^{(0),t}_{k}+\left(\hat x_{k}^t\right)^2\right)
-   S_{ik}^2 \Delta_{k}^{(1),t}  \right] \, ,\label{eq:VV1}
\\
V^{(0),t}_{i} =& \frac 1N \sum_{k} S_{ik}^2\Delta_{k}^{(0),t}\, , \label{eq:VV0}
\\
\hat x_{i}^{t+1} =& \frac1s \frac{\partial}{\partial
                    T_{i}}f_{\text{in}}\left[T_i^t,V^{(1),t}_i,V^{(0),t}_i\right]
                    \, , \label{eq:MM}
\\
\Delta^{(0),t+1}_{i}=&\frac{2}{s(s-1)} \left[\frac{\partial}{\partial V^{(1)}_{i}}+\frac{\partial}{\partial V^{(0)}_{i}}\right]f_{\text{in}}\left[T_i^t,V^{(1),t}_i,V^{(0),t}_i\right] 
-\frac{1}{s^2}\left(\frac{\partial
                       f_{\text{in}}\left[T_i^t,V^{(1),t}_i,V^{(0),t}_i\right]}{\partial
                       T_{i}}\right)^2\, , \label{eq:DD0}
\\
\Delta^{(1),t+1}_{i} =& 2\left[\frac{1}{s(1-s)}\left(\frac{\partial}{\partial V^{(1)}_{i}}+\frac{\partial}{\partial V^{(0)}_{i}}\right) 
-\frac{1}{s}\frac{\partial}{\partial V^{(1)}_{i}}
                        \right]f_{\text{in}}\left[T_i^t,V^{(1),t}_i,V^{(0),t}_i\right]\, . \label{eq:DD1}
\end{align}
Strictly speaking we have different algorithms for different values of the Parisi parameter $s$. 
We will comment more about the use of this parameter in Sec.~\ref{sec:optimality_and_mismatching}.
Note that for $s=1$ the equations depend only on $\sigma_i = \Delta^{(1)}_{i} + \Delta^{(0)}_{i} $ and the algorithm 
reduces to the Low-RAMP Eqs.~(\ref{eq:B_AMP})-(\ref{eq:sigma_AMP})
with $B_i = T_i$ and $A_i = V^{(1)}-V^{(0)}$.
Observe also that, similarly to what happens in Low-RAMP, the equations can be further simplified replacing $S_{ij}^2$ by its mean, without changing the leading order in $N$.
This simplification allows to express Eqs.~(\ref{eq:TT})-(\ref{eq:VV0}) simply as matrix multiplications, cf. Sec. \ref{sec:rigorous_1RSB}.
The introduction of a RSB structure in the algorithm allows ASP to  
converge point-wise and to obtain a low MSE also in regions where the RS solution is unstable
and AMP does not converge point-wise, cf. Fig.~\ref{fig:AMP_vs_1RSB}.

The fixed-points of the ASP equations are stationary points of the following free energy:
\begin{align}
 \Phi^{\text{1RSB}}_{\text{Bethe}} =&
 \max_{T_i, V^{(0)}_i, V^{(1)}_i} 
 \sum_{1 \leq i \leq N} \frac{1}{s} f_{\text{in}} \left( T_i,V^{(0)}_i,V^{(1)}_i \right)
 - T_i \hat x_i  + \frac{1}{2} \left( V^{(1)}_i - V^{(0)}_i \right)
                                      \D^{(1)}_i + \frac{1}{2} \left(
                                      V^{(1)}_i - s V^{(0)}_i \right)
                                      \left( \hat x_i^2 + \D^{(0)}_i
                                      \right) \nonumber
 \\
 & + \frac{1}{2}   \sum_{1 \leq i , j \leq N} \biggl\{  
 \frac{S_{ij}}{\sqrt{N}} \hat x_i \hat x_j  
 -  \frac{\hat R_{ij}}{2N} \left( \hat x_i^2 + \Delta^{(1)}_{i} + \Delta^{(0)}_{i} \right) \left( \hat x_j^2 + \Delta^{(1)}_{j} + \Delta^{(0)}_{j} \right)  +\nonumber
 \\
& \qquad \qquad \qquad 
+   \frac{S_{ij}^2}{2N}  \left[ 2 \hat x_i^2 \left( \Delta^{(1)}_{j} + s \Delta^{(0)}_{j} \right) 
 - \left( \Delta^{(1)}_{i} +s \Delta^{(0)}_{i} \right) \left( \Delta^{(1)}_{j} +s \Delta^{(0)}_{j} \right)  \right] \biggl\}  \label{eq:Bethe_1RSB}
\end{align}
This free energy is known in statistical physics literature as the 1RSB potential and is useful to compare different fixed-points at same
$s$, in the same spirit of the RS case. 
Note that when one tries to compare fixed-points of ASP that are associated to different values of $s$, the extremization of this free energy does not 
correspond to the minimum MSE, cf. Sec.~\ref{sec:use_of_ASP}.

\begin{figure}[h!]
\centering
\includegraphics[width=0.495\columnwidth]{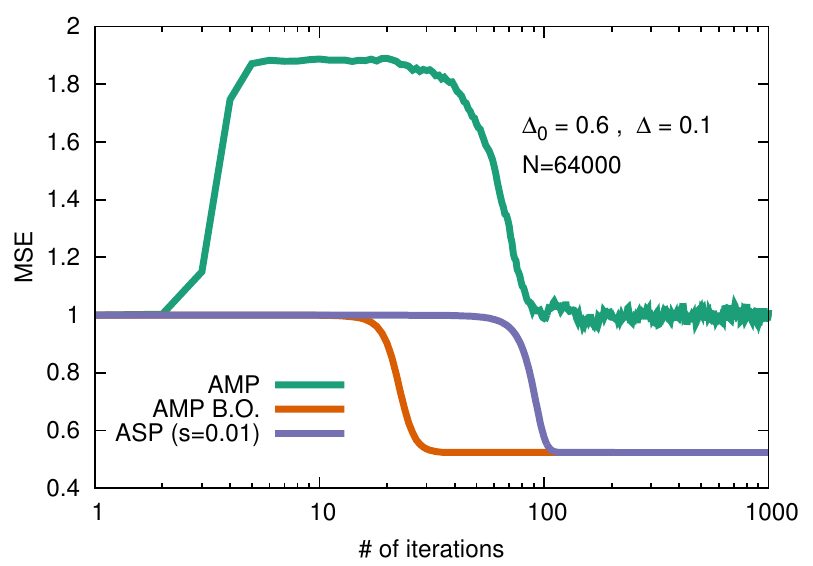} 
\includegraphics[width=0.495\columnwidth]{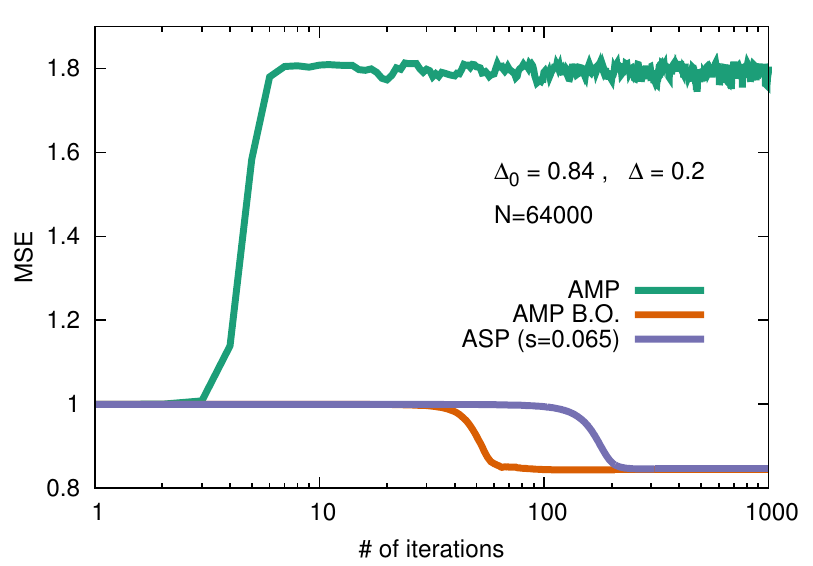} 
\caption{MSE under iterations of ASP for $N=64000$ in the planted SK model 
with $\Delta_0 = 0.6$, $\Delta=0.1$, $s=0.01$ (left) and $\Delta_0 = 0.84$, $\Delta=0.2$, with $s=0.065$ (right).
We compare ASP with AMP run on the same data with the same $\Delta$ as ASP (green lines) and AMP run in the Bayes-optimal setting $\Delta=\Delta_0$ (orange lines).
In these situations, well inside the RS instability region, AMP does not converge point-wise and the associated MSE is typically equal or larger than 1.
The ASP algorithm instead converges point-wise with a MSE close to the AMP run in the Bayes-optimal setting ($\Delta=\Delta_0$), 
at least for some value of $s$, cf. Sec.~\ref{sec:use_of_ASP}.
}
\label{fig:AMP_vs_1RSB}
\end{figure}

\subsection{The 1RSB state evolution equations}

The ASP algorithm can be implemented on a single instance of the inference problem.
%Here we want to address the \emph{typical} performances of the algorithm once we extract the instances of the matrix $Y$ from a given distribution.
Moreover, in the same lines as for AMP (cf. Sec. \ref{sec:state_evolution_RS}) it is possible to obtain the state evolution equations for ASP.
One assumes that $Y$ is generated through the process of first extracting
the signal $\{x_i^{(0)} \}$ from a probability distribution $ P_0 ( \{x_i^{(0)} \}) =\prod_{i=1}^N P_0 (x_i^{(0)})$ and,
then, it is measured through a Gaussian channel of zero mean and variance $\Delta_0$ so that
$P(Y_{ij}) = \exp\left[g^{(0)}\left(Y_{ij}| w_{ij}^{(0)} \right)\right]$ with $g^{(0)}(Y|w)$ given in Eq.~(\ref{eq:def_output_channel}).
%\beq
%Y_{ij} = \frac{x^{(0)}_ix^{(0)}_j}{\sqrt{N}} + \xi^{(0)}_{ij} \ \ \ \ \ \ \forall i\leq j\:.
%\eeq
%Therefore the probability distribution of the matrix elements $Y_{ij}$ is given by 
%\beq
%P(Y_{ij}) = \exp\left[g^{(0)}\left(Y_{ij}| \frac{x_i^{(0)}x_j^{(0)}}{\sqrt N}\right)\right] \ \ \ \ \ \ \ \ \ \ \ g^{(0)}(Y|w) = -\frac{1}{2}\ln(2\pi \Delta_0) - \frac{1}{2\D_0}(Y-w)^2 \:.
%\eeq
Once again, central limit theorem assures that the average over $Y$ of the variables $T_{i}$ of Eq.~(\ref{eq:TT}) become Gaussian  with 
%The averages over $Y$ can be performed in a simple way. Central limit theorem gives that the variables $T_{i}$ of Eq.~(\ref{eq:TT}) become Gaussian  with 
\begin{align}
\overline{T_i^t}&= \frac{M^t}{\D} x_i^{(0)}
\, , &
\overline{\left(T_i^t \right)^2}&= \frac{\D_0}{\D^2}Q^t
                                  +\left(\frac{M^t}{\D}x_i^{(0)}\right)^2\, ,
\end{align}
where $M$ and $Q$ are the usual order parameters (cf. Eq.~(\ref{eq:def_MtQ_AMP}))
\begin{align}
& M^t =\frac 1N \sum_{k=1}^N\hat x_k^t x_k^{(0)}
\, , &
& Q^t =\frac 1N \sum_{k=0}^N \left(\hat x_k^t \right)^2\:.
\label{def_MtQ} 
\end{align}
Instead $V_i^{(1)}$ and $V_i^{(0)}$, as defined in Eq.~(\ref{eq:VV1})-(\ref{eq:VV0}), at leading order in $N$ are concentrated around their mean value
\begin{align}
& V^{(0),t} =\overline{V_i^{(0),t}} = \frac{\D_0}{\D^2}\D^{(0),t}  
 \, , &
& V^{(1),t} =\overline{V_i^{(1),t}}= \frac 1\D \left(\D^{(1),t}+\D^{(0),t}+Q^t
  \right)-\frac{\D_0}{\D^2}\D^{(1),t}\, ,
\end{align}
and we have defined the order parameters
\begin{align}
& \D^{(0),t} = \frac 1N \sum_{k=1}^N \Delta_k^{(0),t}  
\, , &
& \D^{(1),t} = \frac 1N \sum_{k=1}^N \Delta_k^{(1),t}\:.
\label{eq:definition_delta0_delta1}
\end{align}
Starting from this, the parameters $M$, $Q$, $\Delta^{(1)}$ and $\Delta^{(0)}$ are fixed self-consistently using Eqs.~(\ref{eq:MM})-(\ref{eq:DD1}).
Let us consider $M$.  From Eq.~(\ref{eq:MM}) we get that
\begin{align}
\nonumber
M^{t+1} =& \frac 1N \int \left[\prod_{i=1}^N\de x_i^{(0)}P_0 \left(x^{(0)}_i\right)\right]\left(\sum_{k=1}^N\hat x_k^{t+1} \, x_k^{(0)}\right)
\\
\label{eq:M1}
=& \frac{1}{s} \int \de \underline x^{(0)} P_0 (\underline  x^{(0)})  \cdot \\
& \cdot \left.\frac{\partial}{\partial T} f_{\text{in}} \int \frac{\de w}{\sqrt{2\pi}}e^{-w^2/2}
\left[T^t, \frac 1\D 
\left(\D^{(1),t}+\D^{(0),t}+Q^t \right)-\frac{\D_0}{\D^2}\D^{(1),t}, \, 
\frac{\D_0}{\D^2}\D^{(0),t}\right]\right|_{T=\frac{M^t}{\Delta}x^{(0)}+\sqrt{\frac{\D_0}{\D^2}Q^t}W}\:.
\nonumber
\end{align}
Therefore, defining 
\begin{align}
{\mathbb{E}}_{x^{(0)},W}(A) = \int \frac{\de W}{\sqrt{2\pi}}e^{-W^2/2} \int \de \underline x^{(0)} P_0 (\underline x^{(0)}) A\: 
\end{align}
and
\begin{align}
\label{eq:def_T_SE}
T^t &=\frac{M^t}{\Delta}x^{(0)}+ \sqrt{\frac{\Delta_0
      Q^t}{\Delta^2}}W\, ,\\
V^{(1),t} &= \frac 1\D \left(\D^{(1),t}+\D^{(0),t}+Q^t
            \right)-\frac{\D_0}{\D^2}\D^{(1),t}\, , \label{eq:def_V1_SE}\\
V^{(0),t}&=\frac{\D_0}{\D^2}\D^{(0),t}\, ,
\label{eq:def_V0_SE}
\end{align}
one can rewrite Eq.~(\ref{eq:M1}) as
\begin{align}
M^{t+1} =\frac 1s {\mathbb{E}}_{x^{(0)},W}\left[
  \frac{\partial}{\partial T}
  f_{\text{in}}\left[T^t,V^{(1),t},V^{(0),t}\right] \,
  x^{(0)}\right]\, .
\label{eq:1RSB_SE_M}
\end{align}
Using Eqs.~(\ref{eq:DD0})-(\ref{eq:DD1}) and the definitions given in Eq.~(\ref{def_MtQ})-(\ref{eq:definition_delta0_delta1}) one gets similarly
\begin{align}
Q^{t+1} &=\frac{1}{s^2}
          {\mathbb{E}}_{x^{(0)},W}\left[\left(\frac{\partial}{\partial
          T} f_{\text{in}}\left[T^t,V^{(1),t},V^{(0),t}\right]\right)^2\right]\, ,
\label{eq:1RSB_SE_Q}
\\
\Delta^{(0),t+1}&={\mathbb{E}}_{x^{(0)},W}\left[\frac{2}{s(s-1)}
                  \left[\frac{\partial}{\partial
                  V^{(1)}}+\frac{\partial}{\partial
                  V^{(0)}}\right]f_{\text{in}}\left[T^t
                  ,V^{(1),t},V^{(0),t}\right]
                  -\frac{1}{s^2}\left(\frac{\partial
                  f_{\text{in}}\left[T^t
                  ,V^{(1),t},V^{(0),t}\right]}{\partial
                  T}\right)^2\right]\, ,
\label{eq:1RSB_SE_Delta0}
\\
\Delta^{(1),t+1}&={\mathbb{E}}_{x^{(0)},W}\left[2\left[\frac{1}{s(1-s)}\left(\frac{\partial}{\partial
                  V^{(1)}}+\frac{\partial}{\partial V^{(0)}}\right)
                  -\frac{1}{s}\frac{\partial}{\partial V^{(1)}}
                  \right]f_{\text{in}}\left[T^t,V^{(1),t},V^{(0),t}\right]\right]\, ,
\label{eq:1RSB_SE_Delta1}
\end{align}
that are the state evolution equations of ASP, or 1RSB-SE. 
The solution of the state evolution equations coincides with the result of the replica theory for the 1RSB structure 
provided $\D^{(0)}=q_1-q_0$ , $Q=q_0 $ and $\Delta^{(1)} = q_d-q_1$
(cf. Appendix \ref{sec:equivalence_SE_replica}). 
The 1RSB-SE provides the typical asymptotic behaviour of ASP. 
In Fig.~\ref{fig:ASP_and_SE} we show how single instances of ASP
converge to the 1RSB-SE for large sizes.

Note that we derived the 1RSB eqs.~(\ref{eq:1RSB_SE_M}-\ref{eq:1RSB_SE_Delta1}) as
  a state evolution of the ASP algorithm without using the replica
  trick in any way. We used the aid of real replicas in the derivation
of the ASP algorithm, but we could have simply postulated the algorithm
and derive (or prove, see section \ref{sec:rigorous_1RSB}) its state
evolution anyway. Our approach thus provides a concrete algorithmic
meaning to the 1RSB equations, that can be understood independently of
the replica method. The advantage of this algorithmic
interpretation is that the ASP algorithm follows its state evolution
even if the statistical properties of the model are {\it not} described by 1RSB.

\begin{figure}[h!]
\centering
\includegraphics[width=0.47\columnwidth]{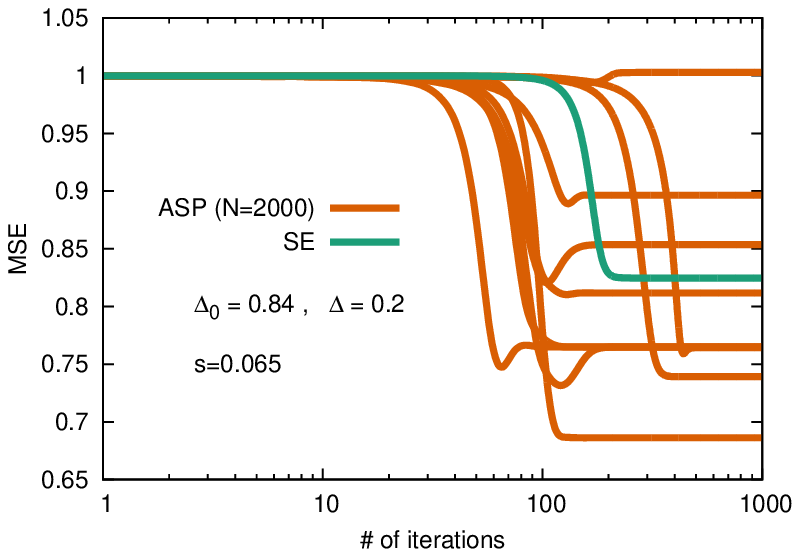} 
\includegraphics[width=0.47\columnwidth]{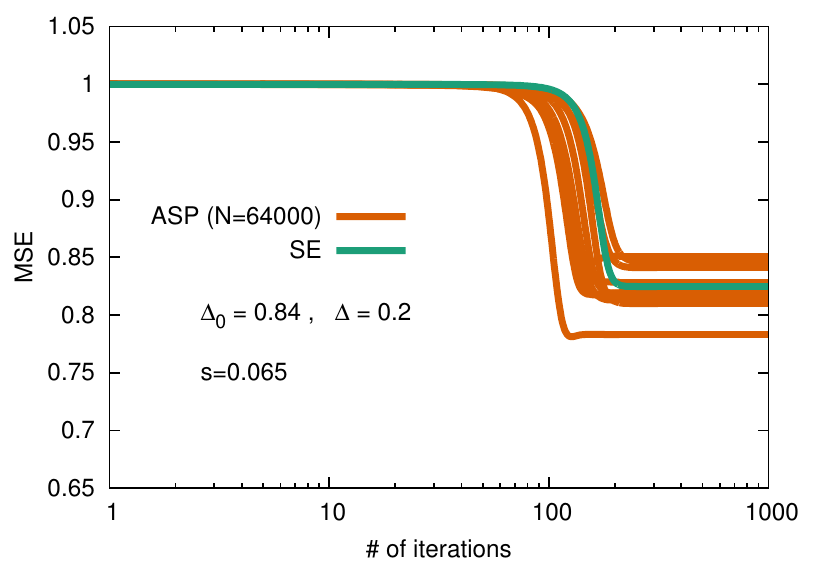} \\
\includegraphics[width=0.47\columnwidth]{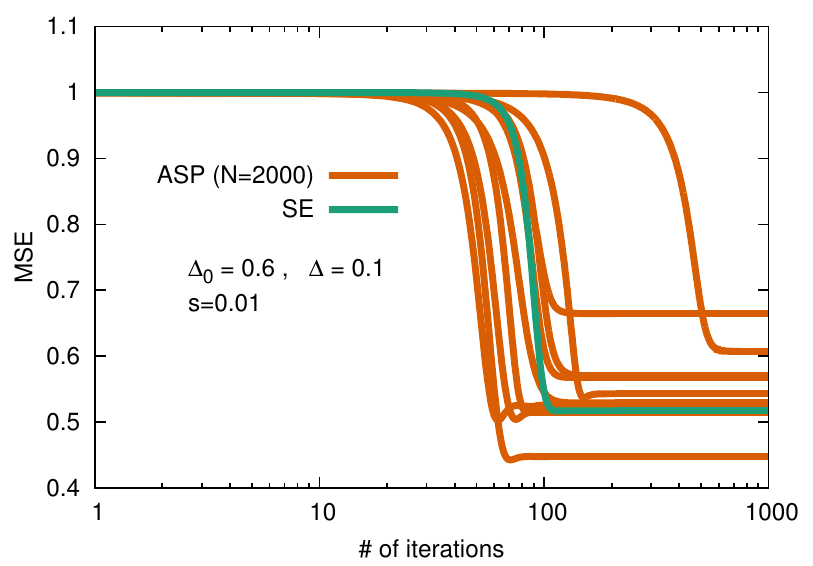} 
\includegraphics[width=0.47\columnwidth]{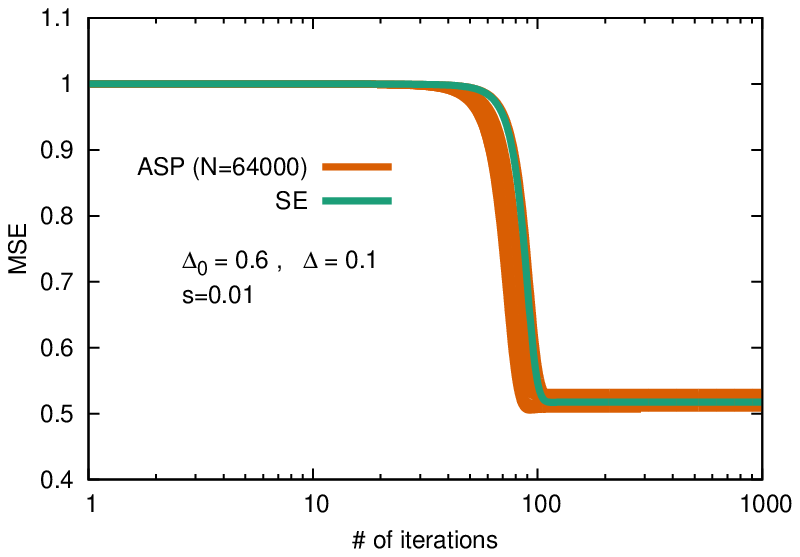} 
\caption{Iterations of ASP for system sizes $N=2000$ (left) and
  $N=64000$ (right) compared with SE for the planted SK model with $\Delta_0 = 0.84$, $\Delta=0.2$, $s=0.065$ (top)
and $\Delta_0 = 0.6$, $\Delta=0.1$, $s=0.01$ (bottom). Note that for these values of $\D_0$ and $\D$ we are in the region where RS solution is unstable, cf. Fig.~\ref{fig:phaseDiagramRS}.
}
\label{fig:ASP_and_SE}
\end{figure}

\subsubsection{The 1RSB free energy and complexity}
The algorithmic interpretation of the 1RSB is not
limited to the fixed point equations, but concerns the free energy as
well. In the same way AMP can be interpreted as extremizing the replica
symmetric Bethe free energy (\ref{eq:freeEnergyBethe}), ASP can be
interpreted as extremizing the 1RSB Bethe
free energy (\ref{eq:Bethe_1RSB}).

The 1RSB-SE equations corresponds to the stationary points of the free
energy,  which we give here without derivation, as it can
  readily be adapted from the RS one:
\begin{align}
 \Phi_{\text{1RSB}} = \max \left\{ \phi_{\text{1RSB}} \left( M,Q,\Delta^{(0)},\Delta^{(1)} \right) , 
 \frac{\partial \phi_{\text{1RSB}}}{\partial M} =  \frac{\partial \phi_{\text{1RSB}}}{\partial Q} 
 =  \frac{\partial \phi_{\text{RS}}}{\partial \Delta^{(0)}} =  \frac{\partial \phi_{\text{RS}}}{\partial \Delta^{(1)}} = 0\right\}
\end{align}
with
\begin{align}
\phi_{\text{1RSB}} \left( M,Q,\Delta^{(0)},\Delta^{(1)} \right) =&
 s \frac{\D_0}{4 \D^2} Q^2 - \frac{1}{2 \D} M^2 - \frac{\D_0 - \D}{4 \D^2} \left( Q + \D^{(0)} + \D^{(1)} \right)^2 +
\nonumber
 \\
& + (1-s) \frac{\D_0 }{4 \D^2} \left( Q + \D^{(0)}  \right)^2
 + \frac{1}{s} {\mathbb{E}}_{x^{(0)},W} \left[ f_{\text{in}} \left(
  T,V^{(0)},V^{(1)} \right) \right]\, ,
\label{eq:1RSB_freeEnergy}
\end{align}
where $T$, $V^{(0)}$ and $V^{(1)}$ are functions of the order parameters as for Eqs.~(\ref{eq:def_T_SE})-(\ref{eq:def_V0_SE}). 
This free energy coincides with the one obtained by replica theory under 1RSB ansatz
and reduces to Eq.~(\ref{eq:RSfreeEnergy}) for $s=1$. % (or equivalently for $\Delta^{(0)}=0$).
From the previous expression one can obtain the value of $s$ that extremizes the free energy.
From the point of view of physics, this value of $s$ would be the one that describes the equilibrium states of the system in a 1RSB phase. 
From the point of view of inference, this value of $s$ does not minimizes the MSE, as we discuss in Sec.~\ref{sec:use_of_ASP}.

In the case the posterior measure develops a 1RSB structure, 
the phase space of configurations becomes clustered in exponentially many basins.
In this situation, the number of basins is counted by the complexity \cite{MPV87}.
The complexity $\Sigma$ as function of the free energy of a single
metastable state can be obtained as the Legendre transform of the 1RSB free
energy of the system, obtaining
\begin{align}
 \Sigma \left( M,Q,\Delta^{(0)},\Delta^{(1)} \right) =&
 s^2 \frac{\D_0}{4 \D^2} \D^{(0)} \left( 2 Q + \D^{(0)} \right)
 - s^2 \frac{\partial}{\partial s} {\mathbb{E}}_{x^{(0)},W} \left[ \frac{1}{s}  f_{\text{in}} \left( T,V^{(0)},V^{(1)} \right) \right]
\label{eq:complexity}
\end{align}
From the point of view of physics, the static solution corresponds to the metastable
states with the lowest free energy and null complexity. 
%For these states $f = \phi$. and one recovers the static stationary condition $\partial_s \phi = 0$. 
%The dynamical behavior of the system is conjectured to be dominated by the large number of
%metastable states with the highest physically acceptable
%free energy, for which the complexity reaches its maximum value.
%From the point of view of inference, the presence of an extensive number of
%metastable states makes the problem more difficult
%and hinders the convergence of AMP.
%exploring directly the posterior measure becomes unfeasibleand AMP does not typically converge. 
The inclusion of the 1RSB structure in ASP allows to analyze
situations with nonzero complexity $\Sigma$.

\subsubsection{Rigorous approach reloaded}
\label{sec:rigorous_1RSB}

Interestingly, we can use the rigorous result for state evolution of
AMP from Sec.~\ref{sec:rigorous_RS} in the present context as well,
and show that ASP follows rigorously a state evolution corresponding
to the 1RSB equations.

\begin{theorem}[State Evolution for Approximate Survey Propagation]
For the ASP algorithm, the empircal averages 
\begin{align}
 M^{t}& = \frac 1N \sum_i \hat x^t_i x^{(0)}_i \,, &Q^{t} &= \frac 1N  \sum_i (\hat x^t_i)^2\,,&\Psi^t &= \frac 1N  \sum_i \psi(B_i^t, x^{(0)}_i)\,,
\end{align}
for a large class of function $\psi$ (see \cite{deshpande2015asymptotic}) converge, when $N\to  \infty$, to their state evolution predictions where 
\begin{align}
M_{\rm SE}^{t} &= \mathbb{E} \left[x^0 \eta_t(Z)\right]\, , &Q^{t}_{\rm SE} &=  \mathbb{E}\left[ \eta_t(Z)^2 \right]\, ,& \Psi^t_{\rm SE} &=  \mathbb{E} \left[\psi(Z, x^{(0)}) \right]\,,
\end{align}
 where $Z$ is a random Gaussian variable with mean $\frac{M^t x^{(0)}}{\Delta}$ and variance $\frac{\Delta Q^t}{\Delta^2}$, and $x^{(0)}$ is distributed according to the prior $P_0$.
\end{theorem}
\begin{proof}
First, let us rewrite the ASP algorithm in a AMP-like form as
considered in \cite{DeshpandeM14,deshpande2015asymptotic}
\begin{align}
{\bf T}^t =& \frac{1}{\sqrt N} S \hat {\bf x}^t - b^t \hat {\bf
             x}^{t-1}\, , \label{flo2-1} \\
\hat x_{i}^{t+1} =& \eta_t({ T}_i^t)\, , \label{flo2-2}
\end{align}
with, again 
\begin{align}
b^t = \frac{ {\mathbb E} [S^2] }N
  \sum_{i=1}^N(\partial_T\eta_t({T}_i^t))\, .  \label{flo2-3}
\end{align}
With the particular choice of the following denoising function
\begin{align}
\eta_t({ T}_i^t) \eqdef \frac1s \frac{\partial}{\partial
  T_{i}}f_{\text{in}}\left[T_i^t,V^{(1),t},V^{(0),t}\right]  \label{flo2-4}
\end{align}
where $V^{(1),t}$, $V^{(1),t}$ are deterministic variables given by
eqs.~(\ref{eq:def_V1_SE}-\ref{eq:def_V0_SE}) with use of eqs.~(\ref{eq:1RSB_SE_Delta0}-\ref{eq:1RSB_SE_Delta1}).
%\begin{align}
%V^{(1),t} &= \frac 1\D \left(\D^{(1),t}+\D^{(0),t}+Q^t            \right)-\frac{\D_0}{\D^2}\D^{(1),t} \, ,& V^{(0),t}&=\frac{\D_0}{\D^2}\D^{(0),t}
%\end{align}
%and
%\begin{align}
%\D^{(0),t} &= {\mathbb{E}}[\D^{(0),t} (Z) ]  \, , &
%\D^{(1),t} &={\mathbb{E}}[\D^{(1),t} (Z) ] 
%\end{align}
%{\bf it would be useful to defined the corresponding function here,
%  otherwise it is hard to formally state the theorem} 
Once ASP is written in this form, we can apply
directly theorem \ref{theorem1} from section \ref{sec:rigorous_RS} and reach the desired result. 
\end{proof}

\subsection{Behaviour and performance of ASP }
\label{sec:use_of_ASP}

The ASP algorithm is a natural generalization of AMP, taking into
account the 1-step replica symmetry broken structure
in place of a replica symmetry. In this section we discuss the performance of ASP in terms of its
convergence and estimation error it reaches as a function of the
Parisi parameter $s$. We recall that for $s=1$ the ASP algorithm reduces
to AMP. We illustrate our findings again on the SK model. We
compare the performance of the algorithm on finite size instances with
the fixed-points of 1RSB-SE Eqs.~(\ref{eq:1RSB_SE_M})-(\ref{eq:1RSB_SE_Delta1}). % (or equivalently replica theory for a 1RSB ansatz).

\subsubsection{MSE as a function of the Parisi parameter $s$}

The most interesting quantity from the inference point of view is the
mean-squared error reached by the ASP algorithm for different values
of the Parisi parameter $s$. In Fig.~\ref{fig:1RSB_SE_MSE_alls} we
show the MSE for the planted SK model for two values of $\Delta_0$
(the two panels) as a function of $\Delta$ for various values of the
Parisi parameter $s$. 

First we note that in the whole region of convergence of AMP, i.e. in
the RS stability region, ASP converges for every value of $s$ to the same fixed point as AMP, cf. Fig.~\ref{fig:1RSB_SE_MSE_alls}.
In this sense, ASP provides a strong test (on a single instance) of the replica symmetry of the problem:
if ASP converges to the same fixed point for any value of $s$, we have
an argument to claim replica symmetry (in absence of discontinuous
phase transitions).

The horizontal line in Fig.~\ref{fig:1RSB_SE_MSE_alls}   is the optimal MSE reached by AMP
in the Bayes-optimal case where $\Delta=\Delta_0$. The black line
corresponds to the MSE reached for the equilibrium value of $s^*(\Delta)$ for which the complexity
Eq.~(\ref{eq:complexity}) vanishes, or equivalently where the free
energy Eq.~(\ref{eq:1RSB_freeEnergy}) is maximized. These are the states dominating
the Boltzmann measure if 1RSB was the correct description of the system (which it is not in this model, as well known and explained in the next section).

\begin{figure}
\centering
\includegraphics[width=0.48\columnwidth]{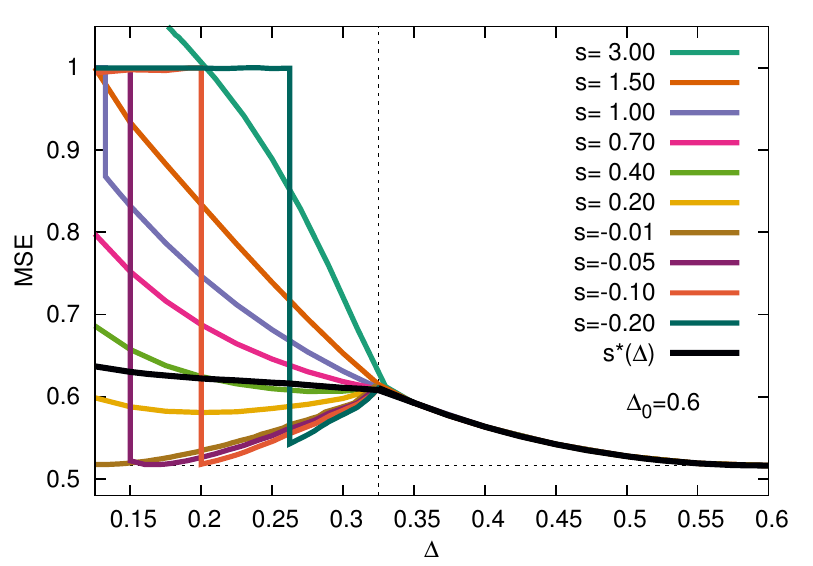} 
\includegraphics[width=0.48\columnwidth]{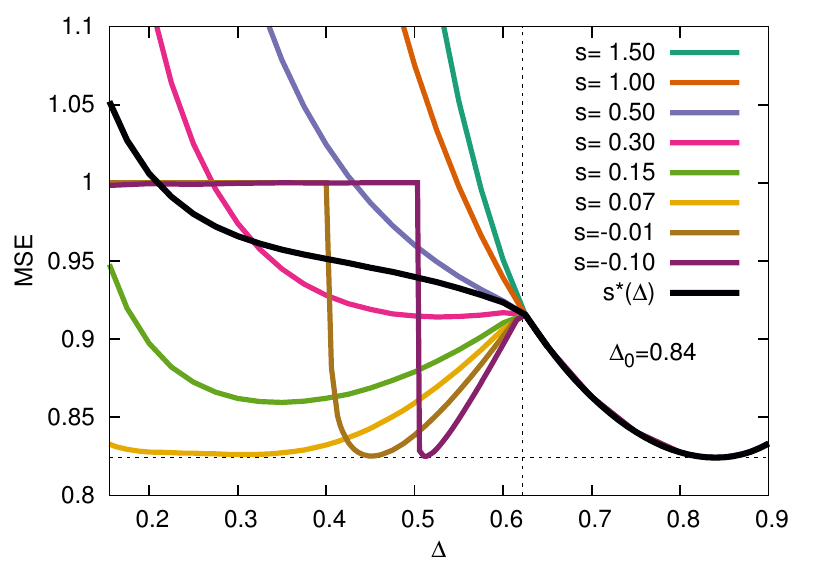} 
\caption{Mean-squared error as extracted from the solution of 1RSB SE Eqs.~(\ref{eq:1RSB_SE_M})-(\ref{eq:1RSB_SE_Delta1}) for SK model with $\Delta_0 = 0.6$ (left) and $\Delta_0 = 0.84$ (right) 
varying $\Delta$ and for several $s$. 
For large $\Delta$ (on the right of the vertical dotted line) the RS solution is stable and ASP converges to the same solution for any $s$. 
For small $\Delta$ (on the left of the vertical dotted line) the RS solution is unstable, AMP stops converging point-wise and ASP gives different solutions varying $s$.
We plot also the value of MSE obtained by ASP when $s$ is fixed to $s^*(\D)$, defined as the value of $s$ that maximizes the Bethe free energy Eq.~(\ref{eq:1RSB_freeEnergy}).
Note that in some cases the MSE obtained by ASP in the RSB region is
indistinguishable from the optimal one obtained for $\Delta=\Delta_0$. 
}
\label{fig:1RSB_SE_MSE_alls}
\end{figure}

\begin{figure}
\centering
\includegraphics[width=0.8\columnwidth]{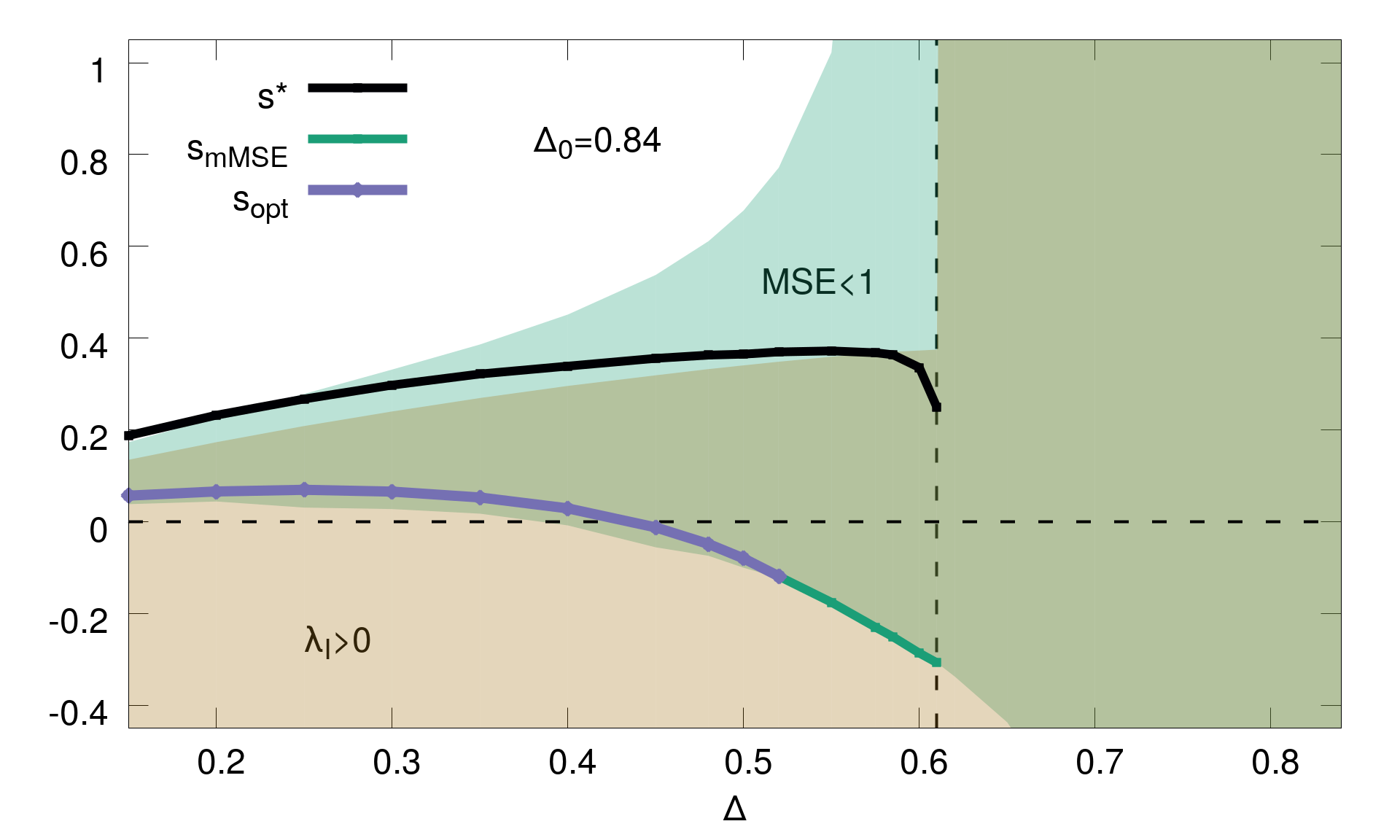} 
\caption{The values of $s_{\rm opt}$ (for which the optimal error is
  achieved), $s_{\rm m MSE}$ (for which a minimum errors
  is achieved), and the equilibrium values $s^*$
as obtained by solving the 1RSB-SE Eqs.~(\ref{eq:1RSB_SE_M})-(\ref{eq:1RSB_SE_Delta1}) for the planted SK model with  
 $\Delta_0 = 0.84$.
The vertical dashed line is where the RS solution becomes unstable.
The light orange and light green regions define the part of the plane where $\l_{\text{I}}>0$ and $\text{MSE}<1$, respectively. 
The common area is dark green.
Outside of the $\text{MSE}<1$ region, the solution is such that $\text{MSE}>1$ for larger $s$ and $\text{MSE}=1$ (trivial solution $M=Q=0$) for smaller $s$.
}
\label{fig:s_optimal_vs_Delta}
\end{figure}

\begin{figure}
\centering
\includegraphics[width=0.48\columnwidth]{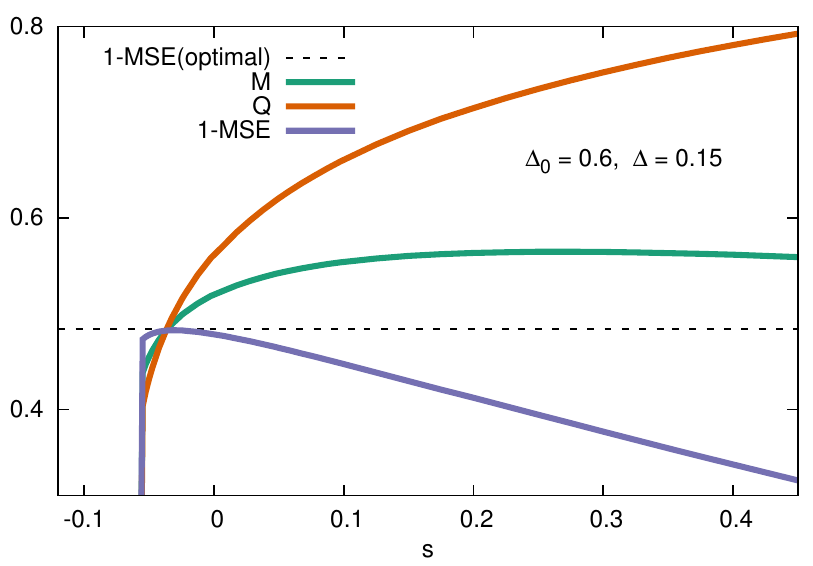} 
\includegraphics[width=0.48\columnwidth]{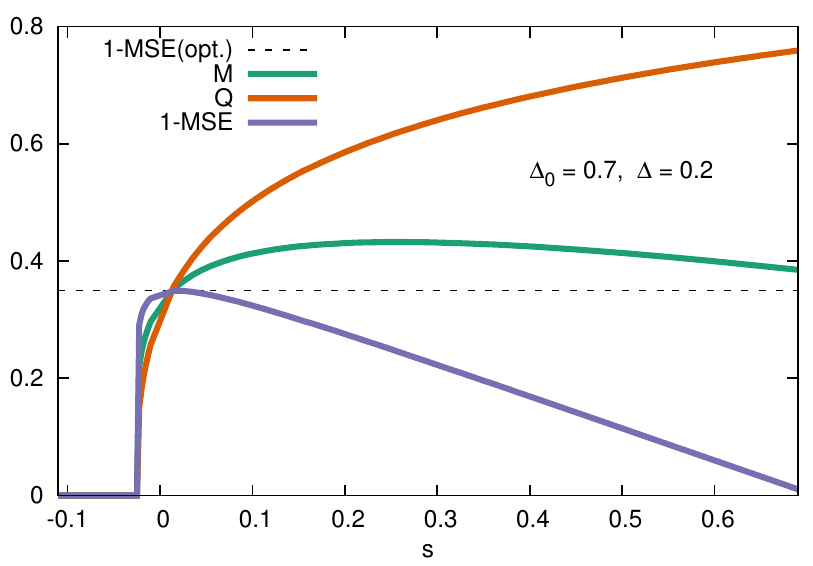} 
\\
\includegraphics[width=0.48\columnwidth]{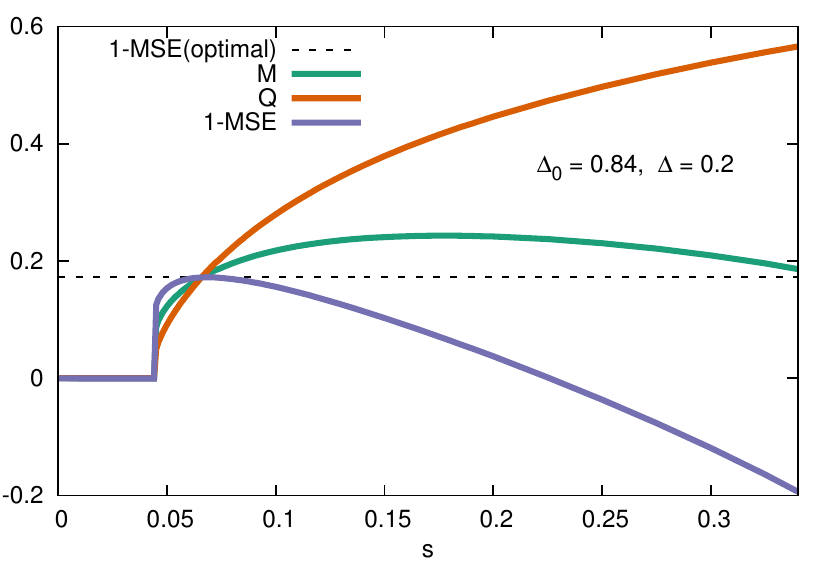}
\includegraphics[width=0.48\columnwidth]{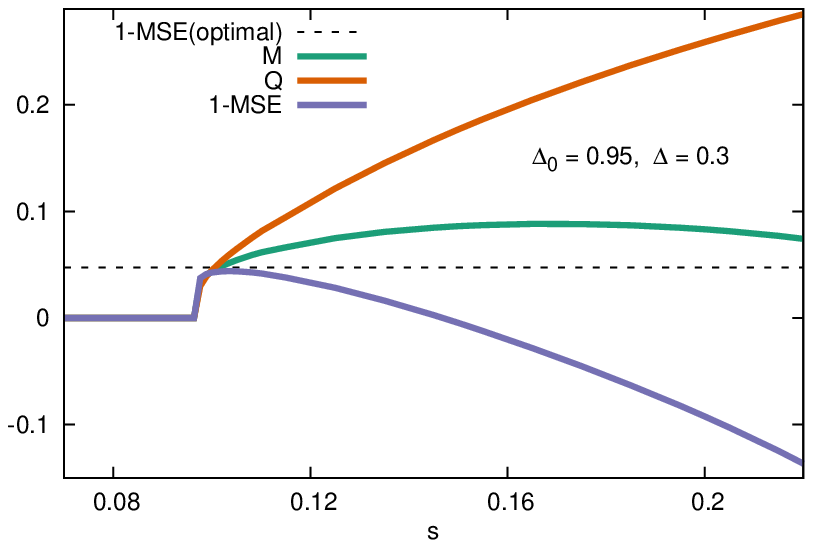} 
\caption{The values of $M$, $Q$
and $1-\text{MSE}=Q-2M$ as obtained by solving the 1RSB-SE
Eqs.~(\ref{eq:1RSB_SE_M})-(\ref{eq:1RSB_SE_Delta1}) for the planted SK model with  
 $\Delta_0 = 0.6, \Delta=0.15$ (top left),
 $\Delta_0 = 0.7, \Delta=0.2$ (top right),
 $\Delta_0 = 0.84, \Delta=0.2$ (bottom left),
 $\Delta_0 = 0.95, \Delta=0.3$  (bottom right) as a function of the
 Parisi parameter $s$.
 Note that all these points lie in the RS instability region, cf. Fig.~\ref{fig:phaseDiagramRS}.
 We show the value of MSE obtained by AMP on the Nishimori line as a dashed black line.
 In all the cases shown, the MSE is minimum when the Nishimori condition $M=Q>0$ is restored: in this point the MSE is equal to the optimal one.
 The value of $s$ in which the Nishimori condition is satisfied is:
$s=-0.0370$ (top left),  
$s=0.0127$  (top right),  
$s=0.0658$  (bottom left),  
$s=0.0994$ (bottom right). 
}
\label{fig:SE_1RSB_minimum_s}
\end{figure}

We denote by $s_{\rm mMSE}$
the value of the Parisi parameter
that minimizes the MSE.
Quite remarkably, we see in Fig.~\ref{fig:1RSB_SE_MSE_alls} that for
some cases the Bayes-optimal MSE is reached also when
$\Delta \neq \Delta_0$ for a particular value of the Parisi parameter
that we denote $s_{\rm opt}$.  This is also
seen in Fig.~\ref{fig:D00p84_M_Q_MSE_lambdaI} left panel where the MSE is plotted as a function
of $s$ for a variety of values of $\Delta$. Moreover, we remark and
illustrate in Fig.~\ref{fig:s_optimal_vs_Delta} that
the value of the Parisi parameter
$s_{\rm mMSE}$ that minimizes the MSE, and $s_{\rm opt}$ for which
the Bayes optimal error is reached, are unrelated to the equilibrium value
$s^*$.
Note that, somehow counter-intuitively, very closely under the RS instability the optimality condition $M=Q$ cannot typically be achieved 
for any value of $s$ (see, e.g., the yellow curve corresponding to $\D = 0.6$ in Fig.~\ref{fig:D00p84_M_Q_MSE_lambdaI}): 
in this case the 1RSB solution depends weakly on $s$ and 
one explores only a narrow interval in $M$ and $Q$ (that may not include the case $M=Q>0$) before jumping to the trivial solution $M=Q=0$.
%However, note that this problem happens only close to the RS instability, as moving further away from the Nishimori line
%the case $M=Q>0$ is typically feasible. 

In Fig.~\ref{fig:s_optimal_vs_Delta} we show the $s_{\text{opt}}$, and the value $s_{\text{mMSE}}$ for which the MSE is minimum at given $\D$ 
(when $s_{\text{opt}}$ exists the two are equal) 
and the region where $\text{MSE}<1$ for the SK model with $\D_0=0.84$.
Near the RS instability the MSE depends weakly on $s$ and the $s_{\text{opt}}$ does not exist: in this case the minimum MSE is obtain for the minimum $s$ 
for which the solution is not trivial (green line). 
%At the same time the positivity of $\l_{\text{I}}$ (and so the point-wise convergence of ASP) is limited only at $s< 0.374$ immediately beyond the RS instability.
%Lowering $\D$ the 1RSB-SE solution depends more strongly on $s$: the region $\text{MSE}<1$ shrinks and $s_{\text{opt}}$ exists for $\D < 0.52$.
In the same figure we also plot the value of $s^*$ that extremizes the Bethe free energy of the model:
it is well distinct from $s_{\text{opt}}$.
% and it has a negative $\l_{\text{I}}$ for $\D<0.575$ and a $\text{MSE}>1$ for $\D<0.217$.

To clarify the origin of the above observation that Bayes-optimal
error can be restored for $s_{\rm opt}$ we remind that in Sec.~\ref{sec:optimality_and_mismatching} we noticed that the inference can be optimal 
out of the Nishimori line. More precisely, we have put forward a hypothesis that
the estimation is optimal every time the Nishimori condition $M=Q>0$ is restored.
We showed one example in Fig.~\ref{fig:M_Q_MSE_rho_neq_rho0}:
introducing a nonzero density of zeros in the prior distribution, it is possible to obtain optimal estimation in the SK model using AMP out of Nishimori.
%Unfortunately, it is not clear how to find the right parameters to restore the condition $M=Q$ for mismatching models.
%There are learning strategies based on the Bethe free energy \cite{zdeborova2015statistical,dempster1977maximum}
%to learn the parameters on the Nishimori line, but these do not usually work for mismatching models, cf. Sec. \ref{sec:optimality_and_mismatching}.
%
%Restoring the optimality condition is then particularly troublesome in AMP: 
%if one is not strictly aiming to restore the optimality condition $M=Q$ on the Nishimori line, the possible choices of additional parameters 
%is potentially unlimited and we do not have a definite method to compare and select them.
%
The ASP algorithm naturally introduces a free parameter in the form of
the Parisi parameter $s$, cf. Eq.~(\ref{eq:rep_Z}). In Fig. \ref{fig:SE_1RSB_minimum_s} we
illustrate that minimization of the MSE as a function of $s$ is again
related to restoration of the Nishimori condition $M(s)=Q(s)>0$.

This can also be seen analytically by computing the 
derivative of the MSE with respect to the parameter $s$ at the fixed-point of 1RSB-SE. 
In the replica theory notation \cite{MPV87,SD84} we can express this derivative, cf. Appendix \ref{sec:derivative_of_MSE}, as
\begin{align}
 \frac{\partial \text{MSE}}{\partial s} &= 
 -  2 \int dh \, \frac{\partial f_{\text{I}}(s,h)}{\partial s}  \, D(s,h) 
\end{align}
while 
\begin{align}
M-Q =  \int dh \, f_{\text{I}}(s,h) \, D(s,h)\, ,
\end{align}
with 
\begin{align}
 D(s,h) = {\mathbb{E}}_{x^{(0)}} 
 \left[ \left( f_{\text{I}}''(s,h) + f_{\text{I}}'(s,h) \frac{T}{\hat Q}  +  x^{(0)} \frac{T}{\hat Q} \right)  
 \frac{1}{\sqrt{2 \pi (- \hat Q)}}  \exp \left(  \frac{ T^2 }{2 \hat Q} \right)  \right] \, ,
\label{eq:D_for_MSE}
\end{align}
$\hat Q = - \Delta_0 Q / \Delta^2 $ and %(cf. Eqs.~(\ref{eq:f1h})-(\ref{eq:P1h}))
\begin{align}
 f_{\text{I}}(s,h) &= \frac{1}{s} \log {\mathbb{E}}_{W} \left[  e^{s
          f_{\text{II}}(h-\sqrt{V^{(0)}} W)} \right] \, ,
 &
 f_{\text{II}}(h) &= \log \, {\mathbb{E}}_{x} \left[ \exp \left( -
          \frac{V^{(1)}}{2} x^2 + h x \right) \right] \, .
\end{align}
So, if there exists a value of $s$ such that $D(s,h)$ is zero for
all $h$, we have that for such value $M=Q$ and MSE is extremized.
This is exactly what we observe: at the solution where MSE is minimum the function $D(s,h)$ is zero for all $h$ (and so $M=Q>0$).
%While it is not evident (at least from the previous equations) that the only solution in $s$ of $\partial_s \text{MSE}=0$ is such that $D(s,h)=0$, 
%at the fixed-points of 1RSB-SE it is the one found.
In this sense we could define the optimal $s_{\rm opt}$ as the value for which at the fixed-point of 1RSB-SE it happens that (if exists)
\begin{align}
& D(s_{\text{opt}},h)
= 
{\mathbb{E}}_{x^{(0)}} 
 \left[ \left( f_{\text{I}}''(s_{\text{opt}},h) + f_{\text{I}}'(s_{\text{opt}},h) \frac{T}{\hat Q}  +  x^{(0)} \frac{T}{\hat Q} \right)  
 \frac{1}{\sqrt{2 \pi (- \hat Q)}}  \exp \left(  \frac{ T^2 }{2 \hat Q} \right)  \right]
\equiv 0  \, .
\end{align}

%[Do you think one could try to algorithmically optimize D toward zero somehow? Note that $P(s,h)$ depends on $M$, so it is not clear how to do so.. ]

%[I do not know how much to keep about this discussion, but probably it is worth keeping, at worst in an appendix]

It seems to us that there should be a deeper theoretical reason behind
the observation that the restoration of the Nishimori condition leads
to the optimal estimation error. We have not found it and let this
question for future work.

%that can in principle be fixed to restore the optimality condition, cf. Fig.~\ref{fig:SE_1RSB_minimum_s}. 
%In the approach based on ASP instead, the choice of the parameter is already naturally present in the $s$ parameter.
%The introduction of this parameter is physically motivated (by the breaking of the replica symmetry) 
%and it reduces automatically to the standard case when the algorithm is performed on the Nishimori line.
%In this logic, ASP (and the equivalent algorithms for further breaking of the replica symmetry)
%is the natural setting to perform optimal inference on mismatching models.
%The Fig.~\ref{fig:SE_1RSB_minimum_s} illustrates 4 examples in the RSB region where the reconstruction is optimal in ASP for some particular value of $s$.
%As expected, the eigenvalue $\l_{\text{I}}$ is positive for such values of $s$, cf. Fig.~\ref{fig:1RSB_lambdaI_s}, so that the point-wise convergence
%is guaranteed in these cases.

%We point out that the examples in
%Fig.~\ref{fig:SE_1RSB_minimum_s} are placed in a favorable region of parameters and other situations are possible.
%Let us fix, for example, $\D_0 = 0.84$ and check the behavior of ASP for different values of $\D$.
%We showed some results for this case in
%Fig.~\ref{fig:1RSB_SE_MSE_alls}, \ref{fig:1RSB_SE_lambdas_alls}
%and we further report also in Fig.~\ref{fig:D00p84_M_Q_MSE_lambdaI} in a different format.

\subsubsection{Point-wise convergence of ASP}

The inclusion of a 1RSB structure into the ASP algorithm extends the region of parameters in which 
the algorithm converges point-wise with respect to AMP.  
Note that the range of parameters for which AMP does
not converge point-wise corresponds exactly to the range of parameters
for which the large iteration time behavior of ASP depends on the
Parisi parameter $s$, see Fig.~\ref{fig:1RSB_SE_MSE_alls}, where below
the vertical line is
exactly where the replicon eigenvalue eq. (\ref{eq:repliconRS})
becomes negative.

\begin{figure}[h!]
\centering
\includegraphics[width=0.48\columnwidth]{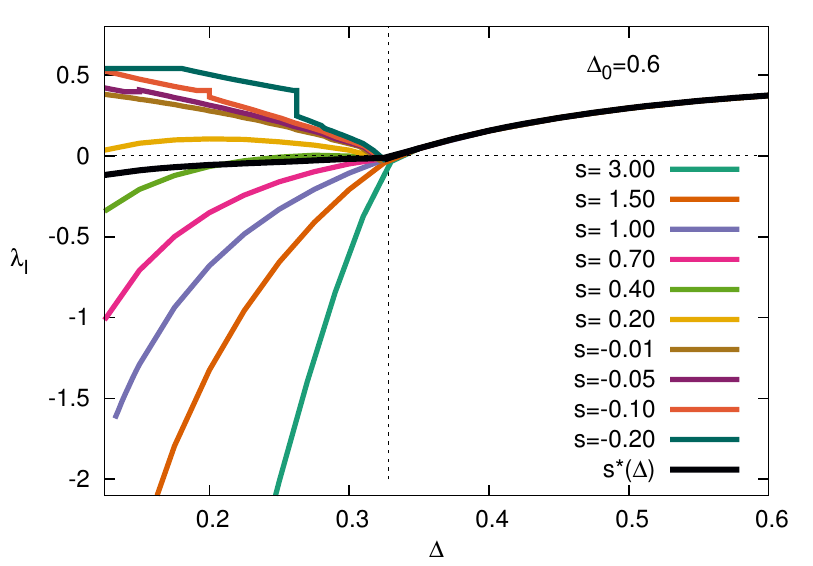} 
\includegraphics[width=0.48\columnwidth]{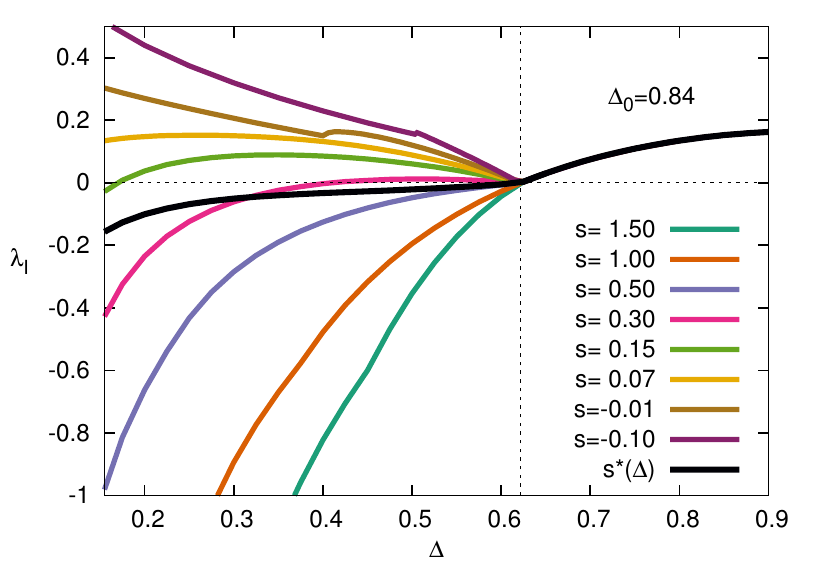} \\
\includegraphics[width=0.48\columnwidth]{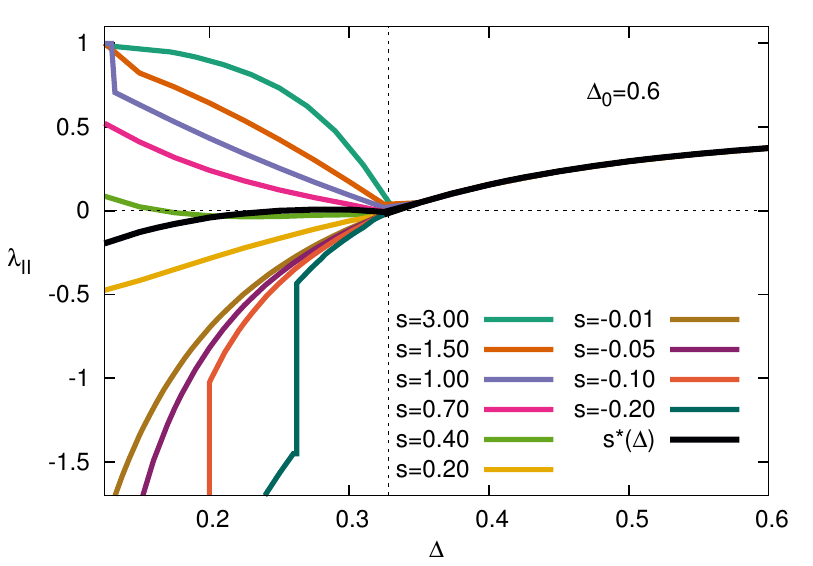} 
\includegraphics[width=0.48\columnwidth]{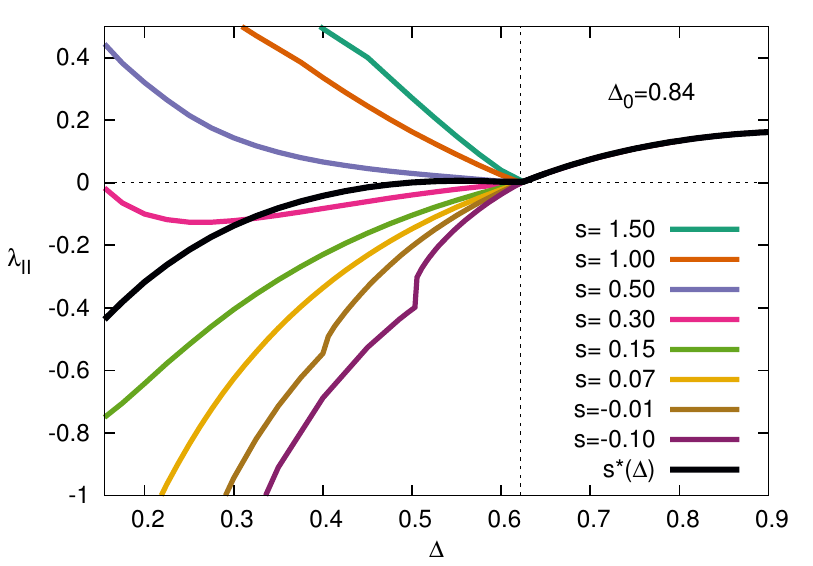} 
\caption{The stability parameters (\ref{eq:lambdaI_1RSB}-\ref{eq:lambdaII_1RSB}) as obtained by
  solving the 1RSB SE
  Eqs.~(\ref{eq:1RSB_SE_M})-(\ref{eq:1RSB_SE_Delta1}) for the planted SK model with $\Delta_0 = 0.6$ (left) and $\Delta_0 = 0.84$ (right) 
varying $\Delta$ and for several $s$. 
For large $\Delta$ (on the right of the vertical line) the RS solution is stable and ASP converges to the same solution for any $s$. 
For small $\Delta$ (on the left of the vertical line) the RS solution
is unstable, AMP stops converging point-wise and ASP converges if and
only if $\l_{\text{I}}$ is positive. 
%gives different
%convergence properties varying $s$
%(some of them with a MSE close to the optimal one obtained for $\Delta=\Delta_0$). 
We plot also the values of $\lambda_{\text{I}}$,$\lambda_{\text{II}}$ obtained by ASP
when $s$ is fixed to the equilibrium value $s^*(\D)$.
%, computed as the value of $s$ that maximizes the Bethe free energy Eq.~(\ref{eq:1RSB_freeEnergy}).
}
\label{fig:1RSB_SE_lambdas_alls}
\end{figure}

%Let us analyze this fundamental point in detail.
%When the model mismatching  is substantial, the replica symmetry gets broken and we enter in a glassy phase:
%here 1RSB-SE acquires different fixed-points as the Parisi parameter $s$ is tuned, 
%cf. Fig.~\ref{fig:1RSB_SE_MSE_alls} in the low $\D$ section (below the vertical dotted line). %, that is the point in which the replica symmetry gets broken).
%In such situation AMP does not converge point-wise.
%, cf. Fig.~\ref{fig:convergence_AMP}.  

ASP converges point-wise in a larger region that AMP, at least for some value of $s$.
The analysis of the point-wise convergence of single instances is obtained looking at the 
Hessian eigenvalues of ASP equations. Again, this can be readily
derived by simply repeating the reasoning we used for AMP in
section \ref{sec:phase_diagram_RS} to reach Eq.~(\ref{eq:repliconRS}), using this time ASP Eqs.~(\ref{flo2-1})-(\ref{flo2-4}), 
instead of the AMP ones Eqs.~(\ref{flo1-1})-(\ref{flo1-4}). 
This leads to a perturbation growing as $(\Delta_0/\Delta^2 ) {\mathbb E}_B
\eta'(B)$. Leading to the condition 
\begin{align}
 \lambda_{\rm ASP} = 1 - \frac{\D_0}{\D^2} \mathbb{E}_{x^{(0)},T} \left[ \left(
\frac 1s  \frac{\partial^2 f_{\rm in}(T,V_1,V_0)}{\partial T^2}  \right)^2 \right] \, .
\label{eq:convergence1RSBASP}
\end{align}

Interestingly, this corresponds to a well known quantity in the replica
theory \cite{MPV87,SD84}, that appears in the stability of the 1RSB
solution against further breaking of replica symmetry. In in this case
(see Appendix \ref{sec:equivalence_SE_replica}) two kind of
instabilities are often discussed
\cite{MonRic03,MonParRic04,KrzPagWei04,MerMezZec05} and the two eigenvalues
that express the 1RSB stability can be written as
\begin{align}
\label{eq:lambdaI_1RSB}
 \l_{\text{I}} &= 1-\frac{\D_0}{\D^2}\int_{-\infty}^\infty \de h
                 P_{\text{I}}(h)\left(f_{\text{I}}''(s,h)\right)^2\, ,\\
\l_{\text{II}} &= 1-\frac{\D_0}{\D^2}\int_{-\infty}^\infty \de h P_{\text{II}}(s,h)\left(f_{\text{II}}''(h)\right)^2\:. \label{eq:lambdaII_1RSB}
\end{align}
where 
\footnote{This notation is convenient in replica theory, in particular to generalize to full replica symmetry breaking (FRSB) ansatz.
The function $P_{\text{II}}(s,h)$ of Eq.~(\ref{eq:P1h}) is equal to the function $P(x,h)$ \cite{SD84}, that enforces the Parisi equation 
and gives the local field distribution, when $x=s$.
}
\begin{align}
 f_{\text{II}}(h) &= \log \, {\mathbb{E}}_{x} \left[ \exp \left( - \frac{V^{(1)}}{2} x^2 + h x \right) \right] 
 \label{eq:f1h}
 \\
 f_{\text{I}}(s,h) &= \frac{1}{s} \log {\mathbb{E}}_{W} \left[  e^{s f_{\text{II}}(h-\sqrt{V^{(0)}} W)} \right]
 \\
 P_{\text{I}} \left( h \right) &=
 \frac{\D}{ \sqrt{2 \pi \D_0 Q}} {\mathbb{E}}_{x^{(0)}} \left[ \exp \left( -  \frac{\D^2}{2 \D_0 Q} \left(  \frac{M}{\D} x^{(0)} + h \right)^2 \right) \right]
 \\
 P_{\text{II}}(s,h) &= e^{s f_{\text{II}}(h)} \, {\mathbb{E}}_{W} \left[  P_{\text{I}}(h- \sqrt{V^{(0)}} W) \, e^{-s f_{\text{I}}(s,h-\sqrt{V^{(0)}} W)}  \right]
  \label{eq:P1h}
\end{align}
with $V^{(1)}, V^{(0)}$ given by Eqs.~(\ref{eq:def_V1_SE})-(\ref{eq:def_V0_SE})
%with
%\begin{align}
%& V^{(1)} =  \frac{1}{\D} \left( Q+\D^{(1)}+\D^{(0)} \right) - \frac{\D_0}{\D^2} \D^{(1)}
% \, , &
%& T =   \frac{M}{\D} x^{(0)} + h % \frac{\sqrt{\D_0 Q}}{\D} W
% \, , &
%& V^{(0)} = \frac{\D_0}{\D^2} \D^{(0)}
%\end{align}
and $x$, $x^{(0)}$ and $W$ being random variables distributed
according $P(x)$, $P^{(0)}(x)$ and a standard Gaussian distribution,
respectively.

A change of variable, as in Appendix \ref{sec:equivalence_SE_replica},
shows that $\lambda_{\rm ASP}= \l_{\text{I}}$: the convergence of the
ASP algorithm is determined by the first-type instability towards
further replica symmetry breaking, just as it is for Survey
Propagation  \cite{MonParRic04,KrzPagWei04,MerMezZec05}.

%
%-------------\\
% here I'm also assuming that  $\hat q_d^{(0)} = 0$  otherwise it will be $P(s,h)$ would be
%\begin{align}
%  P \left(s, h \right) &= \frac{\D}{\hat r (0) \sqrt{2 \pi \D_0 Q}} 
% {\mathbb{E}}_{x^{(0)}} \left[ \exp \left[ - \hat q_d^{(0)} \left( x^{(0)} \right)^2 -  \frac{\D^2}{2 \D_0 Q} T^2 \right] \right]
%\end{align}
%----------------
%
In the region of RS stability $\D^{(0)} = V^{(0)} = 0$ and the two
eigenvalues are equal and coincide with the RS replicon eigenvalue in Eq.~(\ref{eq:repliconRS}).
To understand the meaning of these two eigenvalues consider that the 1RSB structure consists of equivalent metastable states whose distance from each other is given by $Q$. 
There are then two ways in which a further hierarchical level of RSB states can appear:
the 1RSB states can aggregate in a way to establish a new scale of distance between states (type I instability, associated to negative $\lambda_{\text{I}}$),
or each metastable state can split into a hierarchy of new states (type II instability, associated to negative $\lambda_{\text{II}}$). 
%These instabilities have a direct significance for the ASP algorithm \cite{MonRic03,MonParRic04,KrzPagWei04,MerMezZec05},
%in particular the eigenvalue $\l_{\text{I}}$ is equal to $\l_{\text{ASP}}$ and so it guarantees the point-wise convergence of the algorithm. 
%The equality of the two forms for this eigenvalue is easily seen on the same lines of the 1RSB-SE equations, cf. Appendix \ref{sec:equivalence_SE_replica}.

\begin{figure}[h!]
\centering
\includegraphics[width=0.48\columnwidth]{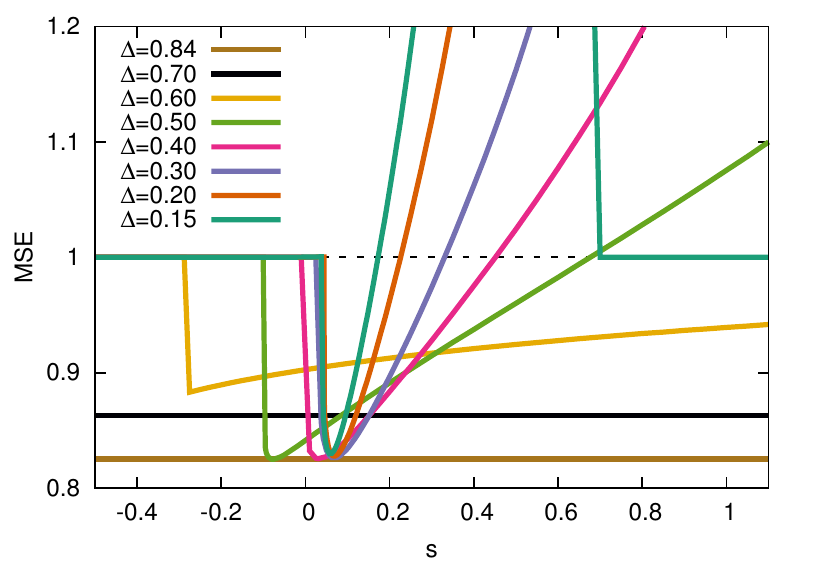}
\includegraphics[width=0.48\columnwidth]{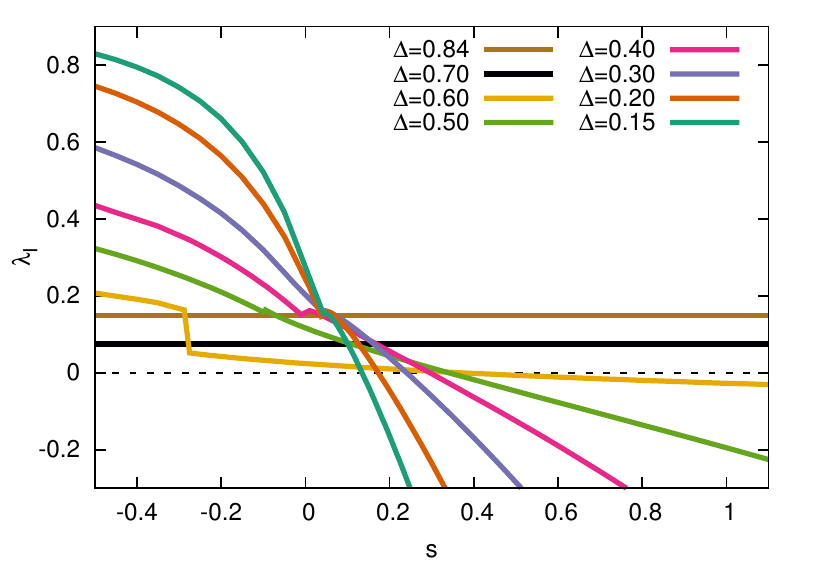} \\
\caption{MSE, and the eigenvalue $\l_{\text{I}}$ Eq.~(\ref{eq:lambdaI_1RSB})
  evaluated from the solution of 1RSB-SE
  Eqs.~(\ref{eq:1RSB_SE_M})-(\ref{eq:1RSB_SE_Delta1}) for the planted SK model with  
$\Delta_0 = 0.84$ varying $s$ for several values of $\D$. In this case the RS instability arrives at $\D = 0.622$.
}
\label{fig:D00p84_M_Q_MSE_lambdaI}
\end{figure}

Let us discuss the behavior of the eigenvalues as we move away from the Nishimori line.
To be concrete, we keep referring to the case of the planted SK model Eq.~(\ref{eq:SK_prior}), cf. Fig.~\ref{fig:1RSB_SE_lambdas_alls}. 
In the RS stability region, the two eigenvalues are equal and positive.
At the boundary of the RS stability region, the eigenvalues become marginal and the replica symmetry gets broken:
beyond this line the solution of the 1RSB-SE depends on $s$ and so do $\l_{\text{I}}$ and $\l_{\text{II}}$.
For the low-rank matrix estimation problem, the transition is towards
a full replica symmetry broken (FRSB) state \cite{MPV87},
so that the 1RSB states are always unstable and at least one among $\l_{\text{I}}$ and $\l_{\text{II}}$ is negative beyond the RS instability line. 
Nevertheless, unlike the RS algorithm laid out by AMP, in ASP there is a scale of distance among state - tuned via the $s$ parameter -
so that the algorithm can converge point-wise also when the inner structure of the states is more complicated than 1RSB.
%Indeed, close to the RS instability, 
%exactly one among $\l_{\text{I}}$ and $\l_{\text{II}}$ must be negative - depending on $s$ being smaller or larger than the breaking point.
%
In Fig.~\ref{fig:1RSB_SE_lambdas_alls} we show the value of $\l_{\text{I}}$ and $\l_{\text{II}}$ crossing the RS stability line for several values of $s$.
We see $\l_{\text{I}}$ is positive for $s$ sufficiently small, indicating that ASP converges point-wise for such values also in the RSB region (cf. Fig.~\ref{fig:AMP_vs_1RSB}).
%This validates our main conclusion about the point-wise convergence of ASP in the RSB region.
%
%We expect that ASP must converge point-wise for some value of $s$ in a mismatching model setting. 
%A possible exception may be a strong wrongly biased prior {\bf [IS IT?]}, as, e.g., wrongly assuming only positive variables when the true ones are also negative,
%but this could be avoided assuming a general enough prior.
%If still Low-RASP does not converge for any value of $s$, we would expect that 
%there is a bigger issue than simple parameters mismatching, as, e.g., noise elements being not independent, 
%that would invalidate any belief-propagation-based approach.
%In this sense, ASP may provides a test on the class of a problem: 
%unlike AMP, using ASP we can exclude that the lack of convergence is due to RSB.

%\subsubsection{Comparison of ASP and the state evolution}
\subsubsection{Results on single instances of ASP}

So far we mostly concentrated on results of the state
evolution of the ASP algorithm. We now show that
these results are indeed describing the behaviour of the ASP algorithm
run on large but finite-size instances. 

In Fig.~\ref{fig:ASP_and_SE} we compared how the error evolves as a
function of the iteration time for the state evolution an the ASP
algorithm of several instances of two different system sizes ($N=2000$
and $N=64000$). We see that for the large system size the finite size
effects are small enough and the agreement with the theory is rather good for each of the runs.

\begin{figure}[h!]
\centering
\includegraphics[width=0.49\columnwidth]{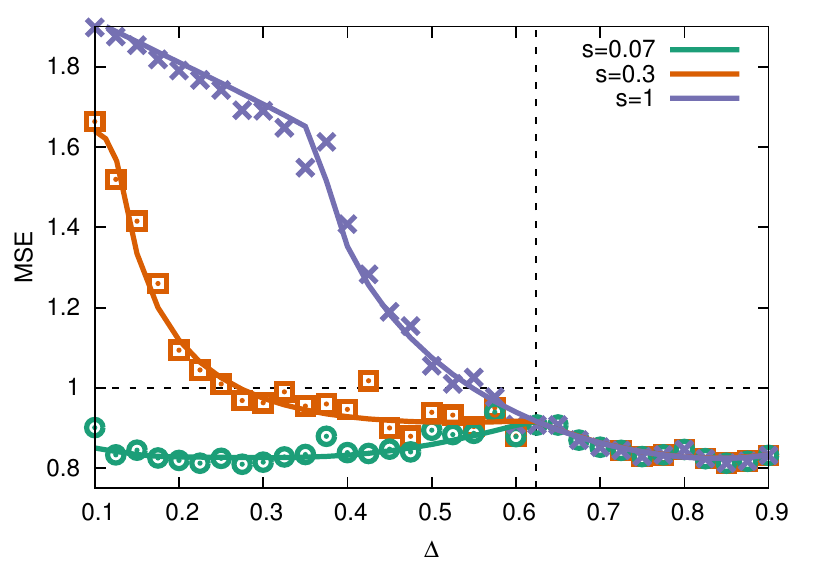}
\includegraphics[width=0.49\columnwidth]{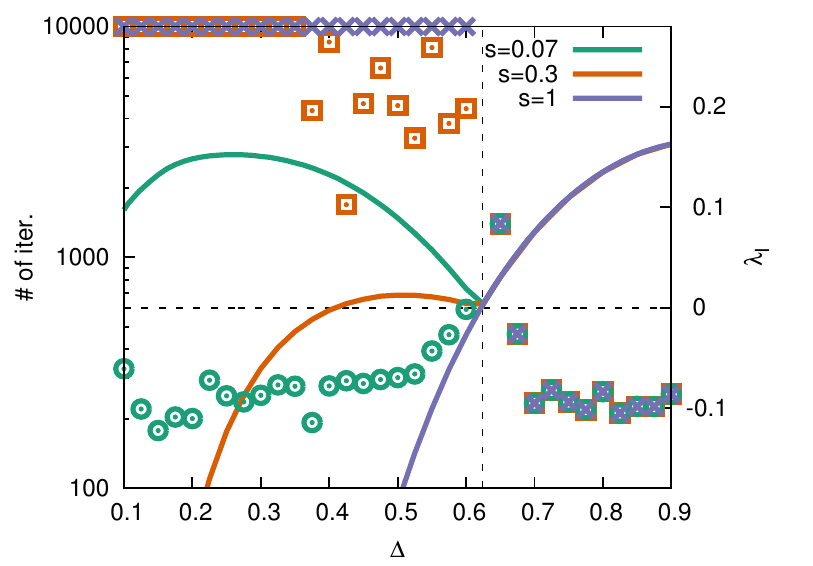}
\caption{MSE and convergence time of ASP for $N=50000$ with $s=0.07$, $s=0.3$ and $s=1$ (equivalent to AMP)
for the SK model with $\Delta_0 = 0.84$ varying $\Delta$. Each point is an average over 3 instances of the problem. 
The iterations are stopped when the average change of a variable in a single iteration of ASP is less than $10^{-8}$ or at $10^4$ iterations if convergence is not reached.
The continuous line are the values of MSE (left) and of $\l_{\text{I}}$ (right) obtained with the 1RSB-SE Eqs.~(\ref{eq:1RSB_SE_M})-(\ref{eq:1RSB_SE_Delta1}).
}
\label{fig:ASP_vs_SE_fixs}
\end{figure}

In Fig.~\ref{fig:ASP_vs_SE_fixs} we report the result of ASP for size $N=5 \times 10^4$ and $s=0.07$, $0.3$, $1$ in the same situation 
of the right column of Figs.~\ref{fig:1RSB_SE_MSE_alls} and \ref{fig:1RSB_SE_lambdas_alls}.
We see that in the RS region the fixed-point is the same for all $s$ while in the RSB region the result depends strongly on $s$.
Note that for $s=0.07$ the MSE is less than 1 for any $\Delta$ in all the instances.
The convergence time is increasing for any $s$ at the RS stability -- where $\l_{\text{I}} =0$ -- but the algorithm keeps converging point-wise for $s=0.07$ even at low $\D$, 
as predicted by the positivity of $\l_{\text{I}}$ from 1RSB-SE at $s=0.07$.
In general, we see as for this size the trend predicted by SE is visible in the MSE obtained by ASP.

Finally, in Figs.~\ref{fig:phaseDiagramASP5k_alls} we show the heat-maps of the
MSE (left panels) and convergence times (right panels) as obtained by the
ASP algorithm run on single instances of $N=5000$ variables and, respectively, $s=0.5$, $s=0.1$ and $s=0.02$. 
We use the same color-map as Fig.~\ref{fig:phaseDiagramAMP5k}, to enhance the difference w.r.t. AMP.
On this scale, ASP with $s=0.5$ is very similar to AMP, but for $s=0.1$ and $s=0.02$ the difference becomes evident.
In particular, the $\text{MSE}<1$ and convergence region extend to lower values of $\D$ including sections of the RSB region, cf. Fig.~\ref{fig:phaseDiagramRS}.
It is visible that the convergence time increases at the RS
instability boundary, where $\l_{\text{I}} = 0$, but then for $s=0.1$ and
$s=0.02$ the convergence time goes down again beyond the RS instability.
Note that also the region of convergence to the trivial point $Q=M=0$ is extended, 
so that for very small $s$ the ASP algorithm convergences point-wise to either the trivial or to nontrivial fixed-point in most of the phase diagram.

\begin{figure}[h!]
\centering
\includegraphics[width=0.46\columnwidth]{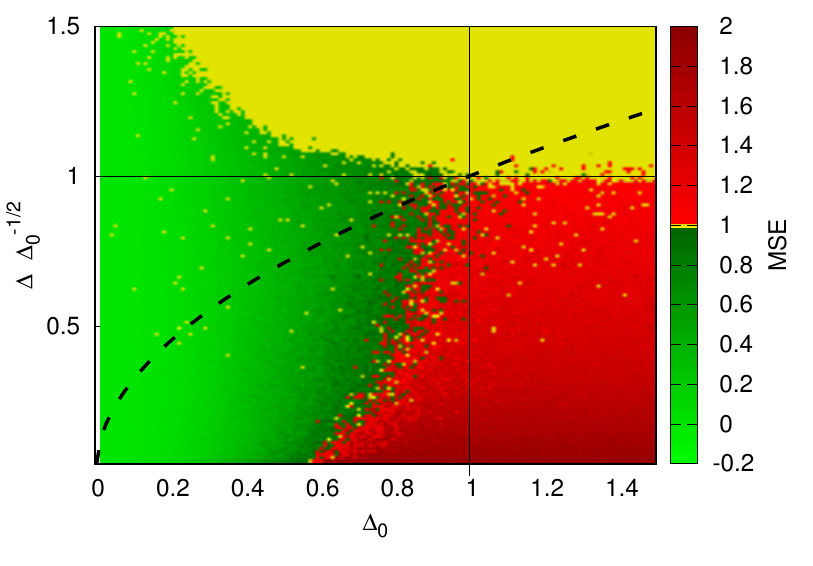}
\includegraphics[width=0.46\columnwidth]{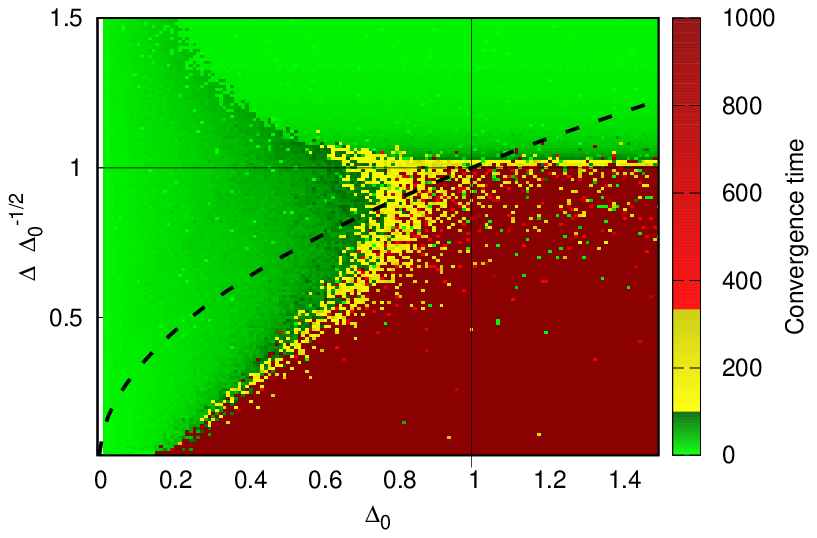} \\
\includegraphics[width=0.46\columnwidth]{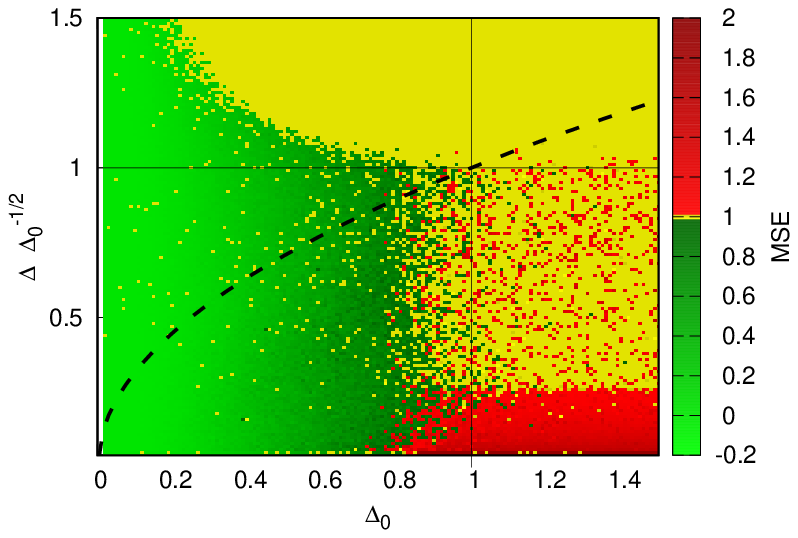}
\includegraphics[width=0.46\columnwidth]{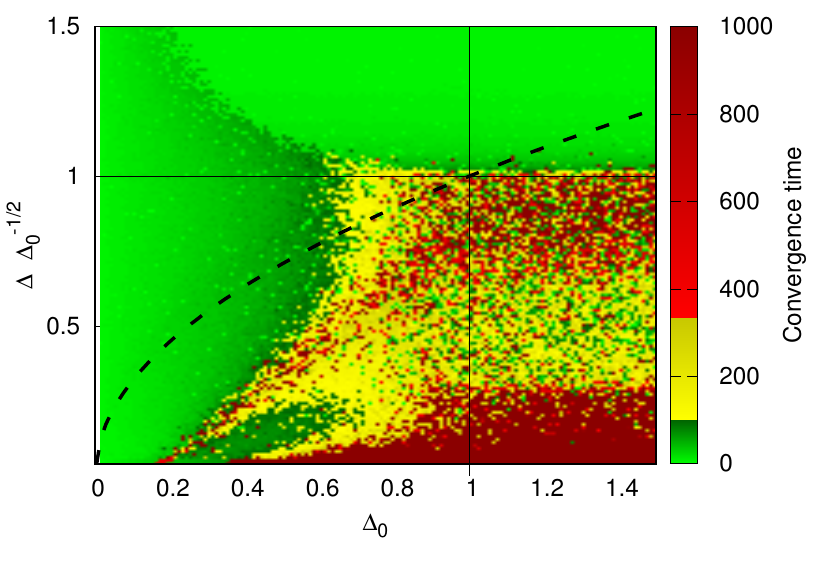} \\
\includegraphics[width=0.46\columnwidth]{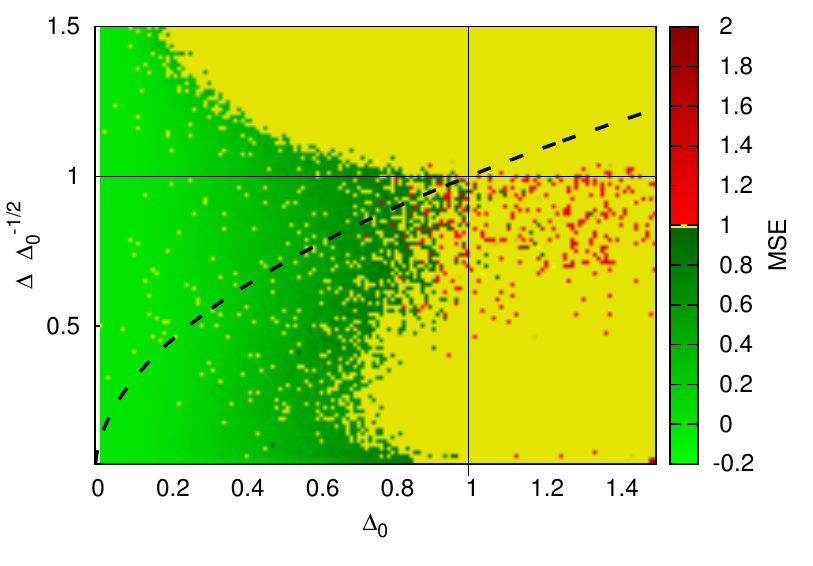}
\includegraphics[width=0.46\columnwidth]{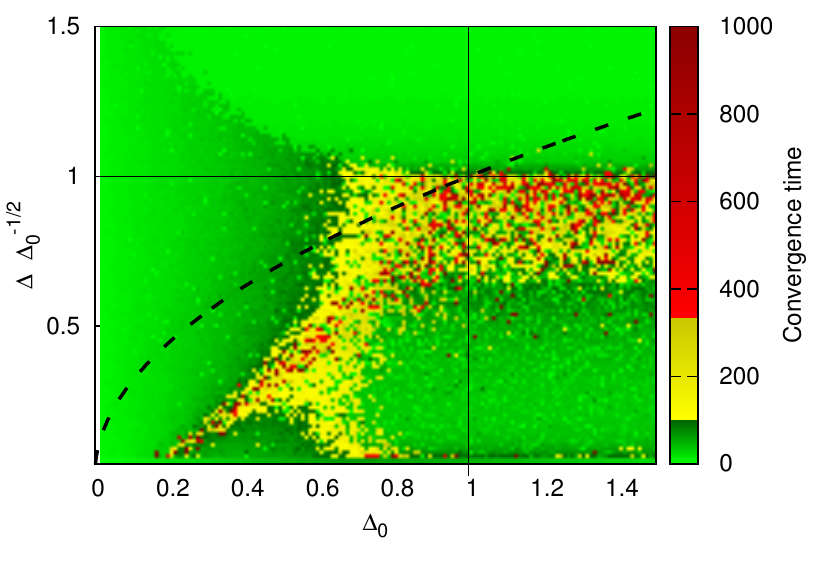}
\caption{Heat-map of the MSE (left) and convergence time (right) obtained running ASP for $s=0.5$ (top), $s=0.1$ (middle) and $s=0.02$ (bottom) for size $N=5000$ for the
 planted SK model. 
The dashed line is the Nishimori line $\Delta = \Delta_0$.
The iterations are stopped after 1000 iterations or as soon as the average change of a variable in a single iteration of ASP is less than $10^{-8}$.
}
\label{fig:phaseDiagramASP5k_alls}
\end{figure}

\section{Conclusions and open questions}

%Speculation that FRSB AMP finds the ground states and this can maybe be proven. 

In this paper we introduced the approximate survey propagation (ASP)
algorithm for the class of low-rank matrix estimation problems.  We
derived the state evolution describing the large size behavior of the
algorithm, finding that the fixed-points of the state evolution of ASP
reproduce the one-step replica symmetric fixed-point equations
well-known in physics of disordered systems.  

This leads to a new algorithmic interpretation of the
replica method where the self-consistent equations actually describe
the behavior of an iterative algorithm: AMP in the replica
symmetric case, and ASP when the replica symmetry is broken (just as
BP and SP corresponds to the same situation on sparse graphs).

We characterize the performance of ASP in terms of convergence and
mean-squared error as a function of the free Parisi parameter $s$.  In
particular, we reported the results of the algorithm for the analysis
of the planted Sherrington-Kirkpatrick model.
%We study as a
%function of the model parameters and of the Parisi parameter $s$ when
%ASP converengece to same fixed-point as AMP and when it does not. 
Notably we found that when there is a model mismatch between the true generative model and the inference model, the performance of AMP rapidly degrades
both in terms of MSE and of convergence. Using ASP for a suitably
chosen value of $s$ we observed we can always restore convergence and
improve the estimation error. 

 %In particular, AMP stops converging as soon as the replica symmetry gets broken.
%The definition of ASP encompasses a replica symmetry breaking structure for the inferred variables, providing an algorithm that can converge also in a RSB region, 
%when the mismatch with the generative model is substantial.
%We observed that for the low-rank matrix estimation problem, in any point of the phase diagram the ASP algorithm converges point-wise for at least some values of $s$.
%This property overcomes a common problem of AMP, as in many practical tasks it is not possible 
%to reach exact knowledge of the generative process. Therefore a mismatch between the inference and generative models is often inevitable. %
%We stress that this property of ASP is pretty general: the real solution of the low-rank estimation problem is not 1RSB, but rather FRSB.
%We reported that the point-wise convergence of ASP is associated to only one particular direction of instability in the 1RSB solution,
%so ASP can converge also for problems where the 1RSB solution is indeed unstable.

Among other results, our analysis leads us to a striking hypothesis that whenever $s$ (or other parameters) can be set in such a way that the Nishimori
condition $M=Q>0$ is restored, then the algorithm is able to reach mean-squared error as low as the Bayes-optimal error obtained on the
Nishimori line, i.e. when the model and its parameters are known and matched in the inference procedure. 
Another hypothesis to which we have not find a counter-example in the present model is whether the
Nishimori condition $M=Q$ implies convergence of the ASP algorithm.
Whenever $M=Q$ we always observed that ASP converges point-wise to a fixed-point with minimal MSE even if the generative model and the inference model are highly mismatched.
Unveiling the physical origin and range of validity of these properties is an interesting direction for future work.

Another direction for future work %in case the inference model is mismatched with the generative model 
is an algorithmic procedure able to select values of the Parisi
parameter $s$ that lead to low estimation errors. 
In this paper we remark that the standard methods based on expectation maximization or choice of the equilibrium value of $s^*$ are
sub-optimal. 
%In particular, the point-wise convergence of $s^*$ is not guaranteed, cf. Fig.~\ref{fig:1RSB_SE_lambdas_alls}.
%Nevertheless, we stress that for the low-rank matrix estimation problem in the RSB region at large $s$ the algorithm does not converge point-wise (in particular it does not at $s=1$, by definition)
%while at very low $s$ only the trivial solution $M=Q=0$ exists. 
%There is then only a finite interval in $s$ for which the algorithm converges point-wise to a nontrivial fixed-point 
%and in pratical situations the choice of $s$ within such interval may be not so crucial, cf. Fig.~\ref{fig:D00p84_M_Q_MSE_lambdaI}.

Finally, we mention that the derivation we have done here to obtain the ASP algorithm can be easily generalized along the same lines to an arbitrary number of replica symmetry breaking levels.
While to formally derive such algorithm is in principle
straightforward, it is not obvious what are the algorithmic
consequences in terms of estimation error and convergence in
comparison to AMP and ASP. This is also left  for future work.  
%As we mention above, the algorithmic
%performance are not immediatly related to thermodynamic properties of
%system that were studied in the physics literature.
%Hence, such comparison is nontrivial and it is left for future works.

%----------

%What happens in $\D \to 0$? numerically hard to study
%[Comparison between Bayes-optimal and MAP: Figure with the MMSE compared the AMP and
%${\rm SP-AMP}(s_{\rm opt})$ MSE for $\Delta \to 0$. ]

\section*{Acknowledgments}

This work is supported by "Investissements d’Avenir"
LabEx PALM (ANR-10-LABX-0039-PALM) (SaMURai
and StatPhysDisSys projects), by the ERC under the European Union’s FP7
Grant Agreement 307087-SPARCS and the European Union's Horizon 2020
Research and Innovation Program 714608-SMiLe, as well as by the French
Agence Nationale de la Recherche under grant ANR-17-CE23-0023-01
PAIL. We gratefully acknowledge the support of NVIDIA Corporation with
the donation of the Titan Xp GPU used for this research.

%%%%%%%%%%%%%%%%%%%%%%%%%%%%%%%%%%%%%%%%%%%%%%%%%%%%%%%%%%%%%%%%%%%%%%%%%%%%%%%%%%%%

\appendix

%\input{appendixA}

%%%%%%%%%%%%%%%%%%%%%%%%%%%%%%%%%%%%%%%%

\section{The 1RSB-SE: equivalence with the 1RSB replica calculation}
\label{sec:equivalence_SE_replica}
We want to show the 1RSB state evolution coincides with the replica result.
Here we will not verify that all the equations coincides, but we will limit ourselves to show that Eq.~(\ref{eq:M1}) for $M$ coincides with 
the equation for the magnetization of the replica result.
%the Eq.~(\ref{}) for $m$ being $m$ and $M$ the same order parameter, given their definitions.
In a 1RSB ansatz, the overlap matrix has the form of a block matrix in which the inner block of size $s$ has values $q_1$ while 
outside of the block the value is in principle different and indicated with $q_0$ \cite{MPV87}. The diagonal value of the overlap matrix is $q_d$.
The free energy of the system takes the form
\begin{align}
F_{\text{1RSB}} =&
\frac{1}{4 \D} \left\{   \frac{\D-\D_0}{\D} q_d^2 -2 M^2 
+ \frac{\D_0}{\D} \left[ s q_0^2 + (1-s) q_1^2 \right]
\right\}
 + \int \de h \, P_{\text{I}}(h) \, f_{\text{I}}(s,h)
\end{align}
where the functions $P_{\text{I}}(h)$ and $f_{\text{I}}(s,h)$ are equal to the expressions in Eqs.~(\ref{eq:f1h})-(\ref{eq:P1h}) provided that
\begin{align}
& \D^{(0)}=q_1-q_0 \, , &
& Q=q_0 \, , &
& \Delta^{(1)} = q_d-q_1 \, .
\end{align}
The value of the magnetization, namely the overlap of the ground truth with the inferred signal, that extremizes the free energy satisfies the following equation
\begin{align}
 M =&  \int dh \, \D \, \frac{\partial P_{\text{I}}(h)}{\partial M} \, f_{\text{I}}(s,h)
 =  \int dh \, \int D x^{(0)} \D \, \frac{\partial p_{\text{I}}(h,x^{(0)})}{\partial M} \, f_{\text{I}}(s,h) =
\\
 =& \int dh \, \int D x^{(0)} \, x^{(0)} \, \frac{\partial p_{\text{I}}(h,x^{(0)})}{\partial h} \, f_{\text{I}}(s,h)
 = - \int dh \, \int D x^{(0)} \, x^{(0)} \,  p_{\text{I}}(h,x^{(0)}) \, f'_{\text{I}}(s,h)
\end{align}
where $p_{\text{I}}(h,x^{(0)})$ is defined such that $P_{\text{I}}(h) = {\mathbb{E}}_{x^{(0)}} \left[ p_{\text{I}}(h,x^{(0)}) \right] $
and $D x^{(0)} \equiv d x^{(0)} P_0(x^{(0)})$.
Now rescaling the integration variable $h$ as
\begin{align}
 h = - \frac{\sqrt{\D_0 Q}}{\D} W - \frac{M}{\D} x^{(0)}
\end{align}
where $W$ is a standard Gaussian random variable, we obtain the 1RSB-SE Eq.~(\ref{eq:1RSB_SE_M}).
The equations for $q_0$, $q_d$ and $q_1$ can be obtained on the same lines.

\section{Derivative of MSE wrt $s$}
\label{sec:derivative_of_MSE}
In replica theory notation, cf. Appendix \ref{sec:equivalence_SE_replica}, the equations for $M$ and $Q$ can be written as
\begin{align}
M &= \int \de h\, \D  \left[\frac{\partial }{\partial M} P_{\text{I}}(h)\right]  f_{\text{I}}(s,h)\, ,
 &
Q &= \int \de h P_{\text{I}}(h)\left(f_{\text{I}}'(s,h)\right)^2\, ,
\end{align}
where the functions $P_{\text{I}}(h)$ and $f_{\text{I}}(s,h)$ are given in Eqs.~(\ref{eq:f1h})-(\ref{eq:P1h}).
%\begin{align}
% f(1,h) &= \log \, {\mathbb{E}}_{x} \left[ \exp \left( -
%          \frac{V^{(1)}}{2} x^2 + h x \right) \right] \, ,
% \\
% f(s,h) &= \frac{1}{s} \log {\mathbb{E}}_{W} \left[  e^{s
%          f(1,h-\sqrt{V^{(0)}} W)} \right] \, ,
% \\
%% P(s,h) &= \frac{1}{\hat r (0) \sqrt{2 \pi (- \hat q_0)}} {\mathbb{E}}_{x^{(0)}} \left[ \exp \left[ - \hat q_d^{(0)} \left( x^{(0)} \right)^2 + \frac{B^2}{2 \hat q_0} \right] \right]
%% P \left(s, \sqrt{\D_0 Q} W / \D \right) &= 
% P \left(s, h \right) &=
% \frac{\D}{ \sqrt{2 \pi \D_0 Q}} {\mathbb{E}}_{x^{(0)}} \left[ \exp
%                        \left( -  \frac{\D^2}{2 \D_0 Q} T^2 \right)
%                        \right] \, . 
%\end{align}
%with, as usual,
%\begin{align}
%& V^{(1)} =  \frac{1}{\D} \left( Q+\D^{(1)}+\D^{(0)} \right) - \frac{\D_0}{\D^2} \D^{(1)}
% \, , &
%& T =   \frac{M}{\D} x^{(0)} + h % \frac{\sqrt{\D_0 Q}}{\D} W
% \, , &
%& V^{(0)} = \frac{\D_0}{\D^2} \D^{(0)}\, .
%\end{align}
So that we have
\begin{align}
M-Q =  \int dh \,  f_{\text{I}}(s,h) 
 \left[  \D \frac{\partial P_{\text{I}}(h)}{\partial M} + \frac{\partial}{\partial h} \left( P_{\text{I}}(h) \, f_{\text{I}}'(s,h) \right)  \right]
 \equiv \int dh \, f_{\text{I}}(s,h) \, D(s,h)\, ,
\end{align}
Therefore, the derivative of the MSE wrt $s$ is given by
\begin{align}
 \frac{\partial \text{MSE}}{\partial s} &= M-Q - 
  {\mathbb{E}}_{x^{(0)},W}
  \left[ \left( x^{(0)} - \frac{1}{s} \frac{\partial f_{in}}{\partial
  T} \right) \, \frac{\partial }{\partial s}  \frac{\partial
  f_{in}}{\partial T}  \right] = \\
&= 
 - 2 \int dh \, \frac{\partial f_{\text{I}}(s,h)}{\partial s} 
 \left[ \D \frac{\partial P_{\text{I}}(h)}{\partial M} + \frac{\partial}{\partial h} \left( P_{\text{I}}(h) \, f_{\text{I}}'(s,h) \right)  \right] =
 \\
 &= -  2 \int dh \, \frac{\partial f_{\text{I}}(s,h)}{\partial s}  \, D(s,h) \, .
\end{align}
Replacing the explicit expression of $P_{\text{I}}(h)$ in the previous definition of $D(s,h)$, we obtain Eq.~(\ref{eq:D_for_MSE}).

Alternately, defining $p_{\text{I}}(h,x^{(0)})$  such that $P_{\text{I}}(h) = {\mathbb{E}}_{x^{(0)}} \left[ p_{\text{I}}(h,x^{(0)}) \right] $, we can express
the two quantities as
\begin{align}
  M-Q &= - \int dh \, f'_{\text{I}}(s,h) \, \int D x^{(0)} \, p_{\text{I}}(h,x^{(0)}) \, \left[ x^{(0)} + f'_{\text{I}}(s,h) \right] \, ,
 \\
 \frac{\partial \text{MSE}}{\partial s} &=  
 - 2 \int dh \, \frac{\partial f'_{\text{I}}(s,h) }{\partial s} \, \int D x^{(0)} \, p_{\text{I}}(h,x^{(0)}) \, \left[ x^{(0)} + f'_{\text{I}}(s,h) \right] \, .
\end{align}
where $D x^{(0)} \equiv d x^{(0)} P_0(x^{(0)})$.
The condition $D(s,h) =0$, that implies $M=Q$ \textit{and} the extremization of the MSE, is then also equivalent to
\begin{align}
& \int D x^{(0)} \, p_{\text{I}}(h,x^{(0)}) \,  x^{(0)}   = - P_{\text{I}}(h) \, f'_{\text{I}}(s,h)
& \rightarrow &
& f'_{\text{I}}(s,h) = - \langle x^{(0)} \rangle_{p_I}
\end{align}
where the average is over the distribution $p_{\text{I}}(h,x^{(0)}) / P_{\text{I}}(h)$ and the condition must hold for any $h$.

%%%%%%%%%%%%%%%%%%%%%%%%%%%%%%%%%%%%%%%%%%%%%%%%%%%%%%%%%%%%%%%%%%%%%%%%%%%%%%%%%%%%%%%%%%%%%%%%%%%%

%\section*{References}

%\bibliographystyle{ieeetr}
%\bibliography{refs_LowRank}

\end{document}